\pdfoutput=1
\documentclass[11pt]{article}

\usepackage[letterpaper,margin=1in]{geometry}

\usepackage[utf8]{inputenc} \usepackage[T1]{fontenc}    \usepackage{url}            \usepackage{booktabs}       \usepackage{amsfonts}       \usepackage{nicefrac}       \usepackage{microtype}      \usepackage{xcolor}         \usepackage{lipsum}
\usepackage{subfig}
\usepackage[shortlabels]{enumitem}
\usepackage{xspace}

\usepackage{tikz}
\usetikzlibrary{arrows.meta}
\usepackage{tabularx}
\usepackage{booktabs}  %
\usepackage{array}     %

\usepackage{amsthm,amsmath,amssymb,amsfonts,amssymb,mathtools}
\usepackage{xcolor}
\usepackage{empheq}
\usepackage{dsfont}
\usepackage{mathrsfs}
\usepackage{graphicx}
\usepackage{enumitem}
\usepackage{subfig}
\usepackage{comment}
\usepackage{framed}
\usepackage{float}
\usepackage{algorithm2e}
\usepackage{needspace}
\usepackage[noend]{algpseudocode}
\usepackage{thm-restate}
\usetikzlibrary{backgrounds}
\pgfdeclarelayer{background}
\pgfdeclarelayer{foreground}
\pgfsetlayers{background,main,foreground}

\SetAlCapFnt{\normalfont\small}\SetAlCapNameFnt{\unskip\itshape\small}

\definecolor[named]{ACMBlue}{cmyk}{1,0.1,0,0.1}
\definecolor[named]{ACMYellow}{cmyk}{0,0.16,1,0}
\definecolor[named]{ACMOrange}{cmyk}{0,0.42,1,0.01}
\definecolor[named]{ACMRed}{cmyk}{0,0.90,0.86,0}
\definecolor[named]{ACMLightBlue}{cmyk}{0.49,0.01,0,0}
\definecolor[named]{ACMGreen}{cmyk}{0.20,0,1,0.19}
\definecolor[named]{ACMPurple}{cmyk}{0.55,1,0,0.15}
\definecolor[named]{ACMDarkBlue}{cmyk}{1,0.58,0,0.21}

\usepackage[colorlinks,citecolor=blue,linkcolor=magenta,bookmarks=true]{hyperref}
\usepackage[nameinlink]{cleveref}
\crefname{ineq}{Inequality}{Inequality}
\creflabelformat{ineq}{#2{\upshape(#1)}#3}
\crefname{sub}{Subsection}{Subsection}
\creflabelformat{Subsection}{#2{\upshape(#1)}#3}
\crefname{sdp}{SDP}{SDP}
\creflabelformat{sdp}{#2{\upshape(#1)}#3}
\crefname{lp}{LP}{LP}
\creflabelformat{lp}{#2{\upshape(#1)}#3}

\newtheorem{theorem}{Theorem}[section]
\newtheorem{lemma}[theorem]{Lemma}
\newtheorem{proposition}[theorem]{Proposition}
\newtheorem{corollary}[theorem]{Corollary}
\newtheorem{claim}[theorem]{Claim}

\newtheorem{definition}[theorem]{Definition}

\newtheorem{fact}[theorem]{Fact}

\theoremstyle{definition}

\newtheorem{remark}[theorem]{Remark}

\crefname{claim}{claim}{claims}

\renewcommand\vec[1]{\mathbf{#1}}
\DeclareMathOperator*{\pr}{\mathbf{Pr}}
\renewcommand{\Pr}{\mathbf{Pr}}
\DeclareMathOperator*{\E}{\mathbf{E}}

\newcommand{\normal}{\mathcal{N}}

\newcommand{\bx}{\mathbf{x}}
\newcommand{\by}{\mathbf{y}}
\newcommand{\bv}{\mathbf{v}}
\newcommand{\bu}{\mathbf{u}}
\newcommand{\bz}{\mathbf{z}}
\newcommand{\bw}{\mathbf{w}}

\newcommand{\err}{\mathrm{err}}

\newcommand{\op}{\mathrm{op}}

\newcommand{\p}{\mathbf{P}}

\newcommand{\R}{\mathbb{R}}

\newcommand{\Z}{\mathbb{Z}}
\newcommand{\N}{\mathbb{N}}

\newcommand{\eps}{\epsilon}

\newcommand{\poly}{\mathrm{poly}}
\newcommand{\polylog}{\mathrm{polylog}}

\newcommand{\sgn}{\mathrm{sign}}
\newcommand{\sign}{\mathrm{sign}}
\newcommand{\calN}{{\cal N}}

\newcommand{\opt}{\mathrm{OPT}}
\newcommand{\D}{\mathcal{D}}

\newcommand{\Ind}{\mathds{1}}

\newcommand{\littlesum}{\mathop{\textstyle \sum}}

\newcommand{\wt}{\widetilde}

\newcommand{\x}{\vec x}
\newcommand{\z}{\vec z}
\newcommand{\w}{\vec w}

\newcommand{\hide}[1]{}

\newcommand{\cH}{\mathcal{H}}

\newcommand{\cN}{\mathcal{N}}

\newcommand{\cS}{\mathcal{S}}

\newcommand{\eqdef}{\coloneqq}

\newcommand{\abs}[1]{\lvert#1\rvert}

\usepackage{multirow}

\def\colorful{0}

\ifnum\colorful=1

\else

\fi

\allowdisplaybreaks

\title{Testable Learning of General Halfspaces under Massart Noise}

\author{
Ilias Diakonikolas\thanks{Supported by NSF Medium Award CCF-2107079, 
ONR award number N00014-25-1-2268, 
and an H.I. Romnes Faculty Fellowship.}\\
UW Madison\\
{\tt ilias@cs.wisc.edu}\\
\and
Giannis Iakovidis\thanks{Supported by ONR award number N00014-25-1-2268.}\\
UW Madison\\
{\tt iakovidis@wisc.edu}\\
\and
Daniel M. Kane\thanks{Supported by NSF Medium Award CCF-2107547.}\\
UC San Diego\\
{\tt dakane@ucsd.edu }
 \and
Sihan Liu\\
UC San Diego\\
{\tt sil046@ucsd.edu}\\
}

\date{}

\begin{document}

\maketitle

\begin{abstract}%
We study the algorithmic task of testably learning {\em general} Massart halfspaces under the Gaussian distribution. In the testable learning setting, the aim is the design of a tester-learner pair satisfying the following properties: (1) if the tester accepts, the learner outputs a hypothesis and a certificate that it achieves near-optimal error, and (2) it is highly unlikely that the tester rejects if the data satisfies the underlying assumptions. Our main result is the first testable learning algorithm for general halfspaces with Massart noise and Gaussian marginals. The complexity of our algorithm is $d^{\mathrm{polylog}(\min\{1/\gamma, 1/\epsilon \})}$, where $\epsilon$ is the excess error and $\gamma$ is the bias of the target halfspace, which qualitatively matches the known quasi-polynomial Statistical Query lower bound for the non-testable setting. The analysis of our algorithm hinges on a novel sandwiching polynomial approximation to the sign function with {\em multiplicative error} that may be of broader interest. 
\end{abstract}

\setcounter{page}{0}

\thispagestyle{empty}

\newpage

\section{Introduction} \label{sec:intro}

This work focuses on the distribution-specific 
learning of (general) halfspaces with Massart noise~\cite{Massart2006} 
in the testable framework introduced in~\cite{RV23}. Before we state our results, we provide the necessary background and motivation. 

\paragraph{Halfspaces and Their Efficient Learnability}

A {\em halfspace} is any Boolean-valued function 
$f: \R^d \to \{\pm 1\}$ of the form $f(\x) = \sign(\vec w^\ast \cdot \x - t^\ast)$, 
for a weight vector $\vec w^\ast \in \R^d$ and a threshold 
$t^\ast \in \R$. The algorithmic task of learning halfspaces from labeled
examples is one of the most basic and extensively studied 
problems in machine learning~\cite{Rosenblatt:58, Novikoff:62,
MinskyPapert:68,  FreundSchapire:97, Vapnik:98, CristianiniShaweTaylor:00}. 
While halfspaces are efficiently PAC learnable in the distribution-free PAC 
model, in the realizable (aka noise-free) setting~\cite{Valiant:84} (see, 
e.g.,~\cite{MT:94}), the algorithmic problem becomes computationally intractable in the 
presence of label noise, both in the adversarial~\cite{Daniely16,Tiegel23} and in semi-random 
label noise models~\cite{DK22-SQ-Massart, NasserT22, DKMR22}.

To circumvent the aforementioned computational limitations, a long line of work has 
developed efficient noise-tolerant learning algorithms for halfspaces in 
challenging noise models under natural distributional assumptions. These include 
both the adversarial label noise setting~\cite{KKMS:08,KLS09,ABL17,YanZ17, DKS18a,DKTZ20c,DKTZ22} 
and semi-random settings (namely, the Massart and Tsybakov 
noise models)~\cite{AwasthiBHU15,ZhangLC17,YanZ17, DKTZ20,ZhangSA20,DKTZ20b, DKKTZ21Tsybakov, DKKTZ22}.

\paragraph{Testable Learning}
A fundamental limitation of the aforementioned noise-tolerant learners 
is that they provide no guarantees if the underlying distributional 
assumptions---on the joint distribution, incorporating the 
marginal on the examples and the label noise itself---are not satisfied. 
This drawback has motivated the definition of the {\em testable learning framework}~\cite{RV23} that, informally speaking, aims to 
``test'' the underlying distributional assumptions so as to 
provide useful guarantees whenever 
the algorithm outputs a candidate hypothesis. In more detail, 
the testable learning framework 
involves the design of a tester-learner pair satisfying the following properties: 
(1) if the tester accepts, the learner outputs a hypothesis and a certificate that it achieves near-optimal error, 
and (2) it is highly unlikely that the tester rejects if the data satisfies the underlying 
assumptions. 

Since the initial work of~\cite{RV23}, there has been a flurry of research activity 
on the complexity of testable learning (for halfspaces and other concept classes) 
in the agnostic model; 
see, e.g.,~\cite{GKK23,DKKL23,GKSV23,GKSV24,DKLZ24,slot2024testably, KSV24, GSSV24,GKSV25} and references therein. To facilitate the subsequent discussion, 
we provide the definition of the testable learning 
framework from~\cite{GKSV25}, which is the natural generalization 
of the original definition~\cite{RV23} to the semi-random noise 
setting we consider.

\begin{definition}[Testable Learning, see~\cite{RV23, GKSV25}]
\label{def:testable-learning}
Let $\mathcal{H} \subseteq \{\,\R^d \to \{\pm 1\}\,\}$ be a concept class, 
$m : (0,1)\times(0,1)\to \Z_+$
and $\mathcal{D}_{\x, y}$ be a family of distributions over $\R^d\times \{\pm1\}$.
We say that an algorithm $\mathcal{A}$ is a {\em testable learner for $\mathcal{H}$ 
with respect to the distributional assumptions induced by $\mathcal{D}_{\x, y}$} 
if it satisfies the following condition.  
On input $\epsilon,\delta\in(0,1)$ and a dataset $S$ of $m = m(\eps,\delta)$ 
i.i.d.\ samples from a distribution $D_{\x, y}$ over $\R^d\times \{\pm 1\}$, 
$\mathcal{A}$ either outputs \emph{Reject} or \emph{(Accept, $h$)} 
for some hypothesis $h:\mathbb{R}^d\to\{\pm 1\}$, satisfying the following:
\begin{enumerate}[leftmargin=*, nosep] 
    \item (Soundness).
    The probability that $\mathcal{A}$ accepts and outputs an $h$ for which  
    $\Pr_{(\x,y)\sim D}\!\big[\,y \neq h(\x)\,\big] >\opt + \eps$,
    where $\opt \eqdef \min_{f\in \mathcal{H}} \Pr_{(\x,y)\sim D_{\x, y}}\!\big[\,y \neq f(\x)\,\big]$, is at most $\delta$.
    \item (Completeness). If $D_{\x, y} \in \mathcal{D}_{\x, y}$, then 
    $\mathcal{A}$ accepts with probability at least $1-\delta$.
\end{enumerate}
\end{definition}

As was pointed out in~\cite{RV23}, 
starting with an algorithm that achieves $\delta=1/3$, 
the success probability can be amplified 
to $1-\delta$ (by standard repetition) 
at the cost of an $O(\log(1/\delta))$ factor increase
in the sample complexity. Hence, throughout this work, we will focus on developing efficient 
testable learners that achieve $\delta = 1/3$.
The distribution class $\mathcal{D}_{\x, y}$ defines the distributional assumptions under which the testable learner should accept. We note that the completeness condition in Definition~\ref{def:testable-learning} 
concerns the joint distribution on $(\x, y) \in \R^d\times \{\pm 1\}$---not just 
the marginal distribution on $\x \in \R^d$ 
(as was the case in the original definition~\cite{RV23}). 
This generalization is required when we are interested 
in also testing properties of the label noise model itself. 

With this motivation, \cite{GKSV25} studied 
testable learning (as per Definition~\ref{def:testable-learning}) 
for {\em homogeneous} halfspaces with Massart noise under the Gaussian distribution.
Concretely, this corresponds to Definition~\ref{def:testable-learning} when 
(i) the concept class $\mathcal{H}$ is the family of all {\em homogeneous} 
halfspaces on $\R^d$ (i.e., halfspaces with zero threshold); and 
(ii) $\mathcal{D}_{\x, y}$ is the family of distributions 
over $(\x, y) \in \R^d\times \{\pm 1\}$ whose marginal on $\x$ is the 
standard Gaussian on $\R^d$, and the distribution of $y \mid \x$ is generated 
by adding Massart noise~\cite{Massart2006} defined as follows.

\begin{definition}[Massart Noise] \label{def:Massart}
Let $f: \R^d \to \{\pm1\}$, $D_{\x}$ be a distribution on $\R^d$, and 
$\eta(\cdot): \R^d \to[0,\eta]$ be a noise function for some 
parameter $\eta<1/2$. 
An $\eta$-Massart noise oracle is a distribution on $(\x, y) \in \R^d \times\{\pm1\}$ 
such that $\x\sim D_{\x}$ and the label $y$ satisfies: 
with probability $1-\eta(\x)$, $y=f(\x)$; and $y=-f(\x)$ otherwise. 
\end{definition}

The main result of~\cite{GKSV25} is a testable learner for homogeneous 
halfspaces in this setting with sample and computational complexity 
$\poly(d, 1/\eps)$, where the noise rate upper bound $\eta = 1/2-\beta$ 
and $\beta>0$ is a universal constant\footnote{While we suppress the dependence on $\beta$ in this discussion, we note that the 
complexity of their algorithm is $(d/\eps)^{\poly(1/\beta)}$. 
In Appendix~\ref{app:sq-lb}, we prove an SQ lower bound giving evidence 
that such a dependence is necessary for efficient algorithms.}.

In this work, we continue this line of investigation. Our goal is to understand
the complexity of learning {\em general} (i.e., not necessarily homogeneous) halfspaces in this framework. 
While at first glance it might seem that the distinction between the 
homogeneous and the general cases is innocuous, it is well-known that there 
can be a substantial gap in the computational complexity of distribution-specific learning in the two cases without the testable requirement. 
In particular, while there exists a $\poly(d/\eps)$ time learner 
for Massart homogeneous halfspaces under the Gaussian 
distribution~\cite{DKTZ20}, the complexity of learning general halfspaces 
in the same setting is {\em quasi-polynomial}, namely 
$d^{\Theta(\log(1/\eps))}$~\cite{DKKTZ22}. Specifically, \cite{DKKTZ22} gave 
an algorithm with complexity $d^{O(\log(1/\eps))}$ and a matching 
SQ lower bound of $d^{\Omega(\log(1/\eps))}$. 
Since testable learning is by definition harder than its non-testable counterpart, this SQ lower bound is 
hence inherited in our scenario. 
On the other hand, no non-trivial upper bound is known 
for general halfspaces in the testable setting (even under the earlier model of \cite{RV23} without the testable requirement on the noise model).
\footnote{
One can of 
course directly use known upper bounds of $d^{O(1/\eps^2)}$ for the agnostic setting~\cite{RV23, GKK23}.}

Motivated by this gap in our understanding, we aim to characterize the complexity of testable learning for general Massart halfspaces. 
Specifically, we study the following question:
\vspace{-0.1cm}
\begin{center}
{\em What is the complexity of testable learning of {\em general} halfspaces\\ 
with Massart noise under the Gaussian distribution?  }
\end{center}
\vspace{-0.1cm}
As our main result, we essentially resolve this question by developing a 
testable learner with complexity $d^{\polylog(1/\eps)}$---qualitatively 
matching the guarantees for the non-testable setting. In more detail, the 
complexity of our algorithm is $d^{\polylog(\min\{1/\gamma, 1/\eps\})}\poly(1/\eps)$, 
where $\gamma$ is the ``bias'' of the target halfspace (see 
Definition~\ref{def:bias}). Since homogeneous halfspaces correspond to 
$\gamma=1/2$, our algorithmic result can be viewed as a generalization of the 
upper bound in~\cite{GKSV25}. A detailed description of our results is given in the 
following subsection. 

\subsection{Our Results} \label{ssec:results}

It turns out that the complexity of testably learning Massart halfspaces depends on the ``bias'' of the target halfspace. Hence, 
to state our main algorithmic result, we require the following definition.

\begin{definition}[$\gamma$-Biased halfspaces] \label{def:bias}
For $d \in \mathds{Z}_+$ and
$\gamma\in(0,1/2]$, define the hypothesis class
\[
\mathcal{H}_{d,\gamma}
\eqdef
\Big\{
f(\x)=\sign(\bv\cdot \x-t), \x,\vec v \in \R^d, t \in \R : \|\bv\|_2=1,
\min_{i=\pm 1} \Pr_{\cN^d}[f(\x)=i] \ge \gamma
\Big\}.
\]
We refer to $\mathcal{H}_{d,\gamma}$ as the class of $\gamma$-Biased halfspaces.
\end{definition}

\noindent Our main result is the following: 

\vspace{-0.2cm}
\begin{theorem}[Testably Learning $\gamma$-Biased Massart Halfspaces]\label{thm:testable-learning-gamma}
Fix parameters $\eta\in[0,1/2), \beta \eqdef 1-2\eta$ and  $\gamma\in(0,1/2]$.
Let $\mathcal{D}_\gamma$ be the class of distributions 
over $\R^d\times\{\pm1\}$
whose $\x$-marginal is the standard Gaussian $\cN^d$ 
and satisfy the $\eta$-Massart noise condition
with respect to some halfspace in $\mathcal{H}_{d,\gamma}$.
There exists an algorithm that, given $\eps,\delta, \eta$ and $\gamma$, uses

$$
N = d^{\tilde{O}(\beta^{-2}) \polylog(\min\{1/\eps, 1/\gamma \})} \, \poly(1/\eps)\log(1/\delta)
$$
samples, runs in $\poly(N,d)$ time,
and testably learns the class $\cH_{d,\gamma}$ with respect to $\mathcal D_{\gamma}$.
\end{theorem}

A few comments are in order regarding the statement of \Cref{thm:testable-learning-gamma}.
First, we note that for the special case of homogeneous halfspaces---or more broadly for halfspaces
where $\gamma$ is a positive universal constant---our algorithm has complexity 
$d^{\tilde{O}(\beta^{-2})}$. For the class of general halfspaces, its complexity is 
quasi-polynomial, namely $d^{\tilde{O}(\beta^{-2}) \polylog(1/\eps)}$.

As shown in~\cite{DKKTZ22}, a $d^{\log(\min\{1/\eps, 1/\gamma \})}$ complexity dependence 
is required for Statistical Query (SQ) algorithms, even in the non-testable setting.
This gives evidence that the quasi-polynomial dependence in our upper bound is inherent. 
Moreover, we show in Appendix~\ref{app:sq-lb} that any SQ algorithm 
for the testable setting requires complexity 
$d^{\tilde{\Omega}(\beta^{-2})}$ even for near-homogeneous halfspaces. 
This gives formal evidence 
that the exponential dependence on $1/\beta^2$ may be necessary for efficient algorithms, and additionally provides a separation between 
the testable and non-testable versions of the problem (as a function of $\beta$).

A key technical ingredient required for the analysis of our algorithm is the following new result
on sandwiching polynomial approximations for the sign function under the Gaussian distribution, 
which may be of broader interest. 

\begin{restatable}[Multiplicative Sandwiching Polynomial Approximation to the Sign Function]
{theorem}{MultSandwichSign}
\label{lem:structural}
Let $t\in \R$, and
\(
h(x)\;\eqdef\;\Ind(x\ge t).
\)
Fix an accuracy parameter $\alpha\in(0,1/2)$.
There exist polynomials $p_-,p_+:\R\to\R$ of degree at most
$\deg(p_\pm)=O({(|t|+1)^6\log^2(1/\alpha)}/{\alpha^2})$
 such that:
\begin{enumerate}
    \item $p_-(x)\le h(x)\le p_+(x)$ for all $x\in\R$
    \item 
    \(
    \E_{x\sim \cN(0,1)}\big[p_+(x)-p_-(x)\big] \le \alpha\E_{x\sim \cN(0,1)}[h(x)].
    \)
\end{enumerate}
\end{restatable}
Within the TCS community, the study of sandwiching polynomial approximations has been a crucial component in the context of 
pseudorandomness and testable learning \cite{DGJ+:10,kane2010k,GKK23,slot2024testably}. 
The novelty of
\Cref{lem:structural} is that it obtains a multiplicative approximation 
with near-optimal degree, as opposed to the additive approximations achieved by previous work.

Within the approximation theory literature, the problem (without the sandwiching constraint) has received substantial attention 
 under the banner of weighted polynomial approximation theory \cite{ditzian1997jackson,freud1977markov,kroo1995weighted,levin1990l}. Particularly, results of such flavor are known as Jackson-type theorems for Freud weights of the form $\exp(-Q(x))$, where $Q$ is some polynomial. 
For the specific setting of $L^1$ polynomial approximation of $\Ind(x \geq t)$ under the Gaussian weight, the error bound is shown to be $\ell^{-1/2} ( \exp(-t^2/2) + \exp(-\ell/2) )$ \cite[Theorem 1.2]{ditzian1997jackson}, suggesting that a degree of $\ell = \Theta(t^2)$ suffices for multiplicative approximation without the sandwiching requirement. 

It is plausible conjecture that $O(t^2)$ remains the optimal 
degree for sandwiching polynomials as well. Such an improvement would allow us to reduce the parameter $\ell$ used in \Cref{alg:testable-massart-gamma} to $\Theta_\beta(\log(1/\gamma))$ and subsequently bring the sample complexity of \Cref{thm:testable-learning-gamma} down to \\
$d^{{O}_\beta(\log(\min\{1/\epsilon, 1/\gamma\}))}$, effectively matching the complexity of the non-testable learning algorithm 
(up to the dependence on $\beta$).

\paragraph{Implication for Non-testable Learning of Massart Halfspaces}
Note that the non-testable learner of \cite{DKKTZ22} 
requires the bias parameter $\gamma$ to be given as input.
In particular, its runtime increases to $d^{ \polylog(1/\eps) }$ 
when $\gamma$ is not known apriori. 
The tester underlying \Cref{thm:testable-learning-gamma} 
has the following interesting implication for the non-testable setting:
it can be used together with any learner for Massart halfspaces 
(satisfying the guarantee of \Cref{fact:general-massart}) 
to obtain a ``bias agnostic'' learner with runtime $d^{\polylog(1/\gamma)}$
when the bias $\gamma$ of the optimal halfspace is unknown.
The formal statement and analysis is deferred to \Cref{app:unknowngamma}.


\subsection{Technical Overview} 
\label{ssec:techniques}


Our starting point is the prior work of~\cite{DKKTZ22} 
for learning general Massart halfspaces under the Gaussian distribution 
in the non-testable setting. They gave an algorithm for this task 
with error $\opt + \eps$ and complexity 
$
d^{\log(\min\{1/\gamma,\,1/\eps\})}\,\poly(1/\eps)
$, 
where $\gamma$ is the bias of the target halfspace.
Our testable algorithm begins by using this procedure 
as a subroutine (with appropriate parameters)
to obtain a candidate halfspace
$h(\x) = \sign(\w \cdot \x - t).$
The main remaining task is to efficiently \emph{certify the optimality} 
of $h$ under a possibly adversarial distribution $D$,
while ensuring that we do not reject when $D$ 
satisfies our distributional assumptions.
Equivalently, our goal is to show that there is no other $\gamma$-biased halfspace $f$ whose error under $D$
is smaller than that of $h$ by more than $\poly(\eps)$.

Essentially, we would like to perform a test that certifies that $\E[h(\x)y\Ind(f(\x)\neq h(\x))]$ is at least $(1-2\eta)\pr[f(\x)\neq h(\x)]$ minus a small error. 
This inequality holds under Massart noise when $h$ is an optimal 
or near-optimal classifier, since for an optimal classifier $h$ 
we have $\E[h(\x)y\mid \x] > 1-2\eta$ for every $\x$. 
However, directly implementing such a test is computationally intensive, 
as it requires a cover over all competing classifiers.
As is standard, we instead replace the indicator of the region of disagreement by low-degree polynomials. 
In particular, we aim to verify that $\E[h(\x)y p^2(\x)]$ 
is approximately at least $(1-2\eta)\E[p^2(\x)]$.
We call any test that verifies the above for low-degree polynomials $p$ 
a polynomial non-negativity test.
If the polynomial $p$ closely approximates the disagreement region, 
then this test approximately certifies our goal. 
The difficulty is that the region $\Ind(f(\x)\neq h(\x))$ 
is an intersection of halfspaces, which may require high-degree 
polynomials to approximate or lead to a complicated analysis.

To address this, we partition the space into sufficiently 
fine stripes orthogonal to $\w$ on which $h$ is constant, 
and we certify the above property conditioned on each stripe. 
Within a single stripe, the disagreement region is described 
by a single halfspace, making it easier to analyze. 
To combine the tests across all stripes, we additionally verify 
that the probability mass of each stripe matches 
that of the corresponding Gaussian stripe. 
Finally, to justify the use of the polynomial test within each stripe, 
we check that the first few moments of $\x$ in each stripe 
approximately match the corresponding Gaussian moments. 
Overall, our certification procedure consists of (1) a stripe mass test, 
(2) a moment-matching test, and (3) a polynomial non-negativity test.

Now let us analyze what these tests yield for a single stripe. 
For this intuitive explanation, we assume for simplicity that
the stripe is infinitely thin. Assuming that $h$ is positive on the stripe, 
we would like to show that for $g(\x)=\Ind(f(\x)=-1)$ the quantity 
$\E[g(\x)y]$ is positive. This can be achieved by constructing {\em sandwiching 
polynomials} for $g$. In particular, suppose that we can find low-degree 
polynomials $p_+$ and $p_-$ such that $p_+ \ge g \ge p_-$ 
and such that the expectations of $p_+$ and $p_-$ are close under the 
Gaussian distribution. Since $p_+ \ge 0$, it can be written 
as a sum of squares of polynomials, and therefore 
by our polynomial non-negativity test $\E[p_+(\x)y]$ is relatively large. 
On the other hand, by our moment-matching test, 
the quantity $\E[p_+(\x)-p_-(\x)]$ is small, 
which in turn bounds the difference between $\E[p_+(\x)y]$ and 
$\E[g(\x)y]$. Thus, as long as we can find such $p_+$ and $p_-$ 
with $\E[p_+(\x)-p_-(\x)]$ being a small multiple of $\E[p_+(\x)]$, 
or equivalently $\E[g(\x)]$, this approach suffices.

Moreover, we note that the halfspace induced by $f$ on each stripe 
is not necessarily $\gamma$-biased: its bias depends both 
on the position of the stripe and on the angle between 
the normal vectors defining $f$ and $h$.
However, by an error-accounting argument, we can show 
that only stripes whose bias is close to $\gamma$ require certification. 
Indeed, highly biased stripes contribute only a small probability mass 
to the region of disagreement under the Gaussian distribution. 
Furthermore, since we approximately match moments 
under the distribution $D$, we can show that this small contribution 
is preserved (see the proof of \Cref{lem:near-parallel}). This 
allows us to safely ignore such stripes in the certification process. 
For details, we refer the reader to Appendix~\ref{sec:soundness-app}.

Next we discuss our structural result for sandwiching polynomials (\Cref{lem:structural}). The optimal degree of such polynomials has 
been studied extensively in the context of fooling Polynomial Threshold Functions with bounded independence; see, e.g., \cite{DGJ+:10,diakonikolas2010bounded,kane2010k}. Moreover, 
such polynomials have been leveraged 
as important technical tools in the prior literature of testable 
learning~\cite{GKK23,KSV24,slot2024testably}. 
These results typically aim for small \emph{additive} error, i.e., the 
$L^1$ (or $L^2$) norm of $p_+ - p_-$ is at most $\eps$ for \emph{all} 
threshold functions. Importantly, achieving such an additive error bound 
is known to require polynomials of degree $\Theta(1/\eps^2)$ for linear 
threshold functions. In our setting, this would translate to degree 
$\Omega(1/\gamma^2)$ for $\gamma$-biased halfspaces, 
leading to a sample complexity of $d^{\Omega(1/\gamma^2)}$\footnote{For general halfspaces, this gives $d^{O(1/\eps^2)}$, which does not improve on the complexity of the agnostic setting.}. 
We circumvent this obstacle by constructing 
sandwiching polynomials with \emph{multiplicative} error guarantees. 
Specifically, for a given linear threshold function $g$,
we only require $\E[p_+ - p_-]$ to be bounded from above by $\alpha \E[g]$, where we eventually set $\alpha$ to be a moderately small 
universal constant.
We show that such sandwiching polynomials exist 
with degree $\poly(t/\alpha)$, where $t$ is the threshold defining $g$. 
In our application  of this structural result, 
it suffices to set $t \approx \sqrt{ \log(1/\gamma) }$, 
which leads to our quasi-polynomial complexity.

Towards showing \Cref{lem:structural}, the approach 
used by \cite{diakonikolas2010bounded,kane2010k} 
of first mollifying the threshold function $g$ into a smooth function  
and then applying Taylor expansion, runs into difficulties. 
In particular, even when the requirement is only to achieve 
$\alpha$-multiplicative approximation, the approximation error 
needs to be about $\exp(-t)$ for points close to the origin 
(since the Gaussian pdf is some constant around the origin) 
when $g$ has a large threshold $t$.
However, known mollifiers can only guarantee a sub-gaussian accuracy 
of $\exp( - a^{ 2-c } )$, where $c>0$ is some constant, 
for points at distance $a$ from the threshold.
This implies that the approximation error of the mollifier, 
even before any Taylor expansion, 
will be at least $\exp(-t^{2-c}) \gg \exp(-t)$.

Instead, we take a detour from the mollification based design paradigm, 
and make use of \emph{Chebyshev polynomials},
which are nearly optimal for producing high-degree polynomials 
that remain bounded over a large interval, 
to construct the approximation polynomials directly.
It is worth nothing that our approach bears some high-level 
similarity with the work of \cite{DGJ+:10}, 
which also makes use of Chebyshev's polynomial approximation theorem 
as a black-box to prove existence of additive sandwiching polynomials. 
Importantly, our construction has the advantage of being fully explicit, 
making it handy to tune various design parameters in order to fit precisely 
the needs of multiplicative sandwiching polynomials.

We now briefly sketch our construction.
Consider the $m$-th order Chebyshev polynomial $T_m$.
For odd $m$, the polynomial $T_m(x)/x$ is large near $0$, and decays like 
$1/|x|$ away from $0$ (at least within the interval $[-1,1]$), and is 
bounded by $|x|^m$ outside of that interval. By taking a sufficiently high 
power $k = \poly(t)$ 
of this function, we obtain a ``bump'' polynomial that (1) is large at $0$, 
(2) decays as $|x|^k$ elsewhere in $[-1,1]$, 
and (3) grows moderately as $|x|^{mk}$ outside $[-1,1]$. 
By appropriately scaling and translating these polynomials, 
we can place such a bump near the relevant threshold 
and ensure that it is small over a sufficiently wide range. 
By convolving this bump with an interval, we obtain an approximation to a 
threshold function, 
and with additional care we convert this approximation into sandwiching 
polynomials. 
A careful error analysis yields the desired guarantees; for details, see 
Appendix~\ref{sec:sandwiching-result}.


\subsection{Comparison to Prior Techniques}
\vspace{-0.2cm}
We briefly compare the techniques employed by the prior work of~\cite{GKSV25} to ours, 
and illustrate the main bottlenecks their approach would face in the general halfspace case.

Firstly, as it is common in the testable learning literature, 
both approaches contain standard procedures (slice mass test, moment matching test) 
to verify the Gaussianity of the data marginal distribution 
conditioned on different stripe regions along the direction of the learned halfspace. 
For simplicity, we will assume that the data marginal is exactly Gaussian in the rest of the discussion.

When it comes to test the noise assumption, while inspired by different technical ideas, 
the end goal of both works is to certify an upper bound on the disadvantage 
of the learned halfspace $h$ compared to the optimal halfspace $h^*$. 
For convenience, define $I(\x) = \Ind( h^*(\x) \neq h(\x) )$ 
as the indicator of the disagreement region between $h^*$ and $h$.
As such, the disadvantage of $h$ can be written as 
$\E[ I(\x)  \Ind\{ y \neq h(\x) \}]$, 
and we always have the upper bound 
$\E[ I(\x)  \Ind\{ y \neq h(\x) \}] \leq \E[  (\eta + \eps) I(\x) ]$ 
under the Massart noise assumption.

Following the recurring theme of testable learning, both works first 
certify the inequality for degree-$k$ non-negative polynomials, and then 
employ polynomial approximation results for LTFs to transfer the guarantee 
to the indicator function $I(x)$ up to some approximation error term.

Yet known polynomial approximation results only allow 
for an approximation error on the order 
of  $\min(1 / \sqrt{k}, 1/t^2)$ when the linear threshold function 
has threshold $t$ (see e.g., Corollary B.8 of \cite{GKSV25}).  
For general halfspaces, this is insufficient 
when the optimal halfspace has a uniform
bias of $t = \sqrt{\log(1/\gamma)}$ 
across all stripes.\footnote{Note that this is impossible if the optimal halfspace is restricted to be homogeneous, 
hence clearing the obstacles for their analysis of homogeneous halfspaces.} 
Indeed, even if we take the degree $k$ to be $\polylog(1/\gamma)$, 
the approximation error will be on the order 
of $1 / \polylog(1/\gamma)$---far greater than the desired bound 
of $\E[  (\eta + \eps) I(x) ]$, 
which is on the order of $\poly(\gamma)$ under a Gaussian marginal.

Fortunately, our new multiplicative polynomial approximation result 
allows us to circumvent this issue. In particular, 
our approximation error is on the order of $C  \E[ I(\x)  ]$ 
for some constant $C$ of our choice. By taking $C$ to be 
a small constant multiple of $(\eta + \eps)$, 
we can still guarantee that 
$\E[  I(\x)  \Ind\{ y \neq h(\x) \}] \leq (1 + o(1)) \E[  (\eta + \eps) I(\x) ]$, 
which turns out to be sufficient for completing the soundness analysis.


\section{Preliminaries} \label{sec:prelims}
%
For $n \in \Z_+$, let $[n] := \{1, \ldots, n\}$.
We use small boldface characters for vectors
and capital bold characters for matrices.  
For $\bx \in \R^d$ and $i \in [d]$, $\bx_i$ denotes the
$i$-th coordinate of $\bx$, and $\|\bx\|_2 := (\littlesum_{i=1}^d \bx_i^2)^{1/2}$ denotes the
$\ell_2$-norm of $\bx$. {Throughout this text, we will often omit the subscript and simply write $\|\x\|$ for the $\ell_2$-norm of $\bx$.}
We will use $\bx \cdot \by $ for the inner product of $\bx, \by \in \R^d$
and $ \theta(\bx, \by)$ for the angle between $\bx, \by$. 
For vectors $\x, \vec v\in \R^d$ we denote by
$\x^{\vec v}$ the projection of $\x$ onto the line spanned by $\vec v$ and by $\x^{\perp \vec  v}$ we denote the projection of $\x$ to the orthogonal complement of $\vec v$. 
We also denote by $\vec v^{\perp}$ the orthogonal complement of $\vec v$.
We slightly abuse notation and denote {by}
$\vec e_i$ the $i$-th standard basis vector in $\R^d$.  
For a matrix $\vec M\in\R^{n\times m}$, we denote by $\|\vec M\|_2,\|\vec M\|_F$ to be the operator norm and Frobenius norm respectively. 

We use the standard asymptotic notation, where 
$\wt{O}(\cdot)$ is used to omit polylogarithmic factors.
{Furthermore, we use $a\lesssim b$ to denote that there exists an absolute universal constant $C>0$ (independent of the variables or parameters on which $a$ and $b$ depend) such that $a\le Cb$, $\gtrsim$ is defined similarly.
We use the notation $g(t)\le \poly(t)$ for a quantity $t\ge1$ to indicate   that there exists constants $c,C>0$ such that $g(t)\le Ct^c$. Similarly we use $g(t)\ge \poly(t)$ for a quantity $t<1$ to denote  that there exists constants $c,C>0$ such that $g(t)\ge Ct^c$.
}We refer the reader to Appendix~\ref{sec:omitted-facts} for additional preliminaries.


\section{Algorithm Description and Analysis} \label{sec:alg}
In this section, we present our testable learning algorithm (see Algorithm~\ref{alg:testable-massart-gamma} for the detailed pseudocode). 
As described in \Cref{ssec:techniques}, our algorithm first runs 
a proper learner to obtain a candidate halfspace, and then performs a sequence of tests 
to verify that the accuracy of the resulting halfspace is near-optimal.

\begin{algorithm}[H]
    \centering
    \fbox{\parbox{6in}{
        {\bf Input:}
        Accuracy $\eps\in(0,1)$, 
        constant $\eta\in (0,1/2)$, bias parameter $\gamma\in(0,1/2]$,
        and sample access to a distribution $D$ over $\mathbb{R}^d\times\{\pm1\}$.\\
        {\bf Output:}
        Either \emph{Accept} or \emph{Reject}
        and a hypothesis $h$, such that the probability of \emph{Accept} and 
        $\Pr_{(\x,y)\sim D}[h(\x)\neq y] > \opt_\gamma+\eps$, where $\opt_{\gamma}\eqdef \min_{f\in \cH_{d,\gamma}} \Pr_{(\x,y)\sim D}[f(\x)\neq y]$ is at most $1/3$.

\medskip

        \begin{enumerate}[leftmargin=*]
            \item\label{line:init-inputs}
            Set $\gamma \gets \max(\gamma,\eps)$, $\beta \gets 1-2\eta$, $\gamma\gets \min(\gamma,\beta) $, $\eps\gets \min(\eps,\beta/2)$.
            
            \item \label{line:init}
             Set  $l\gets C\log^3(1/\gamma)\log^2(1/\beta)/\beta^2$, $\eps'\gets \eps^C/3^{l}$, $N\gets d^{Cl}/\eps^{2C}$, $\Delta\gets \eps^{2}$, $\tau_p\gets\eps^{C}/d^l$, and $\tau_m\gets\eps^C/d^l$
            for a sufficiently large universal constant $C\in \Z_+$.

            \item \label{line:run-prev-alg}
            Run the algorithm of \cite{DKKTZ22} (see \Cref{fact:general-massart}) with parameters $\eps',\delta, \gamma$ on $N$ i.i.d.\ samples from $D$ 
            and obtain a halfspace $h(\x)\eqdef \mathrm{sgn}(\w\cdot \x + t)$.
            \item \label{line:slices} Let \(n\eqdef 2\left\lceil C\sqrt{\log(1/\eps)}/\Delta\right\rceil\) and define breakpoints \(s_i\) for \(i\in[n]\) by \(s_1=-C\sqrt{\log(1/\eps)}\) and \(s_i=s_{i-1}+\Delta\) for \(2\le i\le n\).
Define slices \(S_i\) for \(i\in[n+1]\) as \(S_1\eqdef\{\x\in\mathbb{R}^d:\ \w\cdot\x\in(-\infty,s_1]\}\), \(S_{n+1}\eqdef\{\x\in\mathbb{R}^d:\ \w\cdot\x\in(s_n,\infty)\}\), and \(S_i\eqdef\{\x\in\mathbb{R}^d:\ \w\cdot\x\in(s_{i-1},s_i]\}\) for \(2\le i\le n\).
\item Compute an orthonormal basis of $\w^{\perp}$ and denote by $\vec U\in\R^{d\times(d-1)}$ the matrix whose columns form a basis of $\w^\perp$. \label{line:change-of-basis}
                \item Draw $N$ i.i.d.\ samples from $D$ and denote by $\widehat{D}$ the empirical distribution of these samples.\label{line:empirical}
        
            \item  For each slice $S\in \{S_i: i\in [n+1]\}$:
            \begin{enumerate}[leftmargin=*, nosep]

                \item {\it \textbf{Slice mass test}:}\label{line:mass-test}
                Compute
                $\widehat{p}_S= \Pr_{\x\sim \widehat{D}_\x}[\x\in S]$ and if
                $\abs{\widehat{p}_S-\Pr_{\x\sim \cN^d}[\x\in S]}\geq \tau_p$, then \emph{Reject}.

                \item {\it \textbf{Orthogonal moment matching test}:}\label{line:moment-test}
                For all $\alpha\in \N^{d-1}, 1\leq \abs{\alpha}\leq l$ compute
                $\widehat{m}_{\alpha,S}=
                \E_{\x\sim \widehat{D}_\x}[He_{\alpha}(\vec U^{\top} \x)\mid \x\in S]$
                and if
                $\abs{\widehat{m}_{\alpha,S}-\E_{\x\sim\cN^d}[He_{\alpha}(\vec U^{\top} \x)\mid \x\in S]}\geq \tau_m$, then \emph{Reject}.
                
                \item {\it \textbf{Non-negativity certificate}:}\label{line:negativity-test}
                Let $H(\cdot)$ be the vector of all $(d-1)$-dimensional Hermite polynomials of  of total degree $\le l$,
                and define
                \[
                \widehat{\vec M}\eqdef \E_{(\x,y)\sim \widehat{D}}[H(\vec U^{\top} \x)H^{\top}(\vec U^{\top} \x) \left( yh(\x)- \beta+\eps \right) \mid \x \in S ].
                \]
                If $\widehat{\vec M} \not \succeq \vec 0$, then \emph{Reject}. 
            \end{enumerate}

            \item {\it Accept} and output $h$.
        \end{enumerate}
    }}
  \vspace{0.3cm}
    \caption{Testable Learner for $\gamma$-Biased halfspaces.}
    \label{alg:testable-massart-gamma}
\end{algorithm}

\paragraph{Parameter Description} 
The algorithm takes as input the target accuracy $\epsilon$,
Massart noise bound $\eta < \tfrac{1}{2}$, and bias parameter $\gamma$.
It then sets internal parameters, the noise bias $\beta:=1-2\eta$, the polynomial degree $l$, slice width $\Delta$,
sample size $N$, and tolerance levels $\tau_p$ and $\tau_m$ for estimating matching probabilities
and moments, respectively.

\subsection{Proof of Correctness} \label{ssec:analysis}

In the rest of this section, we prove that our algorithm 
satisfies the guarantee of \Cref{def:testable-learning}.

We start with the completeness part; 
that is, if the distributional assumptions are satisfied, 
then our tests accept with high probability. 
The proof amounts to establishing certain structural properties 
of the Massart noise oracle and some standard Gaussian concentration results. 
We refer the reader to Appendix~\ref{sec:proofofcompletness} for the full proof.

\begin{proposition}[Completeness]\label{lem:completeness}
Let $\eta \in [0, 1/2)$ and define $\beta\eqdef 1-2\eta$. 
Let $D$ be a distribution over $\R^d\times \{\pm 1\}$ whose $\x$-marginal is $\cN^d$ and  that satisfies the $\eta$-Massart noise condition. 
Then \Cref{alg:testable-massart-gamma} using $N\geq d^{\log^3( 1/(\beta\max(\eps,\gamma)))\log^2(1/\beta)/\beta^2}\poly(1/\eps)$ samples runs in $\poly(N,d)$ time and returns a halfspace $h$ such that
with probability at least $2/3$ it holds that \[
\Pr_{(\x,y)\sim D}[h(\x)\neq y]
\le \opt_{\gamma}
 + \eps,\;\; \opt_{\gamma}\eqdef \min_{f\in\mathcal{H}_{d,\gamma}} \Pr_{(\x,y)\sim D}[f(\x)\neq y]\;.
\]
\end{proposition}

We now turn to the soundness guarantee: with high probability, our algorithm accepts and attains error close to optimal.
First, we state our guarantee.

\begin{proposition}[Soundness against $\gamma$-Biased halfspaces]\label{prop:soundness-gamma}
Let $D$ be a distribution over $\R^d\times\{\pm1\}$.
The probability that \Cref{alg:testable-massart-gamma} accepts and  outputs a hypothesis $h$ such that 
$$
\Pr_{(\x,y)\sim D}[h(\x)\neq y]
> \opt_{\gamma}
 + \eps,\;\; \opt_{\gamma}\eqdef \min_{f\in\mathcal{H}_{d,\gamma}} \Pr_{(\x,y)\sim D}[f(\x)\neq y]\;,
$$is at most $1/3$.
\end{proposition}

Before continuing with the soundness proof, we isolate the following key technical lemma.
Informally, it says that our tests are strong enough to certify a \emph{local advantage} of the returned classifier $h$ against any competitor halfspace $f$ on every slice where $f$ is not too biased under the Gaussian.
More precisely, consider a non-tail  slice $S$ on which $h$ is constant, i.e., $S=\{\x : \w\cdot \x\in[a,b]\}$ with the threshold of $h$ not in $[a,b]$, and suppose that $h$ passes the moment-matching and non-negativity tests on $S$.
We call the sets $\{\x : \w \cdot \x = z\}$ where $z\in [a,b]$ the fibers of $S$. Note that any halfspace $f:\R^d\to \R$ restricted to a fiber still defines a halfspace on the subspace $\w^\perp$; we refer to the threshold of this halfspace as the induced threshold on that fiber.
Then for any competitor halfspace $f$, if the disagreement region is non-negligible  (at least $\gamma$) under the Gaussian distribution, and the induced threshold  does not vary much across fibers of $S$, the lemma implies that $h$ beats $f$ on that slice by an advantage of at least $\Omega\!\big((1-2\eta)\gamma\big)$. 

\begin{lemma}[$\gamma$-biased slices lead to $\Omega(\gamma)$ advantage]\label{lem:advantage-lemma}
There exists a sufficiently large universal constant $C\in \Z_+$ such that the following holds. 
Let $\eps,\eta\in(0,1/2), \gamma\in (0,1/2]$ and $\vec v,\w\in \R^d$ be unit vectors. 
Suppose that
$\eps<\min\{\gamma,\beta/2\}$
where $\beta\eqdef 1-2\eta$. 
Let $h(\x)=\sign(\w\cdot \x-\tau)$ 
be a halfspace and let $S=\{\x:\w \cdot \x \in [a,b]\}, a,b\in \R, a<b, \tau\not \in [a,b]$. 
Let $D$ be a distribution over $\R^d \times \{\pm 1\}$.
Let $f(\x)= \sign(\vec v\cdot \x-t)$ be another halfspace.
Assume that:
\begin{enumerate}[leftmargin=*, label=(\roman*)]
    \item The halfspace $h$ and the  distribution $D$ pass the tests on Lines \ref{line:negativity-test} and \ref{line:moment-test}  restricted to the slice $S$ with degree $l=C\log^3(1/\gamma)\log^2(1/\beta)/\beta^2$.
    \item \textbf{Non-trivial Disagreement:}
    $\Pr_{\x\sim\cN^d}[f(\x)\neq h(\x)\mid \x\in S]\ge \gamma$.
    \item \textbf{Non-trivial Angle: }$\frac{|\bv\cdot \bw|}{\|\bv^{\perp \bw}\|}\Delta\leq \beta\gamma/C$, $\Delta\eqdef b-a$.
\end{enumerate}
Then it holds that
$\pr_{(\x,y)\sim D}[h(\x)\neq y]-\pr_{(\x,y)\sim D}[f(\x)\neq y]\lesssim -\beta\gamma$.
\end{lemma}

\begin{proof}%
Since $\tau\notin[a,b]$, the sign of $\w\cdot\x-\tau$ does not change for $\x\in S$, and therefore $h$ is constant on $S$.
Without loss of generality, assume that $h(\x)=-1$ for all $\x\in S$.

Consider the indicator $g(\x)=\Ind(f(\x)=1)$. Note that $g$ restricted to $S$ denotes precisely the disagreement region between $h$ and $f$.
Let $\vec b\eqdef \bv^{\perp \bw}/\|\bv^{\perp \bw}\|$ and write, 
\[
g(\x)=\Ind\!\Big(\vec b\cdot \x^{\perp \bw}\ \ge\ s(\bw\cdot \x)\Big),
\qquad
s(\bw\cdot \x)\eqdef \frac{t-(\bv\cdot \bw)(\bw\cdot \x)}{\|\bv^{\perp \bw}\|}.
\]
Let $t_1\eqdef\min_{\x\in S}s(\bw\cdot \x)$, $t_2\eqdef\max_{\x\in S}s(\bw\cdot \x)$ and define the halfspaces
\[
g_1(\x)\eqdef \Ind(\vec b\cdot \x^{\perp \bw}\ge t_1)
\text{ and } 
g_2(\x)\eqdef \Ind(\vec b\cdot \x^{\perp \bw}\ge t_2).
\]
Notice that $g_1\ge g\ge g_2$ on $S$.

We first show that the thresholds $t_1,t_2$ are  $O(\sqrt{\log(1/\gamma)})$, 
since otherwise the disagreement region will have tiny mass under the Gaussian distribution 
(which contradicts the lemma assumption).
Since $\Pr_{\x\sim\cN^d}[g(\x)=1\mid \x\in S]\ge\gamma$ and $g_1\ge g$ on $S$, we have that $\Pr_{\x\sim\cN^d}[g_1(\x)=1\mid \x\in S]\ge\gamma$, which implies that $t_1= O(\sqrt{\log(1/\gamma)})$. 
Moreover, by assumption~(iii)  $t_2-t_1=\frac{|\bv\cdot \bw|}{\|\bv^{\perp \bw}\|}\Delta\le \gamma<1/2$, hence  we also have that  $t_2\leq O(\sqrt{\log(1/\gamma)})$.

We apply \Cref{lem:structural} with accuracy parameter $\alpha\eqdef (1-2\eta-\eps)/C>0$ (for a sufficiently large absolute constant $C>0$) to the univariate threshold function $\Ind(z\ge t_1)$.
This yields polynomials $q^{(1)}_{\pm}:\R\to\R$ of degree at most $l=C\log^3(1/\gamma)\log^2(1/\beta)/\beta^2$ that sandwich $\Ind(z\ge t_1)$ (i.e., 
$q^{(1)}_{-}(z)\le \Ind(z\ge t_1)\le q^{(1)}_{+}(z)$ for all $z\in\R$) and $\E_{z\sim \cN}[q^{(1)}_{+}(z)- q^{(1)}_{-}(z)]\leq \alpha \pr_{z\sim \cN}[z\geq t_1]$.
Correspondingly, the multivariate polynomials $p_\pm^{(1)}(\x)= q_{\pm}^{(1)}(\vec b\cdot \x)$ will sandwich $g_1$, and the expectations will be relatively close under the Gaussian distribution: $\E_{\x \sim \cN^d}[p_+^{(1)}(\x)-p_-^{(1)}(\x)]\leq \alpha\E_{\x \sim \cN^d}[g_1(\x)]$.
Next, we construct the univariate sandwiching polynomials $q_{\pm}^{(2)}$ for $\Ind(z\ge t_2)$
and 
the multivariate sandwiching polynomials $p_{\pm}^{(2)}$ for $g_2$ similarly.
Note that $q_+^{(1)}$ and $q_+^{(2)}$ are positive univariate polynomials and therefore sums of squares; hence, the same holds for $p_+^{(1)}$ and $p_+^{(2)}$ since they are obtained by composing $q_+^{(1)}$ and $q_+^{(2)}$ with a simple linear projection respectively.
We can moreover show that the two bounding polynomial $p_+^{(2)},p_-^{(1)}$ are close, $\E_{\x\sim \cN^d}[p_+^{(2)}(\x)-p_-^{(1)}(\x) \mid S] \lesssim \alpha\E_{\x\sim \cN^d}[g(\x)\mid S] +\frac{|\vec v\cdot \w|}{\|\vec v^{\perp \w}\|}\,\Delta$. We formalize this as \Cref{claim:closeness_bounding_polynomials}; its statement and proof appear in \Cref{sec:soundness-app}.

Before we finish the proof, we first certify that the expectation of the polynomials $p_+^{(2)}(\x)$ and $p_-^{(1)}(\x)$ are nearly the same under $D_\x$ and the Gaussian. 
Since $p_+^{(2)}(\x)$ and $p_-^{(1)}(\x)$ depend only on $\x^{\perp\w}$ and have degree $\deg\le l$,
\Cref{cl:orth-moment-transfer-poly} of \Cref{sec:omitted-facts}
lets us replace conditional expectations under $D_\x$ by those under $\cN^d$ at an additive cost
$\tau_m\cdot d^{l}(\|p_+^{(2)}\|_2+\|p_+^{(2)}-p_-^{(1)}\|_2)$.
Moreover, by \Cref{cl:l1-to-l2} of \Cref{sec:omitted-facts} this cost is at most 
$ 3\tau_m\cdot d^{l}\cdot 3^{l}\cdot \E_{\cN^d}[g(\x)\mid \x\in S]$,
which is at most $\eps^{C}\E_{\cN^d}[g(\x)\mid \x\in S]\le \eps^{C}\gamma$ for our choice $\tau_m\leq \eps^C/d^{2l}$ and sufficiently large $C$.

Now we will use the fact that the region $g$ is well approximated by polynomials, and that for these polynomials we have essentially verified their error behavior via the non-negativity test in Line~\ref{line:negativity-test}. 
Specifically,                       

\begin{align*}
&\E_{(\x,y)\sim D}\!\big[-g(\x)y\mid \x \in S \big]\\
&=
-\E_{(\x,y)\sim D}\!\big[p_+^{(2)}(\x)y\mid \x \in S \big]
+
\E_{(\x,y)\sim D}\!\big[y(p_{+}^{(2)}(\x)-g(\x))\mid \x \in S\big]
\\
&\ge
-\E_{(\x,y)\sim D}\!\big[p_+^{(2)}(\x)y\mid \x \in S \big]
+
\E_{(\x,y)\sim D}\!\big[p_{+}^{(2)}(\x)-p^{(1)}_-(\x)\mid \x \in S\big]
\\
&\geq
(1-2\eta-\eps)\,
\E_{x\sim D_{\x}}\!\big[p_+^{(2)}(\x)\mid \x \in S \big]
\;-\;
\E_{(\x,y)\sim D}\!\big[p_{+}^{(2)}(\x)-p^{(1)}_-(\x)\mid \x \in S\big]
\\
&\ge
(1-2\eta-\eps)\,
\E_{\x\sim \cN^d}\!\big[g(\x)\mid \x\in S\big]
\;-\;
O\left(\alpha\,\E_{\x\sim \cN^d}\!\big[g(\x)\mid \x \in S\big]
+
\frac{|\bv\cdot\bw|}{\|\bv^{\perp \bw}\|}\,\Delta\right) -\eps^C\gamma
\\
&\gtrsim (1-2\eta-\eps)\gamma-\eps^C\gamma\gtrsim (1-2\eta)\gamma\;,
\end{align*}
where in the first inequality we used that $p_+^{(2)}$ and $p_-^{(1)}$ sandwich $g$, 
and in the second inequality we used that we did not reject the non-negativity test for $S$. 
In the third inequality, we used \Cref{claim:closeness_bounding_polynomials} 
and the fact that we match moments of degree at most $l$.
In the fifth inequality, we used the fact that 
$\frac{|\bv\cdot \bw|}{\|\bv^{\perp \bw}\|}\Delta\leq (1-2\eta-\eps)\gamma/C$ for a sufficiently large absolute constant $C$, our choice of $\alpha$ and the moment matching error. 
The  last inequality uses that $\eps<\min\{\gamma,(1-2\eta)/2\}$.

Finally, note that 
\begin{align*}
 &\E_{(\x,y)\sim D}[-g(\x)y\mid \x \in S]
 \\
 &=   \pr_{(\x,y)\sim D}[f(\x)=1,y=-1\mid \x \in S] -\pr_{(\x,y)\sim D}[f(\x)=1,y=1\mid \x \in S]\\
 &=   \pr_{(\x,y)\sim D}[f(\x)\neq h(\x), y\neq f(\x)\mid \x \in S] -\pr_{(\x,y)\sim D}[f(\x)\neq h(\x),y\neq h(\x)\mid \x \in S]\\
 &=   \pr_{(\x,y)\sim D}[ y\neq f(\x)\mid \x \in S] -\pr_{(\x,y)\sim D}[y\neq h(\x)\mid \x \in S] \;.
\end{align*}
If instead $h(\x)=1$ for all $\x\in S$, the same argument applies by renaming labels.
\end{proof}

\paragraph{Proof Sketch of \Cref{prop:soundness-gamma}}
We provide some intuition about the proof using the schematic in
Figure~\ref{fig:case1}.
The figure depicts our learned halfspace $h$, the ``stripes'' (slices) on which the algorithm runs its tests, and an arbitrary competing halfspace $f$.
The region where $h$ and $f$ disagree is shaded in light blue.

Recall from \Cref{lem:advantage-lemma} that, since \Cref{alg:testable-massart-gamma} uses degree
$l=\poly\log(1/\gamma)/(1-2\eta)^2$, whenever within a slice the disagreement between $h$ and $f$ is at least $\gamma^C$ we can certify an
$\Omega(\gamma^C)$ \emph{advantage} for $h$ over $f$ on that slice.
In particular, in the portion of the space marked as the \emph{Advantage Region} (near the intersection of $f$ with the level
$|y|=\Theta(\sqrt{\log(1/\gamma)})$ in the figure), we obtain that $h$ achieves smaller error than $f$ by $\Omega(\gamma)$.

On the remaining slices our tests may not certify a pointwise advantage.
Nevertheless, we show that the total mass of the disagreement region outside the Advantage Region is small, and hence the 
\emph{disadvantage} of $h$ (i.e., how much smaller the error of $f$ can be than that of $h$) contributed by these slices is also small.
This relies on structural/anti-concentration arguments like \Cref{cl:probability-transfer} of \Cref{sec:soundness-app} (e.g., showing that if a halfspace exhibits a large conditional bias on a slice then it must also
exhibit a corresponding bias under $D$).

We need to make this argument for all possible relative configurations of $f$ and $h$ among $\gamma$-biased halfspaces.
In particular the configuration shown in \Cref{fig:case1} the Advantage Region is relatively large (a constant fraction of the space). 
If $f$ looks more parallel to $h$ the advantage region shrinks as $f$ passes faster from the $|y|=\Theta(\sqrt{\log(1/\gamma)})$ level sets. 
However, as the advantage region shrinks so does the disagreement region. 
To see this consider rotating $f$ clockwise around its depicted vector $v$. 
Moreover, if $f$ is nearly parallel to $h$ the non-trivial angle condition of \Cref{lem:advantage-lemma} does not hold 
In fact $f$ 
is not well approximated 
by its piecewise-constant version on slices parallel to $h$ then for which we ally our polynomial approximation result.
However, in this one does not need to apply any polynomial approximation result, since all of the variation from the orthogonal direction of $f$ vanishes and we need to only consider the rate that $h$ is better than $1/2$ in each slice (tested by the degree $0$ instantiation of the non-negativity test). 
Our proof of soundness follows exactly this case analysis and is deferred to Appendix~\ref{sec:soundness-app}. 

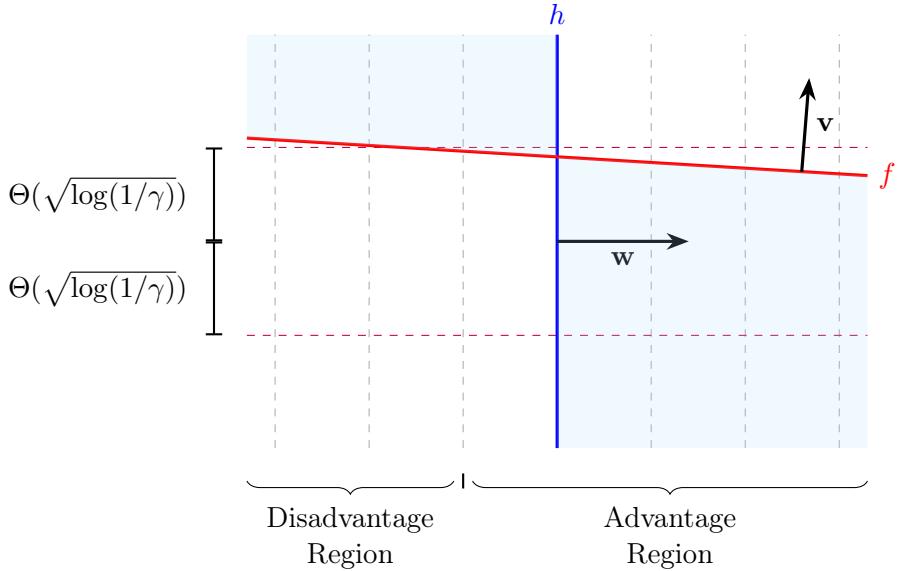
\begin{figure}[H]
  \centering
\usetikzlibrary{decorations.pathreplacing}

\begin{tikzpicture}[scale=1.25,>=Stealth]
  \tikzset{
    slice/.style={gray!65,dashed},
    hline/.style={very thick,blue},
    fline/.style={very thick,red},
    wvec/.style={very thick,->},
    vvec/.style={very thick,->},
  }

  \foreach \x in {-3,-2,-1,1,2,3} {
    \draw[slice] (\x,-2.2) -- (\x,2.2);
  }

  \draw[purple,dashed] (-3.3,1) -- (3.3,1);
  \draw[purple,dashed] (-3.3,-1) -- (3.3,-1);

  \draw[thick,|-|] (-3.65,0) -- (-3.65,1)
    node[midway,left=6pt] {$\Theta(\sqrt{\log(1/\gamma)})$};
  \draw[thick,|-|] (-3.65,0) -- (-3.65,-1)
    node[midway,left=6pt] {$\Theta(\sqrt{\log(1/\gamma)})$};

  \draw[hline] (0,-2.2) -- (0,2.2) node[above] {$h$};

  \draw[wvec] (0,0) -- (1.4,0) node[midway,below] {$\mathbf{w}$};

  \draw[fline] (-3.3,1.10) -- (0,0.90) -- (3.3,0.70) node[right] {$f$};

  \coordinate (P1) at (2.6,0.74);
  \draw[vvec] (P1) -- ++(0.08,1.00) node[midway,right] {$\mathbf{v}$};

\path[fill=cyan!30,opacity=0.18,draw=none]
  (0,-2.2) -- (3.3,-2.2) -- (3.3,0.70) -- (0,0.90) -- cycle;

\path[fill=cyan!30,opacity=0.18,draw=none]
  (-3.3,2.2) -- (0,2.2) -- (0,0.90) -- (-3.3,1.10) -- cycle;

\def\braceY{-2.55}
\def\braceGap{0.10} %

\draw[decorate,decoration={brace,mirror,amplitude=5pt}]
  (-3.3,\braceY) -- (-1-\braceGap,\braceY)
  node[midway,below=6pt,align=center] {Disadvantage\\Region};

\draw[decorate,decoration={brace,mirror,amplitude=5pt}]
  (-1+\braceGap,\braceY) -- (3.3,\braceY)
  node[midway,below=6pt,align=center] {Advantage\\Region};

\draw[thick] (-1,\braceY-0.08) -- (-1,\braceY+0.08);

\end{tikzpicture}

  \caption{Illustration of a learned halfspace $h$, a competing halfspace $f$, the slices on which the tests are performed, and their disagreement region (shaded).}
  \label{fig:case1}
\end{figure}

\section{Conclusions} \label{sec:conclusions}
In this work, we gave the first algorithm for testable learning of general Massart 
halfspaces under the Gaussian distribution. The complexity of our algorithm qualitatively 
matches known Statistical Query lower bounds, even for the non-testable setting. In the 
process, we established a novel upper bound on the degree of sandwiching polynomial 
approximations for the sign function with a multiplicative error guarantee. A natural 
open question is to understand the complexity of testable learning with Massart noise 
under more general input distributions. This goal appears attainable for homogeneous halfspaces, where efficient algorithms are known (in the non-testable setting) 
under a range of structured distributions; see, e.g.,~\cite{DKTZ20}. 
For the case of general halfspaces, the only known upper bounds~\cite{DKKTZ22} 
rely on the Gaussian assumption. Hence, further progress for the general case 
requires a deeper understanding of the non-testable setting.

\newpage 

\bibliographystyle{alphaabbr}

\bibliography{colt2026/mybib}

\newpage

\appendix

\section*{Appendix}
The Appendix is structured as follows:  \Cref{sec:omitted-facts} includes additional preliminaries required in subsequent technical sections. 
\Cref{sec:proofofcompletness} contains the proof of the  Completeness guarantee (\Cref{lem:completeness}).
\Cref{sec:soundness-app} contains the proof of the Soundness guarantee (\Cref{prop:soundness-gamma}).
 \Cref{sec:sandwiching-result} contains the full proof of our sandwiching result (\Cref{lem:structural}).
Finally, \Cref{app:sq-lb} contains a full proof of our SQ lower bound result.

\section{Omitted Facts and Preliminaries} \label{sec:omitted-facts}
\noindent {\bf Probability Notation}
We use $\E_{x\sim D}[x]$ for the expectation of the random variable $x$ according to the
distribution $D$ and $\pr[\mathcal{E}]$ for the probability of event $\mathcal{E}$.
For simplicity
of notation, we may omit the distribution when it is clear from the context.  For $(\x,y)$
distributed according to $D$, we denote {by} $D_\x$ to be the distribution of $\x$ and {by} $D_y$ to be the
distribution of $y$.
Let $\normal( \boldsymbol\mu, \vec \Sigma)$ denote the $d$-dimensional Gaussian distribution with mean $\boldsymbol\mu\in  \R^d$ and covariance $\vec \Sigma\in \R^{d\times d}$.
We denote by $\cN^d$ the high dimensional standard normal $\cN(\vec 0,\vec I)$ and we omit the exponent when $d=1$.
We denote by $\phi_d(\cdot)$ the pdf of the $d$-dimensional {standard normal} and we use  $\phi(\cdot)$ for the pdf of the {$1$-dimensional} standard normal.
For a function $f:\mathbb{R}^d\to\mathbb{R}$, define $\|f\|_{L^2(\mathcal{N}^d)}:=\left(\mathbb{E}_{\mathbf{x}\sim\mathcal{N}^d}\!\left[f(\mathbf{x})^2\right]\right)^{1/2}$. When the underlying measure is clear, we simply write $\|f\|_{L^2}$.

\paragraph{Hermite Polynomials}
We will use the following notion of normalized Hermite polynomials.
\begin{definition}[Normalized Hermite Polynomial]\label{def:Hermite-poly}
For $k\in\N$,
we define the $k$-th \emph{probabilist's} Hermite polynomials
$p_k:\R\to \R$
as
$p_k(t)=(-1)^k e^{t^2/2}\cdot\frac{d^k}{dt^k}e^{-t^2/2}$.
We define the $k$-th \emph{normalized} Hermite polynomial 
$He_k:\R\to \R$
as
$He_k(t)=\mathrm{\textit{p}}_k(t)/\sqrt{k!}$.
\end{definition}
Let $d\in\N$ and let $\alpha=(\alpha_1,\dots,\alpha_d)\in\N^d$ be a multi-index.  
We define the \emph{multivariate normalized Hermite polynomial} $He_\alpha:\R^d\to\R$ by
$He_\alpha(\x)=\prod_{i=1}^d He_{\alpha_i}(\x_i)$.
Fix $k\in\N$ and let 
$
\mathcal P_{<k}:=\{p:\R^d\to\R \text{ polynomial with total degree }<k\}.
$
Then the family $\{He_\alpha:\alpha\in\N^d,\ |\alpha|<k\}$ forms an orthonormal basis of $\mathcal P_{<k}$ under the standard Gaussian measure on $\R^d$.

Now we state some important facts that we will use in the technical sections. We start with some algorithmic facts. First we state the guarantee of \cite{DKKTZ22} about learning general halfspaces with Massart noise.

\begin{fact}[Learning General Halfspaces with  Massart Noise under Gaussian Marginals] 
\label{fact:general-massart}
Fix parameters $\eta \in [0, 1/2), \beta\eqdef1-2\eta$ and  $\gamma\in(0,1/2]$. 
Let $\D$ be a distribution on $\R^d\times\{\pm1\}$ whose $\x$-marginal is $\cN^d$, that satisfies
the $\eta$-Massart noise condition with respect to a halfspace in $\cH_{d,\gamma}$. There exists an algorithm that draws $N = d^{O(\log( 1/(\beta\max(\gamma,\eps)) ))}
\poly(1/\eps)\log(1/\delta)$ samples from $\D$, runs in time $\poly(N, d)$, and computes a halfspace $h$ such that
with high probability it holds \(\pr_{(\x, y) \sim \D}[h(\x) \neq y] \leq \opt  + \eps\).
\end{fact}
Next we show that passing the Orthogonal Moment Matching test implies that, on any slice \(S\), the conditional expectation of any degree-\(l\) polynomial depending only on \(\x^{\perp \w}\) is close to its Gaussian counterpart.
\begin{claim}[Moment matching error]
\label{cl:orth-moment-transfer-poly}
Let $d\in \Z_+, d\geq 3$ and let $\w\in \R^d$ be a unit vector.
Fix a slice $S=\{\w\cdot\x\in[a,b]\}$ and a degree parameter $l\in \Z_+$.
Assume that the Orthogonal Moment Matching test (Line~\ref{line:moment-test}) passes on $S$ with tolerance $\tau_m$.
Then for every polynomial $p:\R^{d}\to\R$ of total degree at most $l$ that depends only on $\x^{\perp \w}$
(i.e., $p(\x)= p(\x^{\perp \w})$),
\[
\left|
\E_{\x\sim D_\x}\!\big[p(\x)\mid \x\in S\big]
-
\E_{\x\sim \cN^d}\!\big[p(\x)\mid \x\in S\big]
\right|
\;\le\;
\tau_m\cdot d^{l}\cdot \|p\|_{L^2}.
\]
\end{claim}

\begin{proof}
Let $\vec u^{(1)},\dots,\vec u^{(d-1)}$  be an orthonormal basis of $\w^{\perp}$ and define the matrix $\vec U=[\vec u^{(1)} \cdots \vec u^{(d-1)}]\in \R^{d\times (d-1)}$.
Then $\vec U^{\top} \vec U=\vec I_{d-1}$ and $\vec U\vec U^{\top}=\Pi_{\w^{\perp}}=\vec I-\w\w^{\perp}$. 
We can write $\x^{\perp \w}= \vec U\vec U^{\top}\x$. 
Let $p:\R^{d}\to \R$ be a polynomial of degree at most $l$ such that $p(\x)=p(\x^{\perp \w})$. 
Then we define the induced $(d-1)$-variate polynomial $\widetilde{p}(\z)\eqdef p(\vec U\z), \z \in\R^{d-1}$.
Moreover, if $\z=\vec U^{\top}\x$, then $p(\x)=p(\x^{\perp \w})=p(\vec U\vec U^{\top}\x)=\widetilde{p}(\vec U^{\top}\x)$.

Write the Hermite expansion of $\widetilde p$ (over $\cN^{d-1}$) as
\(
\widetilde p(\z)=\sum_{|\alpha|\le l} a_\alpha\, He_{\alpha}(\z),
\)) $\alpha\in \N^{d-1}$.
If $\x\sim \cN^d$, the random vector $[\w \x, U^{\top}\x]$ is a standard Gaussian vector in $\R^{d}$ with independent coordinates. 
Hence $U^{\top}\x\sim \cN^{d-1}$ and is indepeendent of the event $\x \in S$. 
Therefore, for all $\alpha\in \N^{d-1}$
\(
\E_{\x\sim\cN^d}[He_{\alpha}(\vec U^{\top}\x)\mid \x\in S]=\E_{\z\sim\cN^{d-1}}[He_{\alpha}(\z)].
\)
Thus,
\[
\E_{D_\x}[p(\x)\mid S]-\E_{\cN^d}[p(\x)\mid S]
=\sum_{1\le |\alpha|\le l} a_\alpha\,
\Big(\E_{D_\x}[He_{\alpha}(\vec U^{\top}\x)\mid S]-\E_{\cN^d}[He_{\alpha}(\vec U^{\top}\x)\mid S]\Big).
\]
Taking absolute values and using the test guarantee gives us
\[
\left|\E_{D_\x}[p\mid S]-\E_{\cN^d}[p\mid S]\right|
\le \tau_m\sum_{1\le|\alpha|\le l}|a_\alpha|.
\]
Let $M=\#\{\alpha:1\le|\alpha|\le l\}=\binom{d-1+l}{l}-1\leq d^{2l}$, where the inequality holds for $d\geq 3$.
By Cauchy--Schwarz,
\(
\sum |a_\alpha|\le \sqrt{M}\,\|a\|_2.
\)
Finally, by orthonormality of the Hermite basis (Parseval),
\(
\|a\|_2=\|\widetilde p\|_{L^2(\cN^{d-1})}=\|p\|_{L^2}.
\)
Thus
\[
\left|\E_{D_\x}[p\mid S]-\E_{\cN^d}[p\mid S]\right|
\le \tau_m\sqrt{M}\,\|p\|_{L^2}
\le \tau_m\cdot d^{l}\cdot \|p\|_{L^2},
\]
as claimed.
\end{proof}

We show the following fact, which shows that matching moments ensures that halfspaces that are not too biased under the Gaussian distribution are also not too biased under $D$.
A version of this lemma, stated in terms of monomial moments (rather than Hermite moments), was proved in \cite{KSV24}.
\begin{fact}[Matching $\log(1/\eps)$ moments implies bias preservation up to $\eps$]
There exist sufficiently large absolute constants $C_1,C_2>0$ such that the following holds.
Fix $d\in\Z_{>0}$ and $\eps\in(0,1)$, let $k=\lceil \log(1/\eps)\rceil$ and $T=\eps^{-C_1}$.
Let $D$ be a distribution over $\R^d$ such that for every multi-index $\alpha\in\N^d$ with
$1\le |\alpha|\le k$,
\[
\big|\E_{\x\sim D}[He_{\alpha}(\x)]-\E_{\x\sim \cN^d}[He_{\alpha}(\x)]\big|
\;\le\; \frac{\eps}{d^{C_2k}}.
\]
Then for any halfspace $f(\x)=\sign(\vec v \cdot \x-\tau)$ and $y\in\{\pm1\}$, if
$\Pr_{\x\sim \cN^d}[f(\x)=y]\le C_1/T$, then $\Pr_{\x\sim D}[f(\x)=y]\le \eps$.
\end{fact}

\begin{proof}
Assume that $y=1$ from the symmetry of $\cN^d$ (if $\x\sim \cN^d$ then $-\x\sim \cN^d$) the same proof works for $y=-1$.
It suffices to bound $\Pr_{\x\sim D}[\vec v\cdot\x\ge \tau]$. From $\Pr_{g\sim\cN(0,1)}[g\ge \tau]\le C_1/T$ and a standard Gaussian tail lower bound we get $\tau \ge \tfrac12\sqrt{\log(T/2)}$.
Take $k$ even. Markov's inequality gives
\[
\Pr_{\x\sim D}[\vec v\cdot\x\ge \tau]\le \frac{\E_{\x\sim D}[(\vec v\cdot\x)^k]}{\tau^k}.
\]
By our moment matching hypothesis and 
\Cref{cl:orth-moment-transfer-poly} we have 
\[
\E_{\x\sim D}[(\vec v\cdot\x)^k]\le \E_{\x\sim \cN^d}[(\vec v\cdot\x)^k] \;+\; d^k \cdot \frac{\eps}{d^{C_2k}}\sqrt{\E_{\x\sim \cN^d}[(\vec v\cdot\x)^{2k}]}.
\]
Since $\vec v\cdot\x\sim \cN(0,1)$ under $\cN^d$, we have
$\E_{\cN^d}[(\vec v\cdot\x)^k]=(k-1)!!\le k^{k/2}$.
Combining and using $\tau\ge \tfrac12\sqrt{\log(T/2)}$ with $T=\eps^{-C_1}$ and $k=\Theta(\log(1/\eps))$,
choosing $C_1,C_2$ sufficiently large makes the RHS at most $\eps$.
\end{proof}
Now we state standard Gaussian inequalities that relate \(L^q\), \(L^2\), and \(L^1\) norms of low-degree polynomials.
\begin{fact}[Gaussian Hypercontractivity; see, e.g., \cite{ODonnellbook}]
\label{fact:hypercontractivity}
Let $p:\R^d\to \R$ be a polynomial of degree at most $m$ which has zero mean and variance one under the gaussian distribution.
 For every real number $q \ge 2$, we have
$
\|p\|_{L^q} 
=
(q - 1)^{\,\frac{m}{2}}\,\|p\|_{L^2}\;.$
\end{fact}
\begin{fact}[$L^1\to L^2$ norm control for Gaussian low-degree polynomials]
\label{cl:l1-to-l2}
Let $p:\R^m\to\R$ be a polynomial of degree at most $l$. Under the standard Gaussian measure,
$\|p\|_{L^2}\;\le\; 3^{\,l}\,\|p\|_{L^1}$.
\end{fact}
\begin{proof}
By Hölder interpolation (log-convexity of $L^q$ norms), since $\,\frac12=\frac13\cdot 1+\frac23\cdot\frac14\,$ we have
$\|p\|_{L^2} \le \|p\|_{L^1}^{1/3}\,\|p\|_{L^4}^{2/3}$.
By Gaussian hypercontractivity ( \Cref{fact:hypercontractivity}) with $t=4$,
\[
\|p\|_{L^4} \le (4-1)^{l/2}\|p\|_{L^2} = 3^{l/2}\|p\|_{L^2}.
\]
Substituting,
\[
\|p\|_{L^2} \le \|p\|_{L^1}^{1/3}\,(3^{l/2}\|p\|_{L^2})^{2/3}
= 3^{l/3}\,\|p\|_{L^1}^{1/3}\,\|p\|_{L^2}^{2/3}.
\]
Rearrange to get $\|p\|_{L^2}^{1/3}\le 3^{l/3}\|p\|_{L^1}^{1/3}$, hence $\|p\|_{L^2}\le 3^l\|p\|_{L^1}$.
\end{proof}

\section{Proof of \Cref{lem:completeness} } 
\label{sec:proofofcompletness}
In this section, we provide the full proof of \Cref{lem:completeness}, 
restated below for convenience. 
\begin{proposition}[Completeness]\label{lem:completenes-appendix}
Fix a constant $\eta \in [0, 1/2)$ and define $\beta\eqdef 1-2\eta$. 
Let $D$ be a distribution over $\R^d\times \{\pm 1\}$ whose $\x$-marginal is $\cN^d$ and  that satisfies the $\eta$-Massart noise condition. 
Then \Cref{alg:testable-massart-gamma} using $N\geq d^{\log^3( 1/(\beta\max(\eps,\gamma)))\log^2(1/\beta)/\beta^2}\poly(1/\eps)$ samples runs in $\poly(N,d)$ time and returns a halfspace $h$ such that
with probability at least $2/3$ it holds that \[
\Pr_{(\x,y)\sim D}[h(\x)\neq y]
\le \opt_{\gamma}
 + \eps,\;\; \opt_{\gamma}\eqdef \min_{f\in\mathcal{H}_{d,\gamma}} \Pr_{(\x,y)\sim D}[f(\x)\neq y]\;.
\]
\end{proposition}
\begin{proof}
Denote by $l$ the polynomial degree of the tests, $\eps'$ the prior algorithm accuracy parameter and $\Delta$ be the slice width as defined in line~\ref{line:init}.
Denote by $S_i$, $i\in [n+1]$, the slices defined as in line~\ref{line:slices}, and let $n+1$ be their number. Note that $n+1 = O(1/\Delta) = \poly(1/\eps)$.

First note that by applying \Cref{fact:general-massart}, Line~\ref{line:run-prev-alg} with probability at least $11/12$ returns a halfspace $h$ with error at most 
$\opt+\eps' \leq \opt + \eps$ in $\poly(N,d)$ time.

It therefore remains to show that, with high probability, none of our tests rejects $h$. We proceed by establishing several concentration bounds for the Gaussian distribution, which we state and prove in the following claims.
\begin{claim}[Slice mass concentration]\label{cl:slice-concetration}
There exists a large enough universal constant $C>0$ such that the following hold.
Let $\widehat{D}$ be the empirical distribution over $N\geq C d^{3l}\log(1/\delta)/\eps^3$ i.i.d.\ samples from $D_\x = \normal(\vec 0, \vec I)$. 
Then with probability $1-\delta$ it holds that $|\E_{\x\sim \widehat{D}_{\x}}[\x\in S_i]-\E_{\x\sim \cN^d}[\x\in S_i] |\leq \eps/{d^l}$ for all slices $S_i, i\in [n+1]$, where the randomness is over the i.i.d.\ samples drawn from $D$.
\end{claim}
\begin{proof}%
Fix a slice $S_i ,i \in [n+1]$. 
Let $\x^{(1)},\dots,\x^{(N)}\sim D_\x$ i.i.d.\ and define $Z_j := \Ind(\x_j \in S_i)\in\{0,1\}$.
Write $p := \E[Z_j] = \Pr_{\x\sim D_\x}[\x\in S_i]$.
Also $\E_{\x\sim \widehat{D}_{\x}}[\x\in S_i] = \frac{1}{N}\sum_{j=1}^N Z_j$.

By Hoeffding's inequality, for any $t>0$,
\[
\Pr\!\left[\left|\frac{1}{N}\sum_{j=1}^N Z_j - p\right| \ge t\right]
\le 2\exp(-2Nt^2).
\]
Taking $t=\eps/d^l$ and requiring
$
N \ge \frac{1}{2t^2}\log\frac{2}{\delta}
= \frac{d^{2l}}{2\eps^2}\log\frac{2}{\delta},
$
we get that with probability at least $1-\delta$,
\[
\left|\E_{\x\sim \widehat{D}_{\x}}[\x\in S_i]-\Pr_{x\sim D_x}[x\in S_i]\right|
\le \eps/d^l\;.
 \]
Finally, since the number of slices is $O(1/\Delta)=\poly(1/\eps)$, applying the union bound over the slices $S_i$ completes the proof of \Cref{cl:slice-concetration}.

\end{proof}
Denote by $\vec U$ the change of basis matrix computed in line \ref{line:change-of-basis}.
Note that when $D_{\x}$ is the standard Gaussian distribution, since $\vec U$ is orthonormal, $\x^{\w}$ and $\x^{\perp \w}$ are independent, and the slices $S_i$, $i \in [n+1]$, condition only on $\x^{\w}$, the distribution of the random vector $\z = \vec U^{\top}\x \in \R^{d-1}$ conditioned on $\x \in S_i$ is $\mathcal{N}^{d-1}$.

\begin{claim}[Gaussian Moment concentration]
\label{cl:moment-concetration}
There exists a large enough universal constant $C>0$ such that the following hold.
Let $\widehat{D}$ be the empirical distribution on $N\geq Cd^{2l} \log(1/(\delta \Delta))/(\eps^C)$ i.i.d.\ samples from $D_\x = \normal(\vec 0, \vec I)$.
Then with probability $1-\delta$  for all multi-indices $\alpha \in \N^{d-1}$ satisfying $1\leq \abs{\alpha}\leq l$,
and all slices $S_i, i\in [n+1]$,  it holds that
$|\E_{\x\sim \widehat{D}_{\x}}[He_{\alpha}(\vec U^{\top} \x)|\x\in S_i]-\E_{\x\sim \cN^d}[He_{\alpha}(\vec U^{\top} \x)|\x\in S_i] |\leq \eps/d^{l}$.
\end{claim}\begin{proof}
Fix $\alpha\in \N^{d-1}$ with $1\le |\alpha|\le l$ and some slice $S_i$.
Note that $\E_{\x\sim \cN^d}\!\left[He_{\alpha}(\vec U^{\top}\x)\mid \x\in S_i\right]=\E_{\z\sim \cN^{d-1}}\!\left[He_{\alpha}(\z)\right]$. 

Let $\mu_\alpha := \E_{\z\sim \cN^{d-1}}[He_{\alpha}(\z) ]$ and let
$\widehat\mu_\alpha := \E_{\x\sim \widehat D_\x}[He_{\alpha}(\vec U^{\top}\x) |\x\in S_i]$ be the empirical mean.
Using Chebyshev's inequality we have 
\[
\Pr\big[\,|\widehat\mu_\alpha-\mu_\alpha|\ge t\,\big]
\le \frac{1}{Nt^2}\mathrm{Var}(He_{\alpha}(\vec U^{\top}\x) |\x\in S_i)
\le \frac{1}{Nt^2}\E[He_{\alpha}^2 (\z)]/\pr[\x\in S_i]\leq  \frac{1}{Nt^2\poly(\eps)} \;,
\]
where we used the fact that for all slices as defined in Line \ref{line:slices} it holds that $\pr[\x\in S_i]\geq \poly(\eps)$ for all $i\in [n+1]$. 
Taking $t=\eps/d^l$ and $N \ge C d^{4l} \log(1/(\delta \Delta))/(\eps^C\delta) $
(for a large enough universal $C$) makes the above probability at most  $\delta/d^{2l}$.

Now the number of multi-indices $\alpha\in \N^d$ with $1\le |\alpha|\le l$
is at most  $d^{2l}$ and the number of slices is at most $\Delta^{-1}$, hence by the union bound with at most $Cd^{4l}/(\eps^C\delta)$ samples for a sufficiently large constant $C$ we can estimate all expectations simultaneously to error $\eps/d^{l}$.
Which completes the proof of \Cref{cl:moment-concetration}.
\end{proof}

\begin{claim}[Non-negativity certificate]\label{cl:non-negativity}
There exists a large enough universal constant $C>0$ such that the following hold.
Fix $S_i$ a set defined in Line \ref{line:slices} and 
denote by $\widehat{\vec M}_i$ the matrix computed in Line \ref{line:negativity-test} restricted to the slice $S_i$ defined as in Line~\ref{line:slices}
using $N\geq Cd^{2l}/(\eps^C\delta)$ i.i.d.\ samples from $D$. 
Then with probability $1-\delta$ it holds that
$\widehat{\vec M_i} \succeq 0$ for all $i\in [n+1]$.
\end{claim}
\begin{proof}%
Let $f:\R^d\to\{\pm 1\}$ be the (unknown) halfspace with respect to which $D$ satisfies the
$\eta$-Massart condition. 
Let $H(\z)\in\R^m$ be the vector of all Hermite polynomials $He_{\alpha}(\z), \alpha \in \N^{d-1} ,\z\in \R^{d-1}$ of total degree
$|\alpha|\le l$ (so $m=\binom{d-1+l}{l}=d^{O(l)}$). Define the following three matrices (all conditioned on $\x\in S_i$):

$$
\vec M^{\star} \eqdef \E_{(\x,y)\sim D}\!\left[H(\vec U^{\top}\x)H(\vec U^{\top}\x)^\top\,(y f(\x)-(1-2\eta))\mid \x\in S_i\right]\;,
$$
$$
\vec M \eqdef \E_{(\x,y)\sim D}\!\left[H(\vec U^{\top}\x)H(\vec U^{\top}\x)^\top\,(y h(\x)-(1-2\eta))\mid \x\in S_i\right]\;,
$$
$$
\vec M' \eqdef \E_{(\x,y)\sim \widehat{D}}\!\left[H(\vec U^{\top}\x)H(\vec U^{\top}\x)^\top\,(y h(\x)-(1-2\eta))\mid \x\in S_i\right]\;.
$$
We first prove that $\vec M^\star\succeq 0$.
For notational simplicity denote by $\z\in \R^{d-1}$ the random vector $\z= \vec U^{\top} \x$.
Fix any $\vec a\in\R^m$ and let $p(\z)\eqdef \vec a^\top H(\z)$.
Then
\begin{align*}
\vec a^\top \vec M^\star \vec a
&=\E_{(\x,y)\sim D}\!\left[p(\z)^2\,(y f(\x)-(1-2\eta))\mid \x\in S_i\right]\\
&=\E_{(\x,y)\sim D}\!\left[p(\z)^2\cdot \E[y f(\x)-(1-2\eta) \mid \x]\mid \x\in S_i\right].  
\end{align*}

Under $\eta$-Massart noise w.r.t.\ $f$, we have $\E[y\mid \x]=(1-2\eta(\x))f(\x)$ with
$\eta(\x)\le \eta$, hence $\E[y f(\x)\mid \x]=1-2\eta(\x)\ge 1-2\eta$ and therefore
\[
\E_{(\x,y)\sim D}[y f(\x)-(1-2\eta) \mid \x] = (1-2\eta(\x))-(1-2\eta) \ge  0.
\]
Thus $\vec a^\top \vec M^\star \vec a\ge 0$ for all $\vec a$, i.e. $\vec M^\star\succeq 0$.

Now we prove that $\vec M$ is close to $\vec M^\star$ in operator norm.
Let $B=\{\x: h(\x)\neq f(\x)\}$. Since $h,f\in\{\pm1\}$, we have $|h(\x)-f(\x)|\le 2\Ind(\x\in B)$.
Fix any unit vector $\vec a\in\R^m$ and let $p(\z)=\vec a^\top H(\z)$.
\begin{align*}
\left|\vec a^\top(\vec M-\vec M^\star)\vec a\right|
\le 2\, \E_{(\x,y)\sim D}\!\left[p^2(\z) \Ind(\x\in B)\mid \x\in S_i\right].
\end{align*}
By Cauchy--Schwarz (under the conditional law given $\x\in S_i$),
\[
\E_{(\x,y)\sim D}\!\left[p^2(\z) \Ind(\x\in B)\mid \x\in S_i\right]
\le \sqrt{\E_{(\x,y)\sim D}[p(\z)^4\mid \x\in S_i]}\cdot \sqrt{\Pr_{(\x,y)\sim D}[\x\in B\mid \x\in S_i]}.
\]
Next since $\z$ is independent of $\w\cdot\x$ and $S_i$ is of the form $\{\x:\w\cdot \x \in [a,b]\}$, for some $a<b$ we have 
\[
\E_{(\x,y)\sim D}[p(\z)^4\mid \x\in S_i] =\E_{\z\sim \cN^{d-1}}[p(\z)^4]\;.
\]
Since $p$ has total degree at most $l$, Gaussian hypercontractivity gives
$\|p\|_{L^4} \le (4-1)^{l/2}\|p\|_{L^2} = 3^{l/2}\|p\|_{L^2}$.
Moreover, $\E[p(\z)^2]=\|\vec a\|_2^2=1$ by orthonormality of the Hermite basis
under the Gaussian measure.
Thus,
\(
\sqrt{\E[p(\z)^4\mid \x\in S_i]}\le 3^{l/2}\;.
\)
Also $\Pr[\x\in B\mid \x\in S_i]\le \Pr[\x\in B]/\Pr[\x\in S_i]$. Hence,
$$\left|\vec a^\top(\vec M-\vec M^\star)\vec a\right|\leq 3^{l/2} \sqrt{\frac{\Pr[\x\in B]}{\Pr[\x\in S_i]}}\;.$$
Finally, since $h$ is $\eps'$-optimal and $D$ is $\eta$-Massart w.r.t.\ $f$, the excess-error identity
implies $\Pr[\x\in B]\le \eps'/(1-2\eta)$. Also by the slice construction and slice-mass test 
$\Pr[\x\in S_i]\ge \poly(\eps)$. Hence, choosing $\eps'\leq \eps^{C}/3^{l}$ for a sufficiently large constant $C>0$ implies $\|\vec M-\vec M^\star\|_{\op}\le \eps/3$.

Now we show that $\vec M'$ concentrates around $\vec M$ in Frobenius norm.
Now note that each element of the matrix $\vec M'$ is an empirical estimate of a quantity of the form
$\E_{(\x,y)\sim D}[He_{\alpha}(\z) (yh(\x)-(1-2\eta))\mid \x\in S_i]$ for $|\alpha|\le 2l$.
By the same argument as in \Cref{cl:moment-concetration}, we can estimate all entries using
using $d^{2l}/(\eps^C\delta)$ 
 with probability $1-\delta$ 
 that
\[
\|\vec M-\vec M'\|_F \le \eps/3.
\]
Combining the above
by the triangle inequality we have
\[
\|\vec M'-\vec M^\star\|_{\op}
\le \|\vec M'-\vec M\|_\op + \|\vec M-\vec M^\star\|_\op
\le 2\eps/3.
\]
Since $\vec M^\star\succeq \vec 0$, we have 
\[
\lambda_{\min}(\vec M') \ge \lambda_{\min}(\vec M^\star)-\|\vec M'-\vec M^\star\|_{\op}
\ge -\eps,
\]
after adjusting constants, i.e. $\vec M'\succeq -\eps \vec I$.

Finally note that from the above for a coefficient vector $\vec a$ we have 
$$\vec a^{\top} \vec M'\vec a \geq -\eps \|\vec a\|^2_2= -\eps \|\vec a^\top H(\z)\|^2_{L^2}\;,$$
where we used Parseval's identity.
This implies that $\widehat{\vec M}_i\succeq \vec 0$ after simple rearrangements, which concludes the proof of \Cref{cl:non-negativity}.
\end{proof}

Combining the above claims for parameter $\delta=1/12$ with the union bound (note that there are 3 tests and the event that \Cref{fact:general-massart} returns an $\opt+\eps'$ halfspace), we can conclude that \Cref{alg:testable-massart-gamma}
uses $N= d^{O(\log^{3}( 1/(\beta\max(\eps,\gamma)))\log^2(1/\beta)/\beta^2)}\poly(1/\eps)$ samples and $\poly(N, d)$ time, and outputs a classifier with error at most
$\opt + \eps$. This completes the proof of \Cref{lem:completeness}. 
\end{proof}

\section{Omitted Content from Soundness Proof}\label{sec:soundness-app}
In this section, we prove \Cref{prop:soundness-gamma}.
First, we prove certain  results that are integral to the analysis. 
We prove an important auxiliary lemma which shows that matching the first 
$\polylog(1/\gamma)$ 
moments of our distribution with those of the Gaussian essentially certifies that the two induce similar biases gaussian bias up to $\gamma$.
\begin{lemma}\label{cl:probability-transfer}
There exists a sufficiently large absolute constant $C>0$ such that the following holds.
Let $D_\x$ be a distribution over $\R^d$,  $\gamma\in(0,1/2]$, $\w \in \R^d$ a unit vector, and
$S=\{\x:\w\cdot \x\in[a,b]\}$ for some $a<b$.
Assume that for this set $S$, the distribution $D_\x$ matches
$l\eqdef C\log(1/\gamma)$ \emph{orthogonal (Hermite) moments} with error at most
$\gamma/d^{Cl}$ as in Line~\ref{line:moment-test}.
Let $s$ be a halfspace depending only on the orthogonal complement of $\w$, i.e.,
$s(\x)=s(\x^{\perp \w})$.
If $\Pr_{\x\sim \cN^d}\!\left[s(\x)=1 \mid \x\in S\right]\le \gamma^{C}$, then
$\Pr_{\x\sim D_{\x}}\!\left[s(\x)=1 \mid \x\in S\right]\le \gamma$.
\end{lemma}

\begin{proof}%
Write $s(\x)=\sign(\bu\cdot \x^{\perp \w}-\tau)$ for some unit $\bu\perp \w$ and $\tau\in\R$.
Denote by $\vec U\in \R^{d\times (d-1) }$ the change of basis matrix computed in line \ref{line:change-of-basis}.
For notational simplicity denote by $\z\in \R^{d-1}$ the  vector $\z= \vec U^{\top} \x$.
By the orthogonal moment matching test (Line~\ref{line:moment-test}), for all multi-indices
$\alpha\in \N^{d-1}$ with $1\le |\alpha|\le l$,
\[
\Big|
\E_{\x\sim D_\x}\!\big[He_\alpha(\z) \mid \x \in S\big]
-
\E_{\x\sim \cN^d}\!\big[He_\alpha(\z) \mid  \x \in S\big]
\Big|
\;\le\; \gamma/d^{Cl}.
\]
Under $\cN^d$, $\w \cdot \x$ is independent of $\x^{\perp \w}$ and therefore independent of its $(d\!-\!1)$-dimensional representation $\z$, and $\z \sim \cN^{d-1}$; hence $\E_{\x \sim \cN^d}\!\left[He_\alpha(\z)\mid \x \in S\right]=\E_{\bz \sim \cN^{d-1}}\!\left[He_\alpha(\bz)\right]$.

Let $\widetilde D$ denote the distribution of $\z$ under $\x\sim D_\x$ conditioned on $\x\in S$.
Then $\widetilde D$ matches Hermite moments with $\cN^{d-1}$ up to degree $l$ to error $\gamma/d^l$.

We now apply the following fact, which shows that matching moments ensures that halfspaces that are not too biased under the Gaussian distribution are also not too biased under $D$.
A version of this lemma, stated in terms of monomial moments (rather than Hermite moments), was proved in \cite{KSV24}. We refer the reader to \Cref{sec:omitted-facts} for the proof of \Cref{fact:moment-bias-transfer}.
\begin{fact}[Matching $\log(1/\eps)$ moments implies bias preservation up to $\eps$]
\label{fact:moment-bias-transfer}
Fix $\eps>0, d\in \Z_+$.
There exist sufficiently large absolute constants $C_1,C_2>0$ such that the following holds.
Fix $d\in\Z_{>0}$ and $\eps\in(0,1)$, let $k=\lceil \log(1/\eps)\rceil$ and $T=\eps^{-C_1}$.
Let $D$ be a distribution over $\R^d$ such that for every multi-index $\alpha\in\N^d$ with
$1\le |\alpha|\le k$,
\[
\big|\E_{\x\sim D}[He_{\alpha}(\x)]-\E_{\x\sim \cN^d}[He_{\alpha}(\x)]\big|
\;\le\; \frac{\eps}{d^{C_2k}}.
\]
Then for any halfspace $f(\x)=\sign(\vec v \cdot \x-\tau)$ and $y\in\{\pm1\}$, if
$\Pr_{\x\sim \cN^d}[f(\x)=y]\le C_1/T$, then $\Pr_{\x\sim D}[f(\x)=y]\le \eps$.
\end{fact}

Let $\widetilde{s}:\R^{d-1}\to \R, \widetilde{s}(\z)= \sign(\widetilde{\vec u}\cdot \z -\tau), \widetilde{\vec u}=\vec U^{\top} \vec u$ be the induced $(d-1)$-dimensional representation of $s$. 
Applying \Cref{fact:moment-bias-transfer} to $\widetilde D$ in dimension $d-1$ yields:
if $\Pr_{\bz\sim \cN^{d-1}}[\widetilde{s}(\bz)=1]\le \gamma^{C}$ (for $C$ large enough), then
$\Pr_{\bz\sim \widetilde D}[\widetilde{s}(\bz)=1]\le \gamma$. 

Finally, since $s$ depends only on $\x^{\perp \w}$,
\[
\Pr_{\bz\sim \cN^{d-1}}[\widetilde{s}(\bz)=1]=\Pr_{\x\sim \cN^d}[s(\x)=1\mid \x\in S],
\text{ }
\Pr_{\bz\sim \widetilde D}[\widetilde{s}(\bz)=1]=\Pr_{\x\sim D_\x}[s(\x)=1\mid \x\in S],
\]
which concludes the proof. 
\end{proof}
Bellow we have the claim we deferred from the proof of \Cref{lem:advantage-lemma}.
\begin{claim}[Closeness of sandwiching polynomials with small bias difference]
\label{claim:closeness_bounding_polynomials}
It holds that 
$$\E_{\x\sim \cN^d}[p_+^{(2)}(\x)-p_-^{(1)}(\x) \mid S] \lesssim \alpha\E_{\x\sim \cN^d}[g(\x)\mid S] +\frac{|\vec v\cdot \w|}{\|\vec v^{\perp \w}\|}\,\Delta \;.$$
\end{claim}

\begin{proof}%
Note that $g_1, g_2$ and their respective polynomials depend only on $\x^{\perp \w}$. Therefore, when taking expectation under the Gaussian, it suffices to take expectation over $\x^{\perp \w}$ alone.
We  decompose the error as follows:
\[
\E_{\x^{\perp \w}}\left[ p_{+}^{(2)} - p_{-}^{(1)} \right] = \E_{\x^{\perp \w}}\left[ p_{+}^{(2)} - g_{2} \right] + \E_{\x^{\perp \w}}\left[ g_{2} - g_{1} \right] + \E_{\x^{\perp \w}}\left[ g_{1} - p_{-}^{(1)} \right]
\]We analyze each term separately.
First note that 
$$\E_{\x^{\perp \w}}[g_2(\x^{\perp \w})-g_1(\x^{\perp \w})]\leq \E_{\x^{\perp \w}}[\Ind( t_1\leq \vec b\cdot \x^{\perp \w} \leq t_2)]\leq   \frac{|\vec v\cdot \w|}{\|\vec v^{\perp \w}\|}\,\Delta\;,$$
where we used the anti-concentration property of the standard one-dimensional Gaussian.
Note that by the definition of $p_+^{(2)}$  we have that 
\begin{align*}
\E_{\x^{\perp \w}}[p_{+}^{(2)} - g_{2}]\leq \E_{\x^{\perp \w}}[p_{+}^{(2)} - p_{-}^{(2)}]\leq \alpha\E_{\x^{\perp \w}}[g_2]\;.
\end{align*}
Also we have that
\begin{align*}
\E_{\x\sim \cN^d}[g_2(\x)-g(\x)\mid \x\in  S]\leq \E_{\x^{\perp \w}}[g_2(\x^{\perp \w})-g_1(\x^{\perp \w})]\leq  \frac{|\vec v\cdot \w|}{\|\vec v^{\perp \w}\|}\,\Delta\;.
\end{align*}
Therefore,
$$\E_{\x\sim \cN^{d}}\left[ p_{+}^{(2)}(\x) - g_{2}(\x)\mid \x \in S \right]
\leq \alpha\E_{\x\sim \cN^d}[g(\x)\mid \x\in S]+\frac{|\vec v\cdot \w|}{\|\vec v^{\perp \w}\|}\,\Delta
$$
The same argument holds also for $p_-^{(1)}$. 
This completes the proof of \Cref{claim:closeness_bounding_polynomials}.
\end{proof}
Before continuing with the proof of this result we make the following important remark.
\begin{remark}\label{rem:Generalization}
Note that Algorithm~\ref{alg:testable-massart-gamma} applies several tests (Lines~\ref{line:mass-test}, \ref{line:moment-test}, and \ref{line:negativity-test})
which compare certain expectations to their Gaussian counterparts (and, in the case of the non-negativity test, certify a required expectation inequality).
We say that a distribution satisfies these tests if the corresponding expectations satisfy the required bounds.

Consequently, upon acceptance, these tests are certified for the empirical distribution $\widehat{D}$ used by the tester, and not necessarily for the
population distribution $D$.
In the soundness proof we first use that $\widehat{D}$ satisfies the above properties to show that the returned hypothesis $h$ has nearly-optimal error
with respect to $\widehat{D}$.
We then apply a standard uniform convergence (VC) argument to transfer this guarantee from $\widehat{D}$ to $D$, completing the proof. 
\end{remark}

\subsection{Proof of \Cref{prop:soundness-gamma}}
\begin{proof}%
We follow the plan outlined in Remark~\ref{rem:Generalization}.
Namely, we first work with the empirical distribution $\widehat{D}$ used by the tester and show that, conditioned on acceptance, the returned hypothesis $h$
has nearly-optimal error with respect to $\widehat{D}$.
We then invoke a standard VC/uniform convergence bound to transfer this guarantee from $\widehat{D}$ to the population distribution $D$.

First denote by $\vec w\in \R^d, \|\vec w\|=1$ the defining vector of the returned halfspace $h$.
Denote by $\beta=1-2\eta$ the noise bias, by $l$ the degree, by $\Delta$ the slice width and by $\tau_p$ the slice mass accuracy used by the algorithm (Line \ref{line:init}).  Denote by $S_i, i\in[n+1]$ the slices defined in line \ref{line:slices}.

Consider a competitor $\gamma$-Biased halfspace $f=\sign(\vec v\cdot \x -t), \vec v\in \R^d,t\in \R$ with $\|\vec v\|=1$. 
We aim to prove that 
$\Pr_{(\x,y)\sim \widehat{D}}[h(\x)\neq y]
\le \Pr_{(\x,y)\sim \widehat{D}}[f(\x)\neq y] + \eps$.
Note that the threshold $h$ belongs to at most one slice, which we denote by $S_j$. 
Note that, since our distribution passes the slice mass test (line \ref{line:mass-test}), we have $\pr_{\x\sim \widehat{D}_{\x}}[S_1]+\pr_{\x\sim \widehat{D}_{\x}}[S_j]+\pr_{\x\sim \widehat{D}_{\x}}[S_{n+1}]\le \eps/2$, where $S_1$ and $S_{n+1}$ are the tail slices. Define $\mathcal{S}=\{S_i:i=2,\dots,j-1,j+1,\dots,n\}$ to be the set of all remaining slices defined in line \ref{line:slices}.
Therefore, it suffices to consider the remaining space and show that
$\Pr_{(\x,y)\sim \widehat{D}}[h(\x)\neq y, \x\in \bigcup_{S\in \cS} S]
\le \Pr_{(\x,y)\sim \widehat{D}}[f(\x)\neq y, \x\in \bigcup_{S\in \cS} S] + \eps/2$.

We split the proof into three cases according to the angle between the defining vectors of  $f$ and $h$. 
First assume that $\vec v^{\perp \w}\neq \vec 0$.
We may write $\vec v$ as $\vec v=\w \cos \theta  +\vec b\sin \theta $, where $\vec b$ unit vector such that  $\vec b\perp \w$.
Then $f(\x)= \sign(\vec b \cdot  \x - (t- \w \cdot \x \cos\theta)/\sin \theta  )$. 
For notational simplicity also define the per fiber threshold of $f$,  $s(z)\eqdef (t- z \cos\theta)/\sin \theta, z\in \R$.

\begin{lemma}[Closer to perpendicular] \label{lem:closer2perpendicular}
If $|\vec v \cdot \w |\leq 1/2$ and $\|\vec v^{\perp \w}\|\geq \eps^2$, then  $\Pr_{(\x,y)\sim \widehat{D}}[h(\x)\neq y]<  \Pr_{(\x,y)\sim \widehat{D}}[f(\x)\neq y]+\eps$.
\end{lemma}
\begin{proof}%
Note that for $\w \cdot \x = 0$, the threshold of $f$ is $s(0)=t/\sin\theta$. Since $|\cos\theta|\le 1/2$ implies $|\sin\theta|=\sqrt{1-\cos^2\theta}=\Omega(1)$, we obtain $|s(0)|=\frac{|t|}{|\sin\theta|}=O(t)$.
Thus, on the fiber $\w\cdot \x=\Delta t$, the threshold satisfies $|s(\Delta t)|=O(t)+|\Delta t|\cos\theta/\sin\theta$. Hence, for any $|\Delta t|=O(1)$ we have $|s(\Delta t)|=O(t)+O(1)$.

Denote by $\mathcal{S}_1$ the set of slices $S\in\mathcal{S}$ that intersect the region $|\w\cdot\x|\le \Delta t$.
Since the slice width is bounded by a constant, we have that for all $\w\cdot\x=z$ within these slices, $|s(z)|=O(t)+O(1)$.
Moreover, since $f$ is defined to have $\gamma$-bias under the Gaussian, we have $|t|\le \sqrt{\log(1/\gamma)}$.
Hence, for each slice $S\in\mathcal{S}_1$ we have $\pr_{\x\sim \cN^d}[f(\x)\neq h(\x)\mid \x\in S]\ge \gamma^{C_1}$ for a sufficiently large universal constant $C_1$, since $|s(\w\cdot\x)|=O(t)+O(1)=O(\sqrt{\log(1/\gamma)})$ (because $\gamma<1/2$) and $h(\x)$ is constant for all $\x\in S$ whenever $S\in\mathcal{S}_1$.
Also, since we assumed that $\|\vec v^{\perp \w}\|\ge \eps^2$, by the choice of slice width we have
$|\vec v\cdot \w|\Delta/\|\vec v^{\perp \w}\|\le (1-2\eta)\gamma^{C_1}/C_1$.
Therefore, since the algorithm accepts $h$ for a big enough degree parameter $l$, by \Cref{lem:advantage-lemma} it follows that for every $S\in\mathcal{S}_1$,
\[
\E_{(\x,y)\sim \widehat{D}}[h(\x)\neq y\mid \x\in S]-\E_{(\x,y)\sim \widehat{D}}[f(\x)\neq y\mid \x\in S]\le -(1-2\eta)\gamma^{C_1}.
\]

Note that since the union of all sets $S\in \cS_1$ forms a slice of constant width, say $C'$, across the direction $\w$, we have
$\pr_{\x\sim \cN^d}\!\left[\bigcup_{S\in \cS_1} S\right]=\Omega(1)$.
As a result, since we pass the slice mass test (Line \ref{line:mass-test}), the total mass of the slices in $\cS_1$ satisfies
\begin{align*}
    \sum_{S\in \cS_1} \pr[S]\ge \Omega(1)-\lceil C'/\Delta\rceil \tau_p = \Omega(1)\;,
\end{align*}
where we used the fact that $\tau_p<\Delta/C$ for a sufficiently large constant $C$.

As a result, the total advantage is garnered from slices in $\cS_1$ is
$$\sum_{S\in \cS_1}\left(\E_{(\x,y)\sim \widehat{D}}[h(\x)\neq y,\x\in S ]-\E_{(\x,y)\sim \widehat{D}}[f(\x)\neq y,\x\in S]\right)\leq -(1-2\eta)\gamma^{C_1}\;.$$

Since we choose the degree parameter $l=C\log^3(1/\gamma)\log^2(1/\beta)/\beta^2$ for a sufficiently large constant $C$, \Cref{lem:advantage-lemma}
applies to slices whose Gaussian-bias is at least $\gamma^{C_1'}$, for some constant $C_1'$ (as all other conditions of the lemma hold as before). Hence, for every slice $S$ with Gaussian-bias at least $\gamma^{C_1'}$, we have $\E_{(\x,y)\sim \widehat{D}}[h(\x)\neq y,\ \x\in S]-\E_{(\x,y)\sim \widehat{D}}[f(\x)\neq y,\ \x\in S]\leq 0$.

Let $\cS_2$ denote the set of slices with Gaussian-bias at most $\gamma^{C_1'}$. By \Cref{cl:probability-transfer}, for each such slice $S$ we have $\pr_{(\x,y)\sim \widehat{D}}[h(\x)\neq f(\x)\mid \x\in S]\le \gamma^{C_2}$ for some constant $C_2$ to be choosen later. Hence, for every $S\in \cS_2$,
$\E_{(\x,y)\sim \widehat{D}}[h(\x)\neq y,\ \x\in S]-\E_{(\x,y)\sim \widehat{D}}[f(\x)\neq y,\ \x\in S]\le \gamma^{C_2}$.

Now combining the above gives us that 
\begin{align*}
    \Pr_{(\x,y)\sim \widehat{D}}[h(\x)\neq y, \x\in \bigcup_{S\in \cS} S]
\le \Pr_{(\x,y)\sim \widehat{D}}[f(\x)\neq y, \x\in \bigcup_{S\in \cS} S] - (1-2\eta)\gamma^{C_1}+\gamma^{C_2}
\end{align*}
Since $\gamma<1-2\eta$, as set in line \ref{line:init-inputs}, and $1/2\geq \gamma\geq \eps$ for a sufficient difference between $C_2-C_1$ we have that 
\begin{align*}
    \Pr_{(\x,y)\sim \widehat{D}}[h(\x)\neq y, \x\in \bigcup_{S\in \cS} S]
\le \Pr_{(\x,y)\sim \widehat{D}}[f(\x)\neq y, \x\in \bigcup_{S\in \cS} S] - \eps^{C'}\;.
\end{align*}
for a constant $C'>1$, 
which concludes the proof of \Cref{lem:closer2perpendicular}.
\end{proof}
\begin{lemma}[Closer to parallel] \label{lem:closer2parallel}
If $|\vec v \cdot \w |\geq 1/2$ and $\|\vec v^{\perp \w}\|\geq \eps^2$, then  $\Pr_{(\x,y)\sim \widehat{D}}[h(\x)\neq y]<  \Pr_{(\x,y)\sim \widehat{D}}[f(\x)\neq y]+\eps$.
\end{lemma}
\begin{proof}%
Consider the point $z_0\in \R$ at which the competitor $f$ is unbiased, i.e., $s(z_0)=0$, or equivalently $z_0=t/\cos\theta$.
Note that at the point $z_0+\Delta t$ the bias is $s(z_0+\Delta t)=-\Delta t \cos\theta/\sin\theta$. Hence, for any $|\Delta t|=O(\tan\theta)$ we have $|s(z_0+\Delta t)|=O(1)$.

Denote by $\mathcal{S}_1$ the set of slices $S\in \mathcal{S}$ that intersect the region $|\w\cdot\x-z_0|\le \Delta t$.
Since the slice width is $O(\eps^2)=O(\tan\theta)$, we have that for all $\w\cdot\x=z$ within these slices, $|s(z)|=O(1)$.
Therefore, for each slice $S\in \mathcal{S}_1$ we have $\pr_{\x\sim \cN^d}[f(\x)\neq h(\x)\mid \x\in S]=\Omega(1)$, since $|s(\w\cdot\x)|=O(1)$ and $h(\x)$ is constant for all $\x\in S$ whenever $S\in \mathcal{S}_1$.
Also, since we assumed that $\|\vec v^{\perp \w}\|\ge \eps^2$, by the choice of slice width we have $|\vec v\cdot \w|\Delta/\|\vec v^{\perp \w}\|\le 1/C$ for a sufficiently large constant $C$.
Therefore, since the algorithm accepts $h$, by \Cref{lem:advantage-lemma} it follows that for every $S\in \mathcal{S}_1$,
$\E_{(\x,y)\sim \widehat{D}}[h(\x)\neq y\mid \x\in S]-\E_{(\x,y)\sim \widehat{D}}[f(\x)\neq y\mid \x\in S]\lesssim -(1-2\eta)$.

Note that $z_0+\Delta t\lesssim t/\cos\theta+\tan\theta\lesssim t+O(1)$ for all $\x\in \bigcup_{S\in \cS_1} S$.
Hence, we have $\pr_{\x\sim \cN^d}[f(\x)=1\mid \w\cdot\x=z]\ge \gamma^C$ for a sufficiently large constant $C>0$, for all $\w\cdot\x=z$ with $\x\in \bigcup_{S\in \cS_1} S$.
As a result, $\pr_{\x\sim \cN^d}\!\left[\bigcup_{S\in \cS_1} S\right]\ge \gamma^C\,\Delta t$.
Since we pass the slice mass test (Line \ref{line:mass-test}), the total mass of the slices in $\cS_1$ satisfies
\begin{align*}
    \sum_{S\in \cS_1} \pr[S]
    &\ge \gamma^C\,\Delta t-\left\lceil \Delta t/\Delta\right\rceil \tau_p
    \gtrsim \gamma^C\,\Delta t\;,
\end{align*}
where we used the fact that $\tau_p/\Delta\leq \eps^{C+1}\leq \gamma^C/2$ as chosen in Line \ref{line:init}.

As a result, the total advantage garnered from slices in $\cS_1$ is
$$
\sum_{S\in \cS_1}\left(\E_{(\x,y)\sim \widehat{D}}[h(\x)\neq y,\ \x\in S]-\E_{(\x,y)\sim \widehat{D}}[f(\x)\neq y,\ \x\in S]\right)
\le -(1-2\eta)\gamma^C\tan\theta\;.
$$

Now it suffices to show that the total disagreement a subset of  $\cS \setminus \cS_1$, excluding slices that $\E_{(\x,y)\sim \widehat{D}}[h(\x)\neq y,\ \x\in S]-\E_{(\x,y)\sim \widehat{D}}[f(\x)\neq y,\ \x\in S]$ is negative, is  $O(\gamma^{C'}\tan\theta)$ for some $C'\geq C$ with $C'-C$ sufficiently large. 
Note that $\cS \setminus \cS_1$ decomposes to the slices on the left and the right side of $\cS_1$. 
Let us consider slices on the right side and denote this set of slices by $\cS_2$ (our argument will work also for the slices on the left of $\cS_1$). 

Recall that $s(z_0+\Delta t)=-\Delta t/\tan \theta$.
Hence, as $\w\cdot \x$ increases we have that if $\tan \theta\geq 0$ the bias decreases and $\pr_{\x \sim \cN^d}[f(\x)=1\mid \w\cdot \x=z]$ increases.
Hence, we may assume that $h(\x)=1$, for all $S\in \cS_2$, since we may exclude all slices with greater than $\gamma^C$ disagreement between $h$ and $f$ since they imply  negative $\pr_{(\x,y)\sim \widehat{D}}[h(\x)\neq y, \x\in S]-\pr_{(\x,y)\sim \widehat{D}}[f(\x)\neq y, \x\in S]$.
Therefore, it suffices to bound the total $\pr_{(\x,y)\sim \widehat{D}}[f(\x)=-1, \x\in S]$ over the subset of the slices $S\in \cS_2$ such that $\pr_{(\x,y)\sim \widehat{D}}[f(\x)=-1| \x\in S]\leq \gamma^{C'}$.
We define a subset of slices in $\cS_{2}$, $\widetilde{\cS}_2$ as follows.
Define by $t_{\mathrm{th}}$ a threshold point such that $\x\in \bigcup_{S\in \widetilde{\cS}_2} S$ iff  $s(\w\cdot \x) \leq t_{\mathrm{th}}$ and by $s_{\mathrm{th}}$ a point such that
$\x\in \bigcup_{S\in \widetilde{\cS}_2} S$ iff  $\w\cdot \x \geq s_{\mathrm{th}}$.

We denote the slices  $\widetilde{\cS}_2= \{ [s_{\mathrm{th}}+i\Delta,s_{\mathrm{th}}+(i+1)\Delta]: i \in \Z_+\}$. 
Note that for $\widetilde{\cS}_2$ we have assumed that our slices extend as $\w\cdot \x $ tends to infinity, this is not the case, however we count additional disagreement for tail slices and we are only interested in an upper bound. 
For  $\w\cdot \x=z$, we have  $f(\x)=-1$ if and only if $\vec b \cdot \x\leq (t - z \cos \theta)/\sin \theta= t_{\mathrm{th}}- z/\tan \theta$. Note conditioning on $S_i\eqdef [s+i\Delta,s+(i+1)\Delta]$  a superset of the region $f(\x)=-1$  is the region $\vec b \cdot \x \leq \max_{\x \in S_i} (t_{\mathrm{th}}- \w\cdot \x/\tan \theta)\leq    t_{\mathrm{th}}- i\Delta/\tan \theta$.
Denote by $t_i\eqdef t_{\mathrm{th}}- i\Delta/\tan \theta$.
By Markov's inequality,
\[
\Pr_{\x \sim \widehat{D}_{\x}}[\vec b\cdot \x\le t_i\mid \x\in S_i]
\le \frac{\E_{\x \sim \widehat{D}_{\x}}[(\vec b\cdot \x)^k\mid \x\in S_i]}{(t_i)^k}.
\]
Now since we pass the Orthogonal moment matching test (Line \ref{line:moment-test}), i.e., the  $k\leq l$-th moments in each slice match those of $\cN^d$ for directions orthogonal to $\w$ by \Cref{cl:orth-moment-transfer-poly},
we have $\E_{\x \sim \widehat{D}_{\x}}[(\vec b\cdot \x)^k\mid \x\in S_i]=\E_{\cN^d}[(\vec b\cdot \x)^k\mid \x\in S_i]+\tau_{m}d^k\E_{\cN^d}[(\vec b\cdot \x)^{2k}\mid \x\in S_i]\lesssim \E_{z\sim \cN}[z^k]\lesssim k^{k/2}$. 
Where we used the fact that $(\vec b \cdot \x)^k$ has at most $d^k$ terms of the form $\x^{\alpha}$ for a $\alpha\in \N^{d}$ with $|\alpha|\leq k$ each one with coefficient less than $1$ since $\vec b$ is a unit vector. 
Moreover, note that by the Slice mass test (Line \ref{line:mass-test}) we have that $\pr_{\x\sim \widehat{D}_{\x}}[\x\in S_i]\leq \pr_{\x\sim \cN^d}[\x\in S_i]+\tau_{p} \lesssim \Delta $.

Therefore,
\begin{align*}
  \sum_{i=0}^{\infty} \Pr_{\x \sim \widehat{D}_{\x}}[f(\x)\neq h(\x), \x\in S_i]  \leq \sum_{i=0}^{\infty} \Pr_{\x \sim \widehat{D}_{\x}}[\vec b\cdot \x\le t_i, \x\in S_i]\lesssim k^{k/2}\Delta\sum_{i=0}^{\infty} \frac{1}{(t_{\mathrm{th}}-i\Delta/\tan \theta)^k } \;.
\end{align*}
Let $a \eqdef \Delta/\tan\theta$ and assume that
$t_{\mathrm{th}}= -T$ for some $T>0$. Then, for any $k>1$ even, we have 
\begin{align*}
\sum_{i=0}^{\infty} \frac{1}{(t_{\mathrm{th}}-i\Delta/\tan \theta)^k }
&= \sum_{i=0}^{\infty} \frac{1}{(T+ia)^k}
\leq \frac{1}{T^k}+\int_{0}^{\infty}\frac{dx}{(T+ax)^k}  \\
&= \frac{1}{T^k}+\frac{1}{a}\int_{T}^{\infty}u^{-k}\,du
\;=\; \frac{1}{T^k}+\frac{1}{a(k-1)T^{k-1}}
\;\lesssim\; \frac{1}{a\,T^{k-1}} \;,
\end{align*}
where we used the fact that $1/(T+xa)$ is monotonically decreasing for $x\geq 0$ and the substitution $u=T+ax$.
In the last inequality we used that $t_{\mathrm{th}}$ will be chosen such that $T\geq (k-1)a$ because $a\leq \eps^2$ since $\Delta\leq \eps^{C}$ for a sufficiently large constant $C\geq 0$.
Plugging back $a=\Delta/\tan\theta$ gives
\[
 \sum_{i=0}^{\infty} \Pr_{\x \sim \widehat{D}_{\x}}[f(\x)\neq h(\x), \x\in S_i]
\lesssim \frac{k^{k/2}\tan \theta}{T^{k-1}} \;.
\]
We can choose $t_{\mathrm{th}}=-C'\sqrt{\log (1/\gamma)}$ since then $\pr_{\x\sim \cN^d}[f(\x)\neq i|\mid \w \cdot \x=z]\geq \gamma^{C'}$ for all $i\in \{\pm 1\}$ for any $z\leq s_{\mathrm{th}}$ which would certify and advantage for this region. 
Choosing $T=C'\sqrt{\log (1/\gamma)}$ and $k=C'\log(1/\gamma)$  for a 
sufficiently large constant, we have  $\sqrt{k}/T\leq  1/C'$.  Hence, 
\[
\frac{k^{k/2}\tan\theta}{T^{k-1}} \;=\; \tan\theta\cdot T\Big(\frac{\sqrt{k}}{T}\Big)^k \;\le\; \tan\theta\cdot T\cdot (1/C')^k
\;\le\; \tan\theta\cdot \gamma^{C''}\,,
\]
for a sufficiently large constant $C''\ge C'$, since $k=C'\log(1/\gamma)$ implies $(1/C')^k=\gamma^{\Omega(C'\log C')}$ and $T=\mathrm{polylog}(1/\gamma)$ can be absorbed into the exponent.
Choosing $C''$ to be bigger than $C$, since $\gamma\leq 1-2\eta$ as set in 
Line \ref{line:init-inputs}, completes the proof of \Cref{lem:closer2parallel}.
\end{proof}

\begin{lemma}[Approximately parallel]
\label{lem:near-parallel}
It holds that if $\|\vec v^{\perp \w}\|\leq \eps^2$, then  $\Pr_{(\x,y)\sim \widehat{D}}[h(\x)\neq y]<  \Pr_{(\x,y)\sim \widehat{D}}[f(\x)\neq y]+\eps$.
\end{lemma}

\begin{proof}%
First, we show that for the standard normal distribution for all slices except those very close to the decision boundary of $f$, the classifier $f$ either agrees with $h$ on almost all points in the slice, or disagrees with $h$ throughout the slice. Note that if $\vec v^{\perp \w}=\vec 0$, this holds exactly for all slices.
Assume that $\vec v^{\perp \w}\neq \vec 0$. 
Consider the bias of $f$ conditioned on $\w\cdot \x=z$. If $|t-(\vec v\cdot \w)z|\ge \eps$, then
$$
|s(z)|
=\frac{|t-(\vec v\cdot \w)z|}{\|\vec v^{\perp \w}\|}
\ge \frac{\eps}{\|\vec v^{\perp \w}\|}
\ge \frac{1}{\eps}\;.
$$
Fix a slice $S\in \mathcal{S}$ that does not intersect the region $R\eqdef\{\,\x : |t-(\vec v\cdot \w)(\w\cdot \x)|\le \eps\,\}$ and without loss of generality assume that $s(\w\cdot \x)\geq 1/\eps$ for all $\x\in S$ (same argument holds if $s(\w\cdot \x )\leq -1/\eps$ for all $\x \in S$). 
Let $\vec b\eqdef \vec v^{\perp \w}/\|\vec v^{\perp \w}\|$ and define
$t_{\min} \eqdef \min_{\x\in S} s(\w\cdot \x)$.
Define
$f'(\x)\eqdef \sign(\vec b\cdot \x^{\perp \w}-t_{\min})$.
Then for every $\x\in S$, if $f(\x)=1$ we have $\vec b\cdot \x^{\perp \w}\ge s(\w\cdot \x)\ge t_{\min}$, and hence $f'(\x)=1$.
Therefore $\{\x\in S: f(\x)=1\}\subseteq \{\x\in S: f'(\x)=1\}$, and so
$\Pr_{\x\sim \widehat{D}_{\x}}[f(\x)=1\mid \x\in S]\le \Pr_{\x\sim \widehat{D}_{\x}}[f'(\x)=1\mid \x\in S]$. 
Now since we pass the moment-matching test by Markov's inequality we have that for all $1\leq k\leq l$
\begin{align*}
    \pr_{\x\sim \widehat{D}_{\x}}[\vec b\cdot  \x\geq t_{\min}\mid \x \in S] \leq \E_{\x\sim \widehat{D}_{\x}}[(\vec b\cdot  \x)^k \mid \x\in S]/t_{\min}^k \leq \eps^k(\E_{\x\sim \cN^d}[(\vec b\cdot  \x)^k \mid \x\in S] + \eps)\;,
\end{align*}
where we used \Cref{cl:orth-moment-transfer-poly}, the fact that $\E_{x\sim \cN}[x^{k}]\leq k^{k/2}$ and that our accuracy $\tau_{m}\leq \eps^C/d^{Cl}$ for a sufficiently large constant $C$.
Which for $k=2$ and $\eps$ less than a sufficiently small constant is less than $\eps$.
Therefore, since $h$ is constant within every $S\in \mathcal{S}$ we have shown that for all slices that do not intersect the region $R\eqdef\{\,\x : |t-(\vec v\cdot \w)(\w\cdot \x)|\le \eps\,\}$ either 
$\pr_{\x\sim \widehat{D}_{\x}}[h(\x)\neq f(\x)\mid \x \in S]\leq \eps$ or $\pr_{\x\sim \widehat{D}_{\x}}[h(\x)\neq f(\x)\mid \x\in S]\geq 1-\eps$.

We denote: by $\mathcal{S}_1$ the set of slices that  intersect $R$,
by $\mathcal{S}_2$ the set of slices that do not intersect $R$ and $\pr_{\x\sim \widehat{D}_{\x}}[h(\x)\neq f(\x)\mid \x \in S]\leq \eps$, 
and by $\mathcal{S}_3$ the set of slices that do not intersect $R$ and $\pr_{\x\sim \widehat{D}_{\x}}[h(\x)\neq f(\x)\mid \x \in S]\geq 1-\eps$.

First, we show that the total mass of all slices in $\mathcal{S}_1$ is small. 
Note that the region $R$ forms interval in the direction $\w$ and that length at most $2\eps/|\vec v\cdot \w|\leq 4\eps$ along the $\w$-direction, since $|\vec v \cdot \w|\geq \sqrt{1-\|\vec v^{\perp \w}\|^2}\geq 1/2$.
Moreover, note that $|\mathcal{S}_1|\leq \lceil 4\eps/\Delta\rceil$. 
Therefore, since we pass the slice-mass test (Line \ref{line:mass-test}) we have
\begin{align*}
\pr_{\x\sim \widehat{D}_{\x}}[\bigcup_{S\in \cS_{1}} S] \leq 4\eps+\tau_{p} \lceil 4\eps/\Delta\rceil +2\Delta\lesssim \eps. 
\end{align*}
Hence, $\sum_{S\in S_1}( \p_{(\x,y)\sim \widehat{D}}[h(\x)\neq y,\x\in S]-
\p_{(\x,y)\sim \widehat{D}}[f(\x)\neq y,\x\in S])\lesssim \eps$.

Second, note that for every $S\in {\cS}_2$ we have that 
$\p_{(\x,y)\sim \widehat{D}}[h(\x)\neq y,\x\in S]-
\p_{(\x,y)\sim \widehat{D}}[f(\x)\neq y,\x\in S]\leq \p_{(\x,y)\sim \widehat{D}}[h(\x)\neq f(\x),\x\in S]\leq \eps$.

Third, we show that in each slice in $\mathcal{S}_3$ the non-negativity test helps us certify an advantage of choosing  $h$ v.s. $f$. 
Fix a slice $S\in \mathcal{S}_3$.
Since $S$ passes the degree-$0$ instantiation of the non-negativity test, we have
$$
\E_{(\x,y)\sim \widehat{D}}\big[\,y h(\x)-(1-2\eta)+\eps \ \bigm|\ \x\in S\,\big]\ \ge\ -\eps\;.
$$
Rearranging gives $\pr_{(\x,y)\sim \widehat{D}}[y=h(\x)\mid \x\in S]\ge \pr_{(\x,y)\sim \widehat{D}}[y\neq h(\x)\mid \x\in S] +(1-2\eta)-2\eps$.

As a result, for all $S\in \cS_3$ we have that 
$\pr_{(\x,y)\sim \widehat{D}}[y=h(\x)\mid \x\in S]\ge \pr_{(\x,y)\sim \widehat{D}}[y= f(\x)\mid \x\in S] +(1-2\eta)-2\eps$.
This holds since in our initialization step (Line \ref{line:init-inputs}) 
we set $\eps$ sufficiently small such that $(1-2\eta)-2\eps\geq 0$, 
and therefore $h$ is better than $f$ for all  $S\in \cS_3$.

Finally, combining the above we have that 
\begin{align*}
&\p_{(\x,y)\sim \widehat{D}}[h(\x)\neq y]-
\p_{(\x,y)\sim \widehat{D}}[f(\x)\neq y]\lesssim \eps\;.
\end{align*}
Now, if we choose the constants in the algorithm sufficiently large so that $\eps$ is smaller (multiplicatively) by a constant satisfying the above inequality, then we have
$\p_{(\x,y)\sim \widehat{D}}[h(\x)\neq y]-
\p_{(\x,y)\sim \widehat{D}}[f(\x)\neq y]\leq \eps$,
which concludes the proof of \Cref{lem:near-parallel}.
\end{proof}
Finally, we are ready to prove the near-optimality of $h$ with respect to $D$ with constant probability.
For a distribution $Q$ over $\R^d\times\{\pm1\}$, denote by $\opt_{\gamma}(Q)\eqdef \min_{f\in \cH_{d,\gamma}} \pr_{(\x,y)\sim Q}[f(\x)\neq y]$.
By a standard uniform convergence bound for halfspaces (which have VC dimension $d+1$) since $N\geq C(d+1)/\eps^2$ for a sufficiently large constant $C>0$, with constant probability
(e.g.\ at least $2/3$) we have simultaneously for all halfspaces $g$ that
$|\err_{D}(g)-\err_{\widehat{D}}(g)|\le \eps/20$.
In particular, $\opt_{\gamma}(D)\ge \opt_{\gamma}(\widehat{D})-\eps/20$ and
$\err_D(h)\le \err_{\widehat{D}}(h)+\eps/20$.
Moreover, note that we have already shown that upon acceptance  $\err_{\widehat{D}}(h)\le \opt_{\gamma}(\widehat{D})+9\epsilon/10$.
Denote the event that generalization fails, i.e., there exists a halfspace $g$ such that $|\err_{D}(g)-\err_{\widehat{D}}(g)|> \eps/20$
by $G$.
Denote by $A$ the event that \Cref{alg:testable-massart-gamma} accepts and by $B$ the event that 
$\err_{D}(h)> \opt_{\gamma}(D)+\epsilon$. 
Note that conditional $A$ and $\text{not } G$ we have that 
$$\err_{D}(h)\leq \err_{\widehat{D}}(h)+\eps/20\leq \opt_{\gamma}(\widehat{D})+9\eps/10+\eps/20\leq \opt_{\gamma}(D)+\eps\;,$$
thus $B$ can not happen.
Therefore,  by the triangle inequality, we obtain 
\begin{align*}
 \pr[A,B]\leq \pr[G]+\pr[A,B,\text{not }   G]=\pr[G]\leq 1/3\;,    
\end{align*}
which concludes the proof of \Cref{prop:soundness-gamma}.
\end{proof}

\section{Multiplicative Approximation Sandwiching Polynomials} \label{sec:sandwiching-result}

In this section, we prove the polynomial sandwiching result that is used in our proof of soundness.
In essence, we show that polynomials of degree $O(t^6)$ suffice to sandwich a halfspace
with threshold at most $t$ (in absolute value) up to a constant multiplicative error.

The intuition behind our proof is as follows.
In one dimension, a halfspace is simply a step function, which can be viewed as the
integral of a Dirac $\delta$ function. This suggests the following strategy:
first approximate the $\delta$ function by a smooth bump function; then raise this bump
function to a high power to sharpen it and make it more closely resemble a $\delta$
function; finally, integrate this approximation to obtain an approximation to the step function.

Consequently for our bump function approximation we exploit a small modification to Chebyshev polynomials of odd degree already behaves
like bump functions, so we use them as the building blocks for our approximation.

\begin{figure}[h]
  \centering

\setlength{\abovecaptionskip}{0pt}
\setlength{\belowcaptionskip}{0pt}

  \begin{minipage}[t]{0.49\linewidth}
    \centering
    \includegraphics[width=\linewidth]{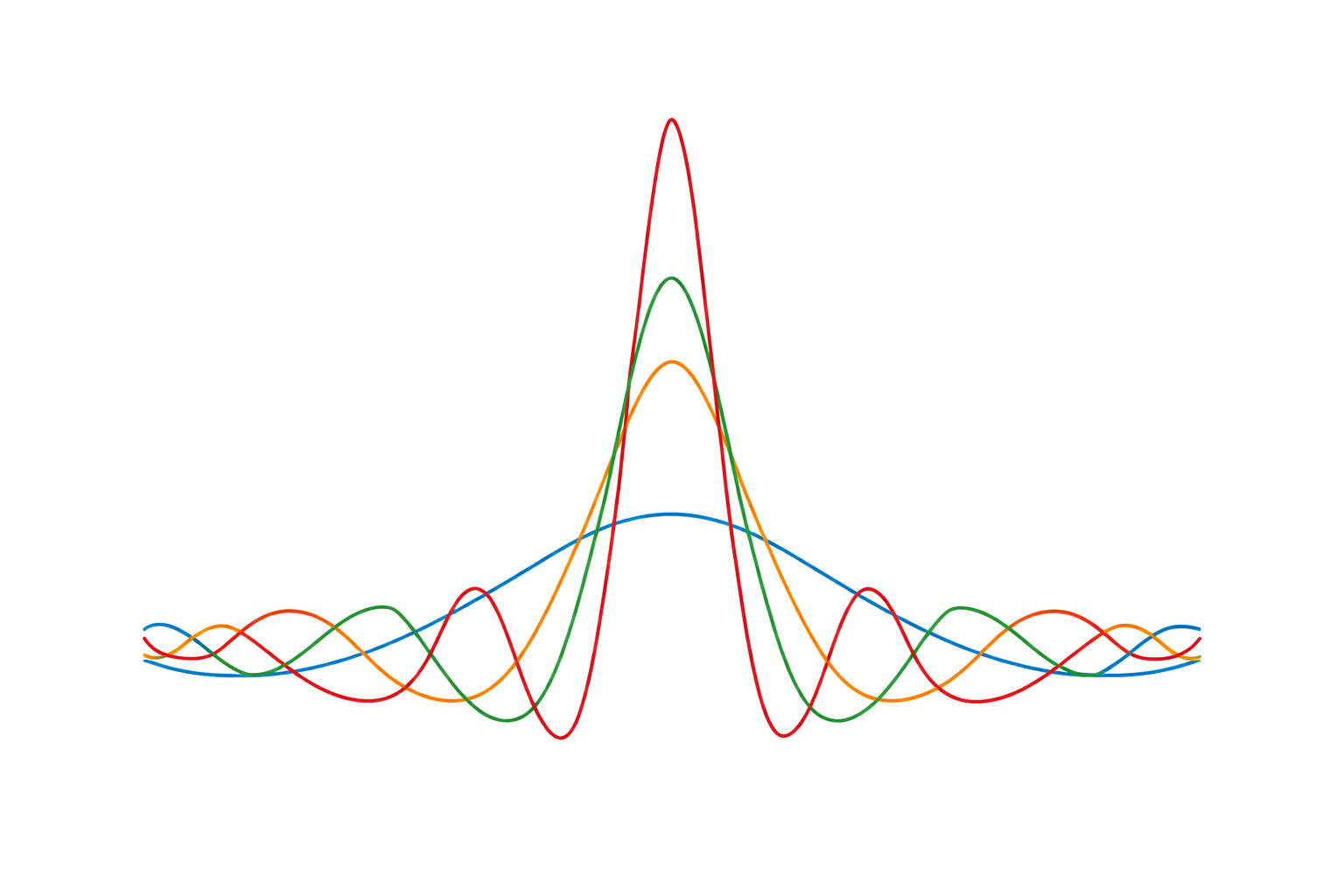}
  \end{minipage}\hfill
  \begin{minipage}[t]{0.49\linewidth}
    \centering
    \includegraphics[width=\linewidth]{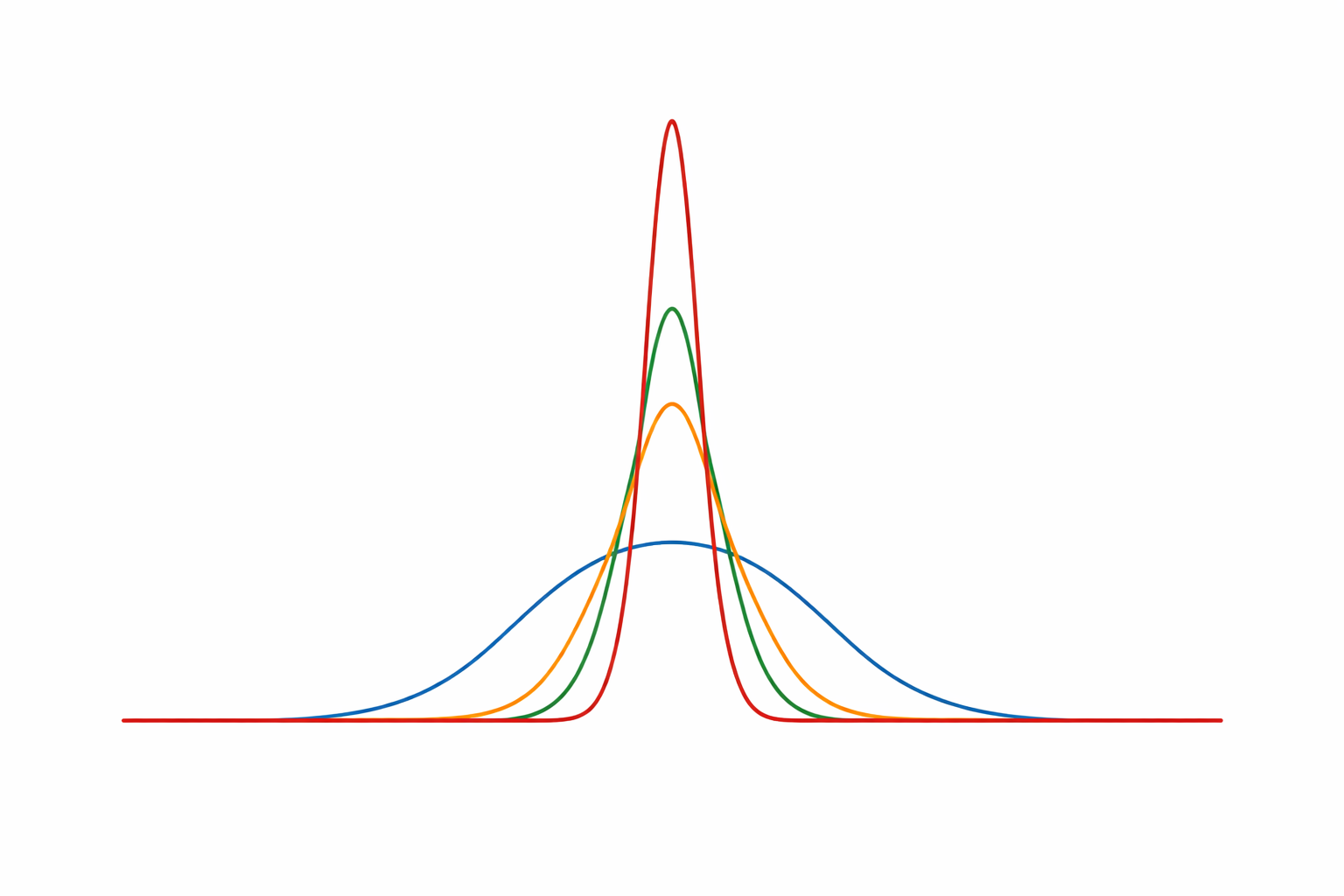}
  \end{minipage}
  \vspace{-1em}
  \includegraphics[width=0.49\linewidth]{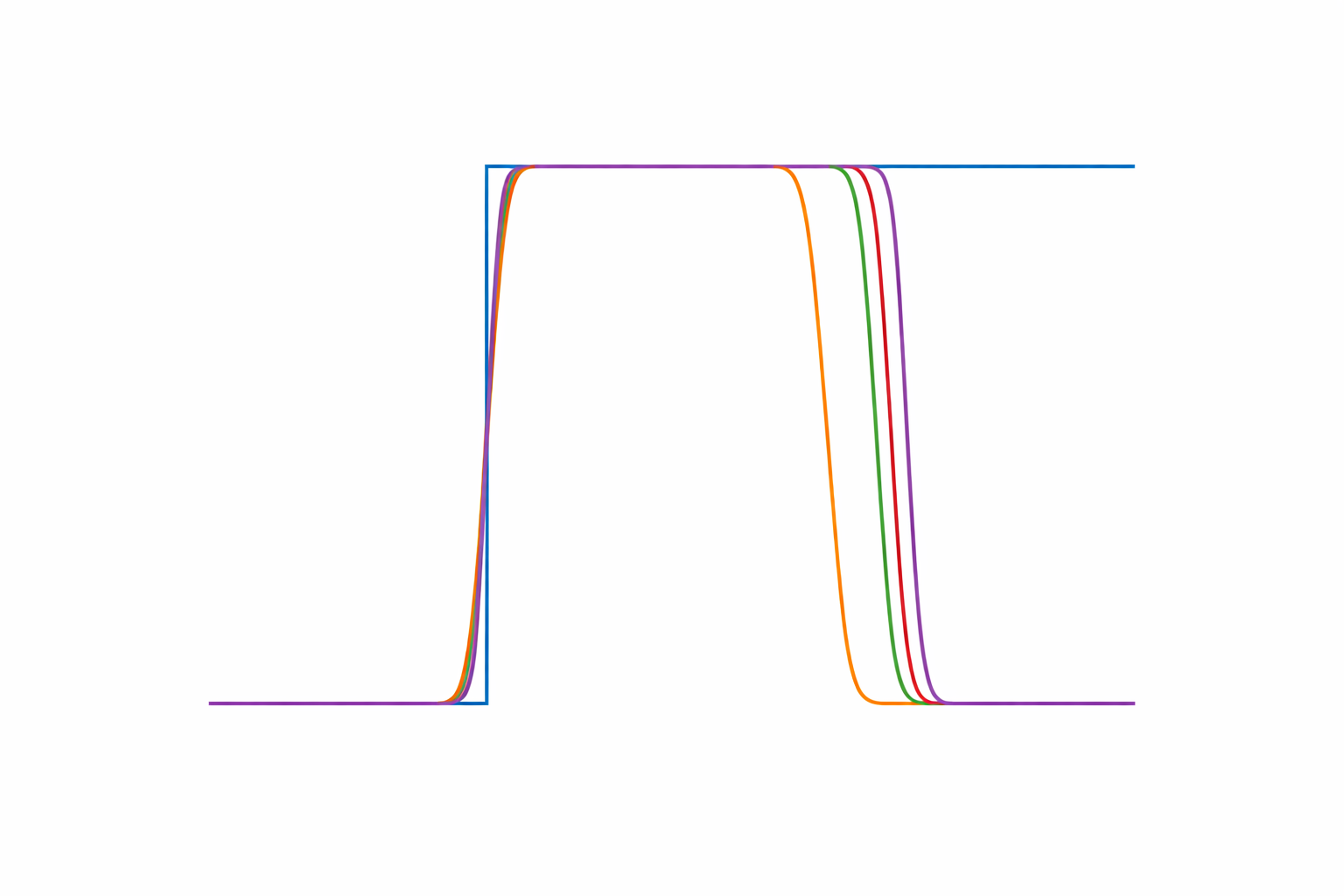}
\vspace{-0.6em}
\caption{This figure shows our approximation procedure. \textbf{Top left:} $\dfrac{T_m(x)}{mx}$; we see that it is a polynomial with a bump-function shape. \textbf{Top right:} $f(x)=\left(\dfrac{T_m(x)}{mx}\right)^k$ (\Cref{eq:f-definition}); as the power $k$ increases, $f$ becomes increasingly concentrated and approximates $\delta$. \textbf{Bottom:} the integral $p$ of $f$ over a sliding window (\Cref{eq:p-definition}), which approximates a step function.
The sliding-window length is chosen according to the desired accuracy; consequently, our approximation is forced to drop to $0$ outside the window. In all examples, the plotted curves grow polynomially to infinity beyond the figure, but this effect is controlled because the Gaussian tails dominate.
}

  \label{fig:three_images}
\end{figure}

\MultSandwichSign*

\begin{proof}
First, we may assume without loss of generality that \(t \ge 0\). If \(t < 0\), let \(s = -t > 0\) and note that \(\Ind(x \ge t) = \Ind(x \ge -s) = 1 - \Ind(-x \ge s)\). Thus, given sandwiching polynomials \(p_-, p_+\) for \(\Ind(x \ge s)\), i.e., \(p_-(x) \le \Ind(x \ge s) \le p_+(x)\) for all \(x \in \mathbb{R}\), define \(q_-(x) = 1 - p_+(-x)\) and \(q_+(x) = 1 - p_-(-x)\); then for all \(x \in \mathbb{R}\) we have \(q_-(x) \le 1 - \Ind(-x \ge s) = \Ind(x \ge t) \le q_+(x)\). Moreover, by symmetry of the standard normal for ${x \sim \mathcal{N}(0,1)}$,  we have \(\E[q_+(x) - q_-(x)] = \E[p_+(x) - p_-(x)]\) and \( \Pr[x \ge t] = \Pr[x \le s] = 1 - \Pr[x \ge s] \ge \Pr[x \ge s] \). Hence any guarantee \(\E[p_+(x) - p_-(x)] \le \alpha \E[\Ind(x \ge s)]\) implies \(\E[q_+(x) - q_-(x)] \le \alpha \E[\Ind(x \ge t)]\). Therefore it suffices to prove the theorem for \(t \ge 0\).

Let $T_m$ denote the $m$'th Chebyshev polynomial.
Fix an odd integer $m\ge 1$  and an even integer
$k\ge 2$ (to be chosen later).  Define
\begin{equation}
f(x) \;\eqdef\; 
\begin{cases}
\Big(\dfrac{T_m(x)}{m x}\Big)^k, & x\neq 0,\\[6pt]
1, & x=0\;.
\end{cases}    \label{eq:f-definition}
\end{equation}

Note that $f$ is a polynomial of degree at most $(m-1)k$, since for $m$ odd $T_m$ is a polynomial with only odd coefficients, \Cref{fact:Chebychev_properties}.

We show the following standard properties of $f$ that will aid us in proving the sandwiching result. 
These properties essentially show that $f$ looks like a bump function.  

\begin{claim}[Properties of $f$]
\label{clm:f-properties}
There exist absolute constants $c_0,c_1>0$ such that:
\begin{enumerate}[label=(\roman*)]
\item For all $|x|\le 1$, 
      \(
      0\le f(x)\le 1.
      \)
\item For all $|x|\le c_0/(mk)$, 
      \(
      \tfrac12 \le f(x)\le 1.
      \)
\item For all $c_1/m\le |x|\le 1$, 
      \(
      0\le f(x)\le 2^{-k}.
      \)
\item For all $|x|\ge 1$,
      \(
      |f(x)|\le (2|x|)^{k(m-1)}.
      \)
\end{enumerate}
\end{claim}
We refer the reader to the \Cref{sec:omited-sandwiching} for the proof of \Cref{clm:f-properties}.

Now define the normalized bump
\begin{equation}\label{eq:g-definition}
g(x) \;\eqdef\; {f(x)}/{I_m}, \quad \text{ for }    
I_m \;\eqdef\; \int_{-1}^1 f(x)\,dx.
\end{equation}

By lower bounding the normalization constant $I_m$ we can show the following properties.
\begin{claim}[Properties of $g$]
\label{clm:g-properties}
It holds that
\begin{enumerate}[label=(\alph*)]
\item For $c_1/m\le |x|\le 1$,
      \(
      0\le g(x)\lesssim  mk\,2^{-k}.
      \)
\item For $|x|\ge 1$,
      \(
      |g(x)|\le mk\,|2x|^{(m-1)k}.
      \)
\end{enumerate}
\end{claim}
We refer the reader to  \Cref{sec:omited-sandwiching} for the proof of \Cref{clm:g-properties}.

Fix parameters $w,B>0$ (to be chosen later) and define
\begin{equation}\label{eq:p-definition}
p(x) \;\eqdef\; \int_{(x-t-w)/B}^{(x-t)/B} g(y)\,dy.    
\end{equation}

First note that $g$ is a polynomial of degree $(m-1)k$ since $f$ is a polynomial of degree $(m-1)k$.
Therefore, integrating and evaluating at a linear argument makes $p$ a polynomial of the degree at most $(m-1)k+1$.

Assume that $B\geq C (w+t)$, for a sufficiently large constant $C>0$ and $k\ge 2\log_2(m)$.
Let $\Delta \;\eqdef\; {c_1 B}/{m}$.  
We analyze the behavior of $p$ at the different regimes.
\begin{claim}[Properties of $p$]
\label{clm:p-approx-step}
There exists  a sufficiently large constant $C>0$ such that:
\begin{enumerate}[label=(\roman*)]
\item If $|x|\ge B/2$, then 
      $|p(x)| \leq (C|x|/B)^{mk}$.
\item For $-B/2<x<t-\Delta$ we have that $p(x)\leq C 2^{-k/4}$.
\item For $t+\Delta<x<t+w-\Delta$ we have that $p(x)\geq 1- C2^{-k/4}$.
\item For $|x|\leq B/2$ we have that  $p(x)\in [0,1]$.
\end{enumerate}
\end{claim}
We refer the reader to  \Cref{sec:omited-sandwiching} for the proof of \Cref{clm:p-approx-step}.
Define
\begin{equation}
p_-(x) \;\eqdef\; p(x-\Delta) 
 - \Big(\frac{C x}{B}\Big)^{mk} - C2^{-k/4},\;    \label{eq:pminus}
\end{equation}
for a sufficiently large constant $C$.
By construction $p_-$ is a polynomial of degree at most $mk$.

\begin{lemma}\label{lem:p-minus-below}
Assume that $m$ is greater than a sufficiently large constant. 
Then, it holds that $p_-(x)\le h(x)$ for all $x\in \R$.
\end{lemma}
\begin{proof}%
We partition our proof in two the cases, whether or not $x< t$.\\
\noindent\textbf{Case 1: $x<t$.}
Here $h(x)=0$, so it suffices to show that $p_-(x)\le 0$, i.e.
\[
p(x-\Delta)\;\le\;\Big(\frac{C x}{B}\Big)^{mk}+C2^{-k/4}.
\]

We split into two subcases depending on whether $|x-\Delta|\le B/2$ or not.
First assume  $|x-\Delta|\le B/2$.
Since $x<t$, we have $x-\Delta < t-\Delta$.
By  \Cref{clm:p-approx-step} applied to $z=x-\Delta$ we have that
$p(x-\Delta)\le C_0\,2^{-k/4},$
for some absolute constant $C_0>0$. 
Choosing $C\ge C_0$ in the definition of $p_-$, we obtain
\[
p(x-\Delta)\;\le\;C\,2^{-k/4}\;\le\;\Big(\frac{C x}{B}\Big)^{mk}+C2^{-k/4},
\]
and therefore $p_-(x)\le 0$ in this subcase.

Now assume that $|x-\Delta|>B/2$.
By \Cref{clm:p-approx-step}, we have
\[
|p(x-\Delta)| \;\le\; \Big(\frac{C_0 |x-\Delta|}{B}\Big)^{mk}.
\] 
Now since
$\Delta = O(B/m)$ and $m$ greater than a sufficiently large constant, we have that $x \geq B/4$. 
Therefore, 
we have that 
$|x-\Delta|\leq |x|+\Delta\lesssim |x|+B\lesssim |x|$.
Therefore, we have that there exists $C_1>0$ such that 
\[
\Big(\frac{C_0 |x-\Delta|}{B}\Big)^{mk}
\;\le\;
\Big(\frac{C_1 |x|}{B}\Big)^{mk}
\quad\text{for all }x\in\R.
\]
Thus
$
p(x-\Delta)
\;\le\;
\Big(\frac{C_1 |x|}{B}\Big)^{mk}.
$
Taking $C\ge C_1$ in the definition of $p_-$, we get
\[
p(x-\Delta)
\;\le\;
\Big(\frac{C |x|}{B}\Big)^{mk}
\;\le\;
\Big(\frac{C x}{B}\Big)^{mk} + C2^{-k/4},
\]
and again $p_-(x)\le 0$.

\smallskip

\noindent\textbf{Case 2: $x\ge t$.}
Here $h(x)=1$, and we want to show that $p_-(x)\le 1$.

We again split into two subcases. First $|x-\Delta|\le B/2$.
By  \Cref{clm:p-approx-step}, for all $z$ with $|z|\le B/2$ we have
$0\le p(z)\le 1$. Taking $z=x-\Delta$, we get
\[
p(x-\Delta)\le 1
\Rightarrow
p_-(x)\le p(x-\Delta)\le 1 = h(x).
\]

Now for $|x-\Delta|>B/2$. As in Case 1, \Cref{clm:p-approx-step} yields
\[
|p(x-\Delta)| \;\le\; \Big(\frac{C_0 |x-\Delta|}{B}\Big)^{mk}
\;\le\;
\Big(\frac{C_1 |x|}{B}\Big)^{mk}
\]
for suitable constants $C_0,C_1>0$ and all $x$. Taking $C\ge C_1$, we obtain
\[
p_-(x)
\;\le\;
\Big(\frac{C_1 |x|}{B}\Big)^{mk}
\;-\;
\Big(\frac{C x}{B}\Big)^{mk}
\;-\;C2^{-k/4}
\;\le\;
0
\;\le\;
1.
\]
Combining Case~1 and Case~2, we have shown that
\[
p_-(x)\;\le\;h(x)\qquad\text{for all }x\in\R,
\]
which completes the proof of \Cref{lem:p-minus-below}.
\end{proof}

We now specify concrete values for the parameters $w,B,m,k$ used in the
construction above.
Let $C_w,C_k,C_m,C_B>0$ be sufficiently large absolute constants (to be
fixed within the analysis below) and set
\begin{equation}\label{eq:parameters}
\begin{aligned}
  w &\eqdef C_w(\sqrt{\log(1/\alpha)}+t),\\
  k &\eqdef 2\left\lceil C_k\bigl(t^2 + \log(1/\alpha)\bigr)\right\rceil
      \quad\text{(even)},\\
  m &\eqdef 2\left\lceil C_m\,\frac{(t+1)^2\bigl(t^2+\log(1/\alpha)\bigr)}{\alpha^2}\right\rceil+1
      \quad\text{(odd)},\\
  B &\eqdef C_B \sqrt{mk}.
\end{aligned}
\end{equation}
Notice that for $C_k$ large enough, we have $k\ge 2\log_2 m$ and $B\geq C_B \sqrt{k}\gtrsim C_B(t+ w)$ (where we used the inequality $\sqrt{a+b}\geq (\sqrt{a}+\sqrt{b})/\sqrt{2}$ ).
Hence the assumptions of
\Cref{clm:p-approx-step} are met.
Moreover
\[
\deg(p_-)\;\le mk
=O \left(\frac{(t+1)^6\log^2(1/\alpha)}{\alpha^2}\right)\;,
\]
which matches the degree bound in the statement.
Next we bound the error of $p_-$ with respect to $h$. 
\begin{lemma}\label{lem:error}
With the above choice of parameters, it holds that
\[
  \E_{x\sim\cN(0,1)}\big[h(x)-p_-(x)\big]
  \lesssim \alpha\,\E_{x\sim\cN(0,1)}[h(x)].
\]
\end{lemma}
\begin{proof}%
First note that 
\begin{equation*}
  \E\big[h(x)-p_-(x)\big]
  \;=\;
  \E\big[h(x)-p(x-\Delta)\big]
  \;+\;
  \E\Big[\Big(\frac{C x}{B}\Big)^{mk}\Big]
  \;+\;
  C2^{-k/4}.
\end{equation*}
We will bound the last two terms directly, and then control
$\E[h(x)-p(x-\Delta)]$ by partitioning according to the regions in
\Cref{clm:p-approx-step}.

We may write
$\E\Big[\Big(\frac{C x}{B}\Big)^{mk}\Big]=
\Big(\frac{C}{B}\Big)^{mk}\E\big[|x|^{mk}\big]$.
For the absolute moment of a standard Gaussian we have the known 
bound (for even $j$)
$\E\big[|x|^j\big]\;\le\;j^{j/2}.$
Taking $j=mk$ gives
$\E\big[|x|^{mk}\big]
\;\le\;
(mk)^{mk/2}$.
hence
\[
\E\Big[\Big(\frac{C x}{B}\Big)^{mk}\Big]
\;\le\;
\Big(\frac{C}{B}\Big)^{mk}( mk)^{mk/2}
=
\Big(\frac{C\sqrt{mk}}{B}\Big)^{mk}\;.
\]
Setting $B\geq  2C\sqrt{mk}$ suffices to have 
\[
\E\Big[\Big(\frac{C x}{B}\Big)^{mk}\Big]
\le
2^{-mk}
\le
2^{-k}
\leq
\alpha\,\E[h(x)].
\]
The last inequality holds by our choice of
$k=\Theta\big(t^2+\log(1/\alpha)\big)$ and the standard lower bound
$\E[h(x)]=\Pr[x\ge t]\ge e^{-t^2/2}/(t+1)$.
Similarly, since $k\geq C_k( t^2+\log(1/\alpha))$, we may choose $C_k$ 
large such that that
$2^{-k}
\le \alpha\E[h(x)]$.

It remains to bound
$\E\big[h(x)-p(x-\Delta)\big]$.
Let us write $z\eqdef x-\Delta$ so that $p$ is evaluated at $z$.
We bound the contribution of $h(x)-p(z)$ from each region of $z$ defined in \Cref{clm:p-approx-step}.

\noindent\textbf{Region A: $|z|\ge B/2$.}
On this region, \Cref{clm:p-approx-step} gives
$
|p(z)|\;\le\;\Big(\frac{C|z|}{B}\Big)^{mk}.
$
Moreover $0\le h(x)\le 1$, so
\[
|h(x)-p(z)|
\;\le\;
1 + |p(z)|
\;\le\;
1+\Big(\frac{C|z|}{B}\Big)^{mk}.
\]
Hence
\begin{align*}
\E\big[|h(x)-p(z)|\Ind\{|z|\ge B/2\}\big]
&\le
\Pr[|z|\ge B/2]
+
\E\Big[\Big(\frac{C|z|}{B}\Big)^{mk}\Ind\{|z|\ge B/2\}\Big]\\
&\leq 
\Pr[|x|\ge B/4]
+
\E\Big[\Big(\frac{2C|x|}{B}\Big)^{mk}\Ind\{|x|\ge B/4\}\Big]\;,
\end{align*}
where in the last inequality we used a similar argument as in the proof of \Cref{lem:p-minus-below} Case 1, to argue that if $|z|\geq B/2$ then $\Delta\leq B/4$ and $|x|\geq B/4$.

Now the second term is $O(\alpha \E[h(x)])$ by the same analysis as for the term of $\E\Big[\Big(\frac{C x}{B}\Big)^{mk}\Big]$ and our choice of $B$.

For the first term we have that since $B\geq C\sqrt{k}$ for a sufficiently large constant $C$, we have that $\Pr[|x|\ge B/4]\leq e^{-k}$ hence by our previous analysis we have that $\Pr[|x|\ge B/4]\leq \alpha \E[h(x)]$.

Therefore, in total we obtain\begin{equation*}\label{eq:RA}
\E\big[|h(x)-p(z)|\Ind\{|z|\ge B/2\}\big]
\lesssim \alpha \E[h(x)].
\end{equation*}

\medskip
\noindent\textbf{Region B: $-B/2<z<t-\Delta$.}
In this region we are strictly to the left of the threshold in $z$,
hence $x=z+\Delta<t$ and $h(x)=0$.  By
\Cref{clm:p-approx-step}, for $-B/2<z<t-\Delta$ we have
$0\;\le\;p(z)\;\le\;C\,2^{-k/4}$.
Thus
\[
|h(x)-p(z)|\;\le\;C\,2^{-k/4}\;.
\]
Therefore, it follows that
\[
\E\big[(h(x)-p(z))\Ind\{-B/2<z<t-\Delta\}\big]
\;\le\;
C\,2^{-k/4}\,\Pr[-B/2<z<t-\Delta]
\;\le\;
C\,2^{-k/4}.
\]
Note that as previously because of the value of $k$ we have 
\begin{equation*}
\E\big[(h(x)-p(z))\Ind\{-B/2<z<t-\Delta\}\big]
\leq \alpha \E[h(x)].
\end{equation*}

\medskip
\noindent\textbf{Region C : $t-\Delta\le z\le t+\Delta$.}
Note that $x=z+\Delta$ lies in $[t,\,t+2\Delta]$ by \Cref{clm:p-approx-step} hence in this region  we have $p(z)\in[0,1]$.
Thus
$
|h(x)-p(z)|\;\le\;1.
$
Therefore
\[
\E\big[|h(x)-p(z)|\Ind\{t-\Delta\le z\le t+\Delta\}\big]
\;\le\;
\Pr[t-\Delta\le z\le t+\Delta]
=
\Pr[t\le x\le t+2\Delta].
\]
Using the standard bound $\Pr[a\le x\le a+b]=\int_{a}^{a+b}\phi(x)dx \leq b\,\phi(a)$, we get
\[
\Pr[t\le x\le t+2\Delta]
\;\le\;
2\Delta\,\phi(t).
\]
Now we prove that $\phi(t)\lesssim (t+1)\Pr[x\geq t]$ for all $t\geq 0$.
Note that for $t\leq 1$ we have $\Pr[x\geq t] \gtrsim \phi(t)$ because both of them are constants. 
Also for $t\geq 1$ we have that $\phi(t)\leq (t+1/t)\Pr[x\geq t]\leq (t+1) \Pr[x\geq t]$ from the fact that 
$\Pr[x\ge t]\geq t \phi(t)/(1+t^2)$.
Hence in total $\phi(t)\lesssim (t+1)\Pr[x\geq t]$ for all $t\geq 0$.

Therefore, 
\[
\Pr[t\le x\le t+2\Delta]
\lesssim 
\Delta\Bigl(t+1\Bigr)\Pr[x\ge t]\;.
\]
Now 
$$\Delta(t+1)\lesssim \frac{\sqrt{k}(t+1)}{\sqrt{m}}$$
Hence, if we choose $m\geq C(t+1)^2(t^2+\log(1/\alpha))/\alpha^2$ for a sufficiently large constant $C>0$ we have that
$\Delta(t+1)\lesssim \alpha$. Therefore,
\begin{equation*}
\E\big[|h(x)-p(z)|\Ind\{t-\Delta\le z\le t+\Delta\}\big]
\lesssim \alpha \E[h(x)].
\end{equation*}

\noindent\textbf{Region D: $t+\Delta<z<t+w-\Delta$.}
In this region we are to the right of the threshold in $z$ but still
inside the ``bump window'' of width $w$.  Note that
$x=z+\Delta\in(t+2\Delta,\,t+w)$, so $h(x)=1$.
By \Cref{clm:p-approx-step},
$p(z)\;\ge\;1-C\,2^{-k/4}$ for $ t+\Delta<z<t+w-\Delta,$
and hence
\[
0\;\le\;h(x)-p(z)\;\le\;C\,2^{-k/4}.
\]
This implies that
\[
\E\big[(h(x)-p(z))\Ind\{t+\Delta<z<t+w-\Delta\}\big]
\lesssim 2^{-k/4}\;.
\]
Again, because of our choice of $k$, we have
\begin{equation*}
\E\big[(h(x)-p(z))\Ind\{t+\Delta<z<t+w-\Delta\}\big]
\lesssim \alpha\E[h(x)].
\end{equation*}

\noindent\textbf{Region E: $t+w-\Delta\le z\le B/2$.}
Here $z$ is to the right of the bump interval, but still with
$|z|\le B/2$, so by \Cref{clm:p-approx-step} we have
$0\le p(z)\le 1$.  Moreover $x=z+\Delta\ge t+w\geq t$, so $h(x)=1$.
Thus
\[
0\;\le\;h(x)-p(z)\;\le\;1\;,
\]
and therefore
\[
\E\big[(h(x)-p(z))\Ind\{t+w-\Delta\le z\le B/2\}\big]
\;\le\;
\Pr[x\ge t+w].
\]
By our choice of $w\geq C_w\sqrt{\log(1/\alpha)}$ and a standard Gaussian
tail ratio bound, we can ensure
\[
\frac{\Pr[x\ge t+w]}{\Pr[x\ge t]}
\leq \alpha,
\]
which yields
\begin{equation*}
\E\big[(h(x)-p(z))\Ind\{t+w-\Delta\le z\le B/2\}\big]
\;\le\; \alpha\E[h(x)]\;.
\end{equation*}
Combining the above bounds completes the proof of \Cref{lem:error}.
\end{proof}

Ending we show that $p_+$ can be constructed by a simple transformation to $p_-$.
\begin{lemma}\label{cl:pplus}
Define
$p_+(x) \;\eqdef\; 1 - p_-\big(2t - x\big)$. It holds that 
\begin{enumerate}[label=(\roman*)]
    \item $p_+(x)\geq h(x)$ for all $x\in \R$.
    \item $\E\big[p_+(x)-h(x)\big] \lesssim \alpha\E[h(x)]$.
\end{enumerate}
\end{lemma}
The proof is similar to that of \Cref{lem:error}; we refer the reader to \Cref{sec:omited-sandwiching} for details.

Putting everything together, we obtain
\[
\E\big[p_+(x)-p_-(x)\big]
=
\E\big[p_+(x)-h(x)\big] + \E\big[h(x)-p_-(x)\big]
\;\le\;
\alpha\,\E[h(x)]\,
\]
completing the proof.
\end{proof}

\subsection{Omitted Proofs and Facts for the Sandwiching Result}\label{sec:omited-sandwiching}
In this subsection we collect auxiliary facts and provide the proofs that were omitted from the proof of \Cref{lem:structural}. 
These details are included here for completeness and to keep the presentation in \Cref{sec:sandwiching-result} streamlined.
\begin{fact}[Properties of Chebyshev polynomials (e.g., \cite{mason2002chebyshev})]\label{fact:Chebychev_properties}
Let $T_m$ be the $m$-th Chebyshev polynomial of the first kind. Then:
\begin{enumerate}[label=(\roman*)]
    \item 
    For every $x \in [-1,1]$ there exists $\theta \in [0,\pi]$ with $x = \cos\theta$, and
    \[
      T_m(x) \;=\; T_m(\cos\theta) \;=\; \cos(m\theta).
    \]
    In particular, for all $x\in[-1,1]$,
    \[
      |T_m(x)| \le 1.
    \]

    \item  
    For all $x\in\mathbb{R}$,
    \[
      T_m(-x) = (-1)^m T_m(x).
    \]
    Hence if $m$ is odd, $T_m$ is an odd polynomial and contains only odd-degree monomials.

    \item 
    If $m$ is odd, then for every $\theta\in\mathbb{R}$,
    \[
      T_m(\sin\theta) 
      \;=\; T_m\big(\cos(\tfrac{\pi}{2}-\theta)\big)
      \;=\; \cos\big(m(\tfrac{\pi}{2}-\theta)\big)
      \;=\; (-1)^{(m-1)/2}\,\sin(m\theta),
    \]
    so in particular
    \[
      \big|T_m(\sin\theta)\big| = \big|\sin(m\theta)\big|.
    \]

    \item 
    For $|x|\ge 1$ we have the explicit formula
    \[
      T_m(x)
      = \frac{1}{2}\Big( \big(x+\sqrt{x^2-1}\big)^m + \big(x-\sqrt{x^2-1}\big)^m \Big),
    \]
    and therefore
    \[
      |T_m(x)|
      \;\le\; \bigl(|x|+\sqrt{x^2-1}\bigr)^m
      \;\le\; (2|x|)^m.
    \]

    \item 
    If $m$ is odd, then
    \[
      T_m'(0) \;=\; (-1)^{(m-1)/2}\,m.
    \]
\end{enumerate}
\end{fact}
\begin{fact}\label{lem:sin-mtheta-bound}
For every integer $m\ge 1$ and every $\theta\in\mathbb{R}$,
\[
|\sin(m\theta)| \;\le\; m\,|\sin\theta|.
\]
\end{fact}

\begin{proof}
We proceed by induction on $m$. For $m=1$ the claim is trivial.
Assume $|\sin(m\theta)| \le m|\sin\theta|$ holds for some $m\ge 1$.
Using the angle-addition formula,
\[
\sin((m+1)\theta)
= \sin(m\theta)\cos\theta + \cos(m\theta)\sin\theta,
\]
so
\[
|\sin((m+1)\theta)|
\le |\sin(m\theta)|\,|\cos\theta| + |\cos(m\theta)|\,|\sin\theta|
\le |\sin(m\theta)| + |\sin\theta|
\le (m+1)|\sin\theta|.
\]
The proof follows by induction.
\end{proof}

\begin{claim}[Properties of $f$]
Let $f$ the function defined in \Cref{eq:f-definition}.
There exist absolute constants $c_0,c_1>0$ such that:
\begin{enumerate}[label=(\roman*)]
\item For all $|x|\le 1$, 
      \(
      0\le f(x)\le 1.
      \)
\item For all $|x|\le c_0/(mk)$, 
      \(
      \tfrac12 \le f(x)\le 1.
      \)
\item For all $c_1/m\le |x|\le 1$, 
      \(
      0\le f(x)\le 2^{-k}.
      \)
\item For all $|x|\ge 1$,
      \(
      |f(x)|\le (2|x|)^{k(m-1)}.
      \)
\end{enumerate}
\end{claim}

\begin{proof}
\medskip\noindent\textbf{Proof of (i):} 
Since $k$ is even, $f(x)\ge 0$ for all $x$, so it suffices to show that $f(x)\leq 1$ for all $\abs{x}\leq 1$.

Fix $x\in[-1,1]\setminus\{0\}$. Because $|x|\le 1$, we can write $x = \sin\theta$ for some
$\theta\in[-\pi/2,\pi/2]$. By \Cref{fact:Chebychev_properties}, 
$
T_m(\sin\theta) = (-1)^{(m-1)/2}\sin(m\theta),
$
hence
\[
\left|\frac{T_m(x)}{m x}\right|
= \left|\frac{T_m(\sin\theta)}{m\sin\theta}\right|
= \frac{|\sin(m\theta)|}{m|\sin\theta|}.
\]
Applying \Cref{lem:sin-mtheta-bound}, we obtain
\[
\left|\frac{T_m(x)}{m x}\right|
= \frac{|\sin(m\theta)|}{m|\sin\theta|}
\le 1.
\]
Thus $0\le f(x)\le 1$ for all $x\neq 0$ with $|x|\le 1$.

Finally, since $f$ is continuous (it is a polynomial) we obtain also have that $f(0)\leq 1$.
Which concludes the proof of (i).
\medskip\noindent\textbf{Proof of (ii):}
Since $T_m$ has only odd degree terms $T_m$ is odd, hence $f$ is even. Thus it suffices to prove the inequality for $x\ge0$.

Fix $c_0:=\tfrac12$. For $x\ge 0$ with $x\le c_0/(mk)$, we in particular have
$x\le c_0\le\tfrac12$, so $x\in[0,\tfrac12]$ and we may set
$
x=\sin y,$ $y\eqdef \arcsin x\in[0,\tfrac{\pi}{2}].
$

Note that from \Cref{fact:Chebychev_properties} 
$T_m(\sin y)
= \sigma\,\sin(my),
$ $ \sigma:=(-1)^{(m-1)/2}$.
Since $k$ is even, $\sigma^k=1$ we have
\[
f(x)=\left(\frac{T_m(x)}{m x}\right)^k
    = h(y)^k \qquad h(y):=\frac{\sin(my)}{m\sin y}\;.
\]
We rewrite
\[
h(y)=\frac{\sin(my)}{m\sin y}
    = \frac{\sin(my)}{my}\cdot\frac{y}{\sin y}\geq \frac{\sin(my)}{my}\;,
\]
Since for $y\in(0,\tfrac{\pi}{2})$ we have $\sin y\le y$.

Next for $x\in[0,\tfrac12]$,
$
y=\arcsin x\le 2x.
$
Using $x\le c_0/(mk)$ with $c_0=\tfrac12$, we deduce
\[
0<y\le 2x\le \frac{2c_0}{m k}=\frac{1}{m k},
\qquad\text{so}\qquad
0<my\le \frac1k\le \frac12,
\]
because $k\ge2$.

For $u\in(0,1]$ the Taylor expansion of $\sin u$ gives the standard estimate
$
\sin u\ge u-\frac{u^3}{6},
$
so, for $0<u\le1$,
$
\frac{\sin u}{u}\ge 1-\frac{u^2}{6}.
$
Applying this with $u=my$ (which satisfies $0<my\le\tfrac12$ as above), we obtain
$
\frac{\sin(my)}{my}\ge 1-\frac{(my)^2}{6}.
$
Therefore
\[
h(y)
\ge \frac{\sin(my)}{my}\cdot\frac{y}{\sin y}
\ge 1-\frac{(my)^2}{6}\;.
\]

Finally, we bound $(my)^2$ in terms of $k$. Using $y\le 2x$ and $x\le c_0/(mk)$,
\[
(my)^2 \le m^2(2x)^2
        \le m^2\left(\frac{2c_0}{m k}\right)^2
        = \frac{4c_0^2}{k^2}.
\]
With $c_0=\tfrac12$, we have
$
h(y)\ge 1-\frac{1}{6k^2}.
$

For $t\in[0,1]$ and integer $k\ge1$ we have the elementary inequality
$(1-t)^k \ge 1-k t$, 
this follows by expanding $(1-t)^k$ via the binomial theorem and discarding
the nonnegative higher-order terms.

Applying this with $t = \dfrac{1}{6k^2}$ we have
\[
h(y)^k
\;\ge\;
\left(1-\frac{1}{6k^2}\right)^k
\;\ge\;
1-k\cdot\frac{1}{6k^2}
\;=\;
1-\frac{1}{6k}
\;\ge\;
1-\frac{1}{6}
\;=\;
\frac{5}{6}
\;>\;\frac{1}{2}.
\]

Thus, for every $x>0$ with $x\le c_0/(mk)$ (and $c_0=\tfrac12$), we have
$f(x)=h(y)^k>\tfrac12$. Since $f$ is even, the same bound holds for $x<0$ with
$|x|\le c_0/(mk)$. Note that since $f$ is continuous same holds for $x=0$.
Which concludes the proof of (ii). 
\medskip\noindent\textbf{Proof of (iii):}

Fix any constant $c_1\ge 2$. Let $x$ satisfy
$
\frac{c_1}{m}\le|x|\le 1.
$ Then, using $|T_m(x)|\le 1$ and $|x|\ge c_1/m$,
\[
\left|\frac{T_m(x)}{m x}\right|
\;\le\;
\frac{1}{m|x|}
\;\le\;
\frac{1}{m\cdot (c_1/m)}
\;=\;
\frac{1}{c_1}
\;\le\;
\frac{1}{2}.
\]
Therefore,
\[
f(x)
=
\left|\frac{T_m(x)}{m x}\right|^k
\;\le\;
\left(\frac{1}{2}\right)^k
=
2^{-k}.
\]
Combining this with the non-negativity of $f$, we obtain
\[
0 \;\le\; f(x)\;\le\;2^{-k}
\quad\text{for all }x\text{ with }\frac{c_1}{m}\le |x|\le 1,
\]
which proves item (iii).
\medskip\noindent\textbf{Proof of (iv):} Note that from \Cref{fact:Chebychev_properties} we have that for all $|x|\geq 1$
\begin{align*}
    T_m(x)\leq \left(\abs{x}+\sqrt{x^2-1}\right)^m \leq (2|x|)^m\;. 
\end{align*}
As a result $({T_m(x)}/{mx})^k\leq (2|x|)^{k(m-1)}$ which concludes the proof of item (iv).
\end{proof}

\begin{claim}[Properties of $g$]
Let $g$ be the function defined at \Cref{eq:g-definition}. It holds that
\begin{enumerate}[label=(\alph*)]
\item For $c_1/m\le |x|\le 1$,
      \(
      0\le g(x)\lesssim  mk\,2^{-k}.
      \)
\item For $|x|\ge 1$,
      \(
      |g(x)|\le mk\,|2x|^{(m-1)k}.
      \)
\end{enumerate}
\end{claim}
\begin{proof}
By \Cref{clm:f-properties} (i), we have that
\[
I_m \;\ge\; \int_{-c_0/(mk)}^{c_0/(mk)} \frac12\,dx
   \;\ge\; \frac{c_0}{mk}.
\]
 Applying the above along with \Cref{clm:f-properties} (iii) and (iv) concludes the proof of \Cref{clm:g-properties}. 
\end{proof}

\begin{claim}
Let  $p$ the function defined at \Cref{eq:p-definition}.
It holds that:
\begin{enumerate}[label=(\roman*)]
\item If $|x|\ge B/2$, then 
      $|p(x)| \leq (C|x|/B)^{mk}$.
\item For $-B/2<x<t-\Delta$ we have that $p(x)\leq C 2^{-k/4}$.
\item For $t+\Delta<x<t+w-\Delta$ we have that $p(x)\geq 1- C2^{-k/4}$.
\item For $|x|\leq B/2$ we have that  $p(x)\in [0,1]$.
\end{enumerate}
\end{claim}
\begin{proof}

\medskip\noindent\textbf{Proof of (i):}
Fix $|x|\ge B/2$ and denote the integration interval by
$
I_x \;\eqdef\; \Big[\frac{x-t-w}{B},\,\frac{x-t}{B}\Big].
$
Fix a point $y\in I_x$.
Using the assumption $B\ge C(w+t)$ for a sufficiently large constant $C>4$ and $|x|\ge B/2$, we obtain
\[
|y|
\;\ge\;
\frac{|x|}{B} - \frac{t+w}{B}
\;\ge\;
\frac{|x|}{B} - \frac{1}{C}
\;\ge
\frac{1}{2}
-\frac{1}{C}\geq \frac{1}{4}\;.
\]
From  \Cref{clm:g-properties}, we have that for all $y$ with $|y|\ge 1/4$,
$
|g(y)|
\lesssim
\, mk\!\left(2^{-k} + (2|y|)^{(m-1)k}\right).
$
Moreover, for every $y\in I_x$ we have
\[
|y|
\;\le\;
\frac{|x|+t+w}{B}
\;\le\;
\frac{|x|}{B} + \frac{t+w}{B}
\;\le\;
\frac{|x|}{B} + \frac{1}{C}
\;\le\;
\frac{2|x|}{B}\;.
\]
Hence, for all $y\in I_x$,
\[
|g(y)|
\;\le\;
C_0\, mk\!\left(2^{-k} + \left(\frac{4|x|}{B}\right)^{(m-1)k}\right).
\]
Using this bound and the fact that the length of $I_x$ is  $w/B$, we obtain
\begin{align*}
|p(x)|
&\le
\int_{I_x} |g(y)|\,dy
\;\le\;
\frac{w}{B}\,\sup_{y\in I_x} |g(y)| \\
&\lesssim  mk
\Bigg(
2^{-k} + \Big(\frac{4|x|}{B}\Big)^{(m-1)k}
\Bigg).
\end{align*}
Now, since $|x|\ge B/2$, we have $\tfrac{4|x|}{B}\ge 2$, and thus
\(
\big(\tfrac{4|x|}{B}\big)^{(m-1)k} \ge 1
\),
so the term $2^{-k}\le 1$ is dominated by it. Therefore,
\[
|p(x)|
\lesssim mk
\Big(\frac{4|x|}{B}\Big)^{(m-1)k}.
\]
Rewrite this as
\[
|p(x)|
\lesssim mk
\frac{\big(\frac{4|x|}{B}\big)^{mk}}{\big(\frac{4|x|}{B}\big)^k}.
\]
Since $|x|\ge B/2$, we have
\(
\big(\tfrac{4|x|}{B}\big)^k \ge 4^k
\).
Using the assumption $k\ge 2\log_2 m$, we get
\(
4^{k/2} = 2^{k} \ge 2^{2\log_2 m}=m^2
\),
and thus
\[ mk\,\frac{1}{4^k}
\le
\frac{ m}{4^{k/2}}\frac{k}{2^k}
\;\le\;
1
\]
Therefore,
\[
|p(x)|
\lesssim
\bigg(\frac{4|x|}{B}\bigg)^{mk}\;,
\]
which concludes the proof of (i).

\medskip\noindent\textbf{Proof of (ii):}
Assume $-B/2 < x < t-\Delta$.
We first prove that  $I_x \subseteq [-1,-c_1/m]$.
Since $x < t-\Delta$, we have
\[
\frac{x-t}{B}
\;\le\;
\frac{t-\Delta - t}{B}
=
-\frac{\Delta}{B}
=
-\frac{c_1}{m}.
\]
Thus the right endpoint of $I_x$ is to the left of $-c_1/m$.
For the left endpoint, using $x > -B/2$ we get
\[
\frac{x-t-w}{B}
\;>\;
\frac{-B/2 - t - w}{B}
=
-\Big(\frac{1}{2} + \frac{t+w}{B}\Big).
\]
Since $B \ge C(w+t)$, we have $(t+w)/B \le 1/C$. For $C \ge 2$ gives
\[
\frac{x-t-w}{B}
\;>\;
-\Big(\frac{1}{2} + \frac{1}{C}\Big)
\;\ge\;
-1.
\]
Combining the bounds we obtain
$I_x \subseteq [-1,-c_1/m].$
By \Cref{clm:g-properties}, for all $y$ with $c_1/m \le |y|\le 1$,
$
0 \;\le\; g(y) \;\lesssim mk\,2^{-k}.
$
Hence
\[
p(x)
=
\int_{I_x} g(y)\,dy
\le
\frac{w}{B}\,\sup_{y\in I_x} g(y)
\lesssim
\frac{w}{B}\, mk\,2^{-k}
\leq
mk2^{-k}\;.
\]
Using $k \ge 2\log_2 m$, we have $m \le 2^{k/2}$, so
$mk\,2^{-k}\;\le\;k 2^{-k/2}$. 
Also note that $k2^{-k}$ is less than a constant for all $k\geq 1$.
Therefore,
$p(x) \lesssim 2^{-k/4}$
which proves \textup{(ii)}.

\medskip\noindent\textbf{Proof of (iii):}
First we prove that $(-c_1/m,c_1/m)\subseteq I_x\subseteq[-1,1]$.
Let $x$ such that $t+\Delta < x < t+w-\Delta$. 
Note that 
$\frac{x-t-w}{B}\geq \frac{\Delta-w}{B}\geq -\frac{w}{B}\geq -1$.
Also, 
$     \frac{x-t}{B}\leq \frac{w-\Delta}{B}\geq \frac{w}{B}\geq 1\;,
$
which proves that $I_x\subseteq[-1,1]$.
Also note that 
$\frac{x-t-w}{B}\leq -\frac{\Delta}{B}= -c_1/m.$
and 
$    \frac{x-t}{B}\geq \frac{\Delta}{B}= c_1/m$.
Thus $(-c_1/m,c_1/m)\subseteq I_x$. 

By \Cref{clm:g-properties}, for every $y$ with $c_1/m\le |y|\le 1$ we have
$0\le g(y)\le mk\,2^{-k}$.
Therefore
\[
1-p(x)
=
\int_{[-1,1]\setminus I_x} g(y)\,dy
\le
\int_{\{y\in[-1,1]: |y|\ge c_1/m\}} g(y)\,dy
\le
2 mk\,2^{-k}.
\]
As in the proof of item~(ii), using $k\ge 2\log_2 m$ we have that 
$2 mk\,2^{-k} \lesssim 2^{-k/4}$.
Thus
\[
p(x)\ge 1-C2^{-k/4},
\]
which proves (iii).

\medskip\noindent\textbf{Proof of (iv):}
Let $x$ be such that $|x|\le B/2$.
Since $g(y)\ge 0$ for all $y$, we immediately have $p(x)\ge 0$.
It remains to prove that $p(x)\le 1$.
We first check that $I_x\subseteq[-1,1]$.
For the right endpoint,
\[
\Big|\frac{x-t}{B}\Big|
\le
\frac{|x|}{B} + \frac{t}{B}
\le
\frac{1}{2} + \frac{t}{B}
\le
\frac{1}{2} + \frac{1}{C}
\le 1.
\]
Similarly, for the left endpoint we have
\[
\Big|\frac{x-t-w}{B}\Big|
\le
\frac{|x|}{B} + \frac{t+w}{B}
\le
\frac{1}{2} + \frac{w+t}{B}
\le
\frac{1}{2} + \frac{1}{C}
\le 1\;.
\]
Thus both endpoints lie in $[-1,1]$, so $I_x\subseteq[-1,1]$.

By construction of $g$ we have $g(y)\ge 0$ for all $y$ and
$\int_{-1}^1 g(y)\,dy = 1$.
Therefore, 
\[
p(x)
=
\int_{I_x} g(y)\,dy
\le
\int_{-1}^1 g(y)\,dy
= 1\;,
\]
which completes the proof of (iv).
\end{proof}

\begin{lemma}
Define
$p_+(x) \;\eqdef\; 1 - p_-\big(2t - x\big)$ with $p_-$ as defined in \Cref{eq:pminus}. Assume parameter choices in \Cref{eq:parameters}. It holds that 
\begin{enumerate}[label=(\roman*)]
    \item $p_+(x)\geq h(x)$ for all $x\in \R$.
    \item $\E\big[p_+(x)-h(x)\big] \lesssim \alpha\E[h(x)]$.
\end{enumerate}
\end{lemma}
\begin{proof}
\medskip\noindent\textbf{Proof of (i):}
For $x<t$ we have that $2t-x>t$ hence $p_-(2t-x)\leq h(2t-x)=1$.
Therefore, $p_+(x)
= 1 - p_-(2t-x)
\;\ge\; 0
= h(x)$.
For $x>t$ similarly $2t-x<t$, hence $p_-(2t-x)\leq h(2t-x)=0$.
Thus $p_+(x)
= 1 - p_-(2t-x)
\;\ge\; 1
= h(x).$
For $x=t$ we have that $p_+(x)\geq 1$ by left continuity since $p$ is a polynomial and therefore continuous. 

\medskip\noindent\textbf{Proof of (ii):}

Define the error functions
$
e_-(x)\eqdef h(x) - p_-(x),
$$
e_+(x) \eqdef  p_+(x) - h(x).
$
By construction of $p_+$ we have, for every $x\neq t$,
\[
e_+(x)
= p_+(x) - h(x)
= 1 - p_-(2t-x) - h(x)
=
\begin{cases}
1 - p_-(2t-x), & x<t,\ h(x)=0,\\[3pt]
-\,p_-(2t-x), & x>t,\ h(x)=1,
\end{cases}
\]
while
\[
e_-(2t-x)
= h(2t-x) - p_-(2t-x)
=
\begin{cases}
1 - p_-(2t-x), & 2t-x>t\ \Leftrightarrow\ x<t,\\[3pt]
-\,p_-(2t-x), & 2t-x<t\ \Leftrightarrow\ x>t.
\end{cases}
\]
Thus, for all $x\neq t$,
\[
e_+(x) = e_-(2t-x).
\]
The error at $x=t$ has zero-measure under $\calN(0,1)$, so we may ignore it.
Therefore,
\begin{align*}
\E_{x\sim\cN(0,1)}[e_+(x)]
&=
\int_{\R} e_+(x)\phi(x)\,dx
=
\int_{\R} e_-(2t-x)\phi(x)\,dx \\
&=
\int_{\R} e_-(u)\phi(2t-u)\,du= \E_{x\sim \cN(2t,1)}[e_-(x)]\;,
\end{align*}
where we used the change of variables $u=2t-x$.

Now we use the same analysis as in \Cref{lem:error} to prove that the change in distribution does not affect the approximation error significantly. 

First lets consider the terms $(Cx/B)^{mk}$ and $C 2^{-k/4}$. 
Note that the term $C2^{-k/4}$ is not affected by the change in distribution since its constant. Hence, it is bounded as before  by $\alpha \E_{\cN(0,1)}[h(x)]$. 
For the term $(Cx/B)^{mk}$ we have 
\begin{align*}
  \E_{x\sim \cN(2t,1)}[(Cx/B)^{mk}]&= \left(\frac{C}{B}\right)^{mk}\E_{x\sim \cN(0,1)}[(2t+x)^{mk}]\\
  &\leq \left(\frac{C}{B}\right)^{mk}2^{mk-1}\left( (2t)^{mk}+ (mk)^{mk/2}\right)\\
  &\lesssim \left(\frac{2C\sqrt{mk}}{B}\right)^{mk}\lesssim \alpha \E_{x\sim \cN(0,1)}[h(x)]\;,
  \end{align*} 
where in the first inequality, we use that \(x^{mk}\) is convex (because \(mk\) is even) and the moment bound
\(\E_{x\sim \cN(0,1)}[|x|^{j}] \le j^{j/2}\).
In the second inequality we used that $\sqrt{mk}$ is greater than $2t$.
Now, by our choice \(B = C_B\sqrt{mk}\) for \(C_B\) sufficiently large, we obtain the same bound as in \Cref{lem:error}.

Now we consider each of the regions in \Cref{lem:error} separately. Note that our goal is to bound
\(\E_{\x\sim \cN(2t,1)}\big[h(x)-p(x-\Delta)\big]\),
where \(p\) is the polynomial defined in \Cref{eq:p-definition}. Denote \(z \eqdef x-\Delta\).

\noindent\textbf{Region A: $|z|\ge B/2$.} By \Cref{clm:p-approx-step}(i), for $|z|\ge B/2$ we have
$
|p(z)| \le \Big(\frac{C_0|z|}{B}\Big)^{mk}
$
for an absolute constant $C_0>0$. Since $0\le h(x)\le 1$, it follows that
\[
|h(x)-p(z)| \;\le\; 1+|p(z)|
\;\le\; 1+\Big(\frac{C_0|z|}{B}\Big)^{mk}.
\]
Therefore,
\begin{align*}
\E_{x\sim\cN(2t,1)}\!\Big[|h(x)-p(z)|\,\Ind\{|z|\ge B/2\}\Big]
&\le
\Pr_{x\sim\cN(2t,1)}\!\big[|z|\ge B/2\big]
+
\E_{x\sim\cN(2t,1)}\Big[\Big(\frac{C_0|z|}{B}\Big)^{mk}\Ind\{|z|\ge B/2\}\Big]\\
&\le
\Pr_{x\sim\cN(2t,1)}\big[|z|\ge B/2\big]
+
\E_{x\sim\cN(2t,1)}\Big[\Big(\frac{C_0|z|}{B}\Big)^{mk}\Big].
\end{align*}
Note that by a similar argument as above and that $\Delta\leq t$ we have $\E_{x\sim\cN(2t,1)}\Big[\Big(\frac{C_0|z|}{B}\Big)^{mk}\Big]\lesssim \alpha \E_{x\sim \cN(0,1)}[h(x)]$.

Since $\Delta = c_1B/m$ and $m$ is taken sufficiently large, we may assume $\Delta\le B/10$.
Thus $|z|=|x-\Delta|\ge B/2$ implies $|x|\ge B/2-\Delta\ge 2B/5$.
Also, by our choice $B=C_B\sqrt{mk}$ with $C_B$ large, we have $B\ge 20t$, hence
$2B/5-2t\ge B/5$. Therefore,
\[
\Pr_{x\sim\cN(2t,1)}\big[|z|\ge B/2\big]
\le
\Pr_{x\sim\cN(0,1)}\big[x\ge B/5\big]
\;\le\;
2\exp\!\Big(-\frac{B^2}{50}\Big).
\]
Since $B^2=\Theta(mk)$ and $k=\Theta(t^2+\log(1/\alpha))$, choosing constants so that
$B^2\ge 100\,(t^2+\log(1/\alpha))$ gives
$2\exp(-B^2/50)\le 2e^{-2(t^2+\log(1/\alpha))}=2\alpha^2e^{-2t^2}$.
Using the lower bound proved in \Cref{lem:error}, $\Pr_{g\sim\cN(0,1)}[x\ge t]\ge e^{-t^2/2}/(t+1)$,
we obtain
\[
\Pr_{x\sim\cN(2t,1)}\big[|z|\ge B/2\big]\lesssim \alpha\E_{x\sim\cN(0,1)}[h(x)].
\]

\noindent\textbf{Region B: $-B/2\leq z\le t-\Delta$.} 
Note that the pdf of \(\cN(2t,1)\) is at most the pdf of \(\cN(0,1)\) for any \(x \le t\). Hence,
\[
\E_{x\sim \cN(2t,1)}\!\left[(h(x)-p(z))\Ind(-B/2\le x\le t)\right]
\le
\E_{x\sim \cN(0,1)}\!\left[(h(x)-p(z))\Ind(-B/2\le x\le t)\right]
\lesssim
\alpha \E_{x\sim \cN(0,1)}[h(x)] \, .
\]

\noindent\textbf{Region C: $t-\Delta \leq z\le t+\Delta$.} 
As before since $h(x),p(z)\in [0,1]$ in this region we have 
\[
\E_{x\sim \cN(2t,1)}\big[|h(x)-p(z)|\Ind\{t-\Delta\le z\le t+\Delta\}\big]
\;\le\;
\Pr_{x\sim \cN(2t,1)}[t\le x\le t+2\Delta].
\]
Since the interval $[t,t+2\Delta]$ is on the left side of the mean we have that the pdf of $\cN(2t,1)$ is increasing, hence
\begin{align*}
\Pr_{x\sim \cN(2t,1)}[t\le x\le t+2\Delta]&\leq 2\Delta\phi(t+2\Delta-2t)= 2\Delta \phi(t-2\Delta)\\
&\lesssim\frac{\sqrt{k}}{\sqrt{m}}e^{-t^2/2+ 2\Delta t  }
\lesssim\frac{\sqrt{k}}{\sqrt{m}}e^{-t^2/2}\\
&\lesssim \frac{(t+1)\sqrt{k} }{\sqrt{m}}\p_{x\sim \cN(0,1)}[x\geq t]\lesssim  \alpha \E_{x\sim \cN(0,1)}[h(\x)]\;,
\end{align*}
where in the second and fourth inequality we used the value of $\Delta$ and in third we used the fact that $\phi(t)\lesssim (t+1)\Pr[x\geq t]$ which we have proved in \Cref{lem:error}.

\noindent\textbf{Region D: $t+\Delta \leq z\le t+w-\Delta$.} 
The error bound for this region follows directly from the fact that we are inside the ``bump window''.
Therefore, by \Cref{clm:p-approx-step}, and similarly to the proof of \Cref{lem:error} for the same region, we have
\[
0 \;\le\; h(x)-p(z) \;\le\; C\,2^{-k/4}.
\]
Hence,
\[
\E_{x\sim \cN(2t,1)}\Big[(h(x)-p(z))\,\Ind\{t+\Delta<z<t+w-\Delta\}\Big]
\lesssim 2^{-k/4}.
\]
By our choice of \(k\), it follows that
\[
\E_{x\sim \cN(2t,1)}\Big[(h(x)-p(z))\,\Ind\{t+\Delta<z<t+w-\Delta\}\Big]
\lesssim \alpha\,\E_{x\sim \cN(0,1)}[h(x)].
\]

\noindent\textbf{Region E: $t+w-\Delta \leq z\le B/2$.} 
Note that in this case, by \Cref{clm:p-approx-step} we have $|z|\le B/2$, and hence
$0 \le p(z) \le 1$. Moreover, on the event $\{t+w-\Delta \le z \le B/2\}$ we have
$x \ge t+w \ge t$, so $h(x)=1$. Therefore,
\begin{align*}
\mathbb{E}_{x\sim \mathcal{N}(2t,1)}
\bigl[(h(x)-p(z))\mathbf{1}\{t+w-\Delta \le z \le B/2\}\bigr]
&\le \Pr_{x\sim \mathcal{N}(2t,1)}[x\ge t+w] \\
&= \Pr_{x\sim \mathcal{N}(0,1)}[x\ge w-t] \;.
\end{align*}
Since $w \ge C_w(\sqrt{\log(1/\alpha)}+t)$ for a sufficiently large constant $C_w$
(in particular, taking $C_w\ge 2$ gives $w-t \ge t + C_w\sqrt{\log(1/\alpha)}$),
we obtain
\[
\Pr_{x\sim \mathcal{N}(0,1)}[x\ge w-t]
\le \Pr_{x\sim \mathcal{N}(0,1)}\!\left[x\ge t + C_w\sqrt{\log(1/\alpha)}\right]
\le \alpha\,\Pr_{x\sim \mathcal{N}(0,1)}[x\ge t]\;,
\]
as in the proof of \Cref{lem:error}.

Combining the above cases we have that $\E_{x\sim \cN(0,1)}[e_+(x)]\lesssim \alpha \E_{x\sim \cN(0,1)}[h(x)]$, which completes the proof of \Cref{cl:pplus}.
\end{proof}

\section{SQ Lower Bound for Testably Learning Massart Halfspaces} \label{app:sq-lb}
In this section, we prove that an exponential dependence on $1/\beta^2$ is necessary for our testable learning task in the SQ model, even for the near-homogeneous case.
Before we state our SQ lower bound, we record some basics about the SQ model.

\paragraph{Statistical Query Model}
Statistical Query (SQ) algorithms are a class of algorithms that are allowed 
to query expectations of bounded functions of the underlying distribution
rather than directly access samples. 
Formally, an SQ algorithm has access to the following oracle.

\begin{definition}[\textsc{STAT} Oracle] \label{def:stat-oracle}
Let $D$ be a distribution over $\R^d \times  \{\pm 1\}$. 
A statistical query is a function $q: \R^d \times \{\pm 1\} \to [-1, 1]$.  We define
\textsc{STAT}$(\tau)$ to be the oracle that given a query $q(\cdot, \cdot)$
outputs a value $v$ such that
$|v - \E_{ (x,y) \sim D}\left[q(x, y)\right]| \leq \tau$,
where $\tau>0$ is the tolerance of the query.
\end{definition}

The SQ model was introduced by~\cite{Kearns:98} 
as a natural restriction of the PAC model 
and has been extensively studied in learning 
theory~\cite{FGR+13, FeldmanGV17, Feldman17}. 
The class of SQ algorithms is fairly broad: a wide range of known algorithmic techniques 
in machine learning are known to be implementable using SQs
(see, e.g.,~\cite{Chu:2006, FGR+13, FeldmanGV17}).

\medskip

With this setup, we show the following: 

\begin{proposition}[SQ Lower Bound for Testable Learning of Massart Halfspaces]\label{SQLowerBoundProp}
Any SQ algorithm that testably learns  halfspaces under Gaussian marginals on $\R^d$ and Massart noise with parameters $\beta>\eps>0$ and $\gamma=\Theta(1)$, requires either a query of accuracy $d^{-\Omega(1/\beta^2)}$ or $2^{d^{\Omega(1)}}$ queries.
\end{proposition}

Our main technical tool is the following result implicit in~\cite{diakonikolas2021optimality}:
\begin{lemma}\label{SQLowerBoundLem}
There exists an ensemble of distributions $(X,y)$ on $\R^d\times \{\pm 1\}$ such that:
\begin{enumerate}
    \item The marginal on $X$ is the standard Gaussian.
    \item There is a linear threshold function $f$ with $\Theta(1)$ bias so that $\pr(f(X)\neq y) < 1/2 - 10\beta.$
    \item If $y'$ is a uniform random $\{\pm 1\}$ random variable that is independent of $X$, then any SQ algorithm that distinguishes between $(X,y)$ and $(X,y')$ requires either $2^{d^{\Omega(1)}}$ queries or one query of accuracy $d^{-\Omega(1/\beta^2)}.$
\end{enumerate}
\end{lemma}

\noindent We now prove Proposition \ref{SQLowerBoundProp} via a reduction to Lemma \ref{SQLowerBoundLem}.

\medskip

\begin{proof}
For convenience we produce a construction in $d+1$ dimensions, but changing $d$ to $d-1$ will give our result.

We begin by defining two distributions $(x,X,y)$ and $(x',X',y')$ on $\R\times \R^d\times \{\pm 1\}$.
In both cases, we let $x$ or $x'$ be a standard Gaussian and if it is positive, we let $X$ or $X'$ be a standard $d$-dimensional Gaussian and let $y$ or $y'$ be independently $1$ with probability $(1+\beta)/2$ and $-1$ with probability $(1-\beta)/2.$

If $x'\leq 0$, we let $X'$ be a standard Gaussian and let $y'$ be independently $-1$ with probability $(1+\beta)/2$ and $1$ with probability $(1-\beta)/2.$

If $x \leq 0$, with probability $\beta$ we let $X$ be a standard Gaussian and $y=-1$. Otherwise, we let $(X,y)$ be as in Lemma \ref{SQLowerBoundLem}.

We note the following facts:
\begin{enumerate}
\item $(x',X',y')$ has $y'$ equal to the LTF $\sgn(x')$ with $\eta$-Massart noise.
\item For any LTF $f$, $\pr(f(x',X')\neq y') \geq (1-\beta)/2.$
\item There is an LTF $f$ with bias $\Theta(1)$ so that $\pr(f(x,X) \neq y) < 1/2 - 4\beta.$
\item Any SQ algorithm that distinguishes $(x',X',y')$ from $(x,X,y)$ requires either $2^{d^{\Omega(1)}}$ queries or a query of accuracy $d^{-\Omega(1/\beta^2)}.$
\end{enumerate}
To prove these, we note that 1 follows from the definition (and in fact the Massart noise here is just Random Classification Noise). 
Item 2 follows from noting that $f(x',X') = \sgn(x')$ 
is the optimal classifier. 

Item 3 follows by letting $f(x,X) = f(X)$ the classifier guaranteed by item 2 in Lemma \ref{SQLowerBoundLem}. Then $f(x,X)\neq y$ with probability $1/2$ if $x>0$. Otherwise, there is a probability of at most $\beta$ that $y$ is set to $-1$ and a probability of $1/2$ that $(X,y)$ is as in Lemma \ref{SQLowerBoundLem}, which leads to a probability of error of $1/2(1/2-10\beta).$ Thus, the total probability of error is at most $1/4+\beta+1/4-5\beta = 1/2-4\beta.$

For item 4, we note that the two distributions can be produced in the following way:
\begin{itemize}
\item Let the first coordinate be a random Gaussian.
\item If the first coordinate is positive, let the second be a random Gaussian and the third be $1$ with probability $(1+\beta)/2$.
\item If the first coordinate is negative, with probability $\beta$ let the second coordinate be negative and the third coordinate be $-1$.
\item Otherwise, for generating $(x',X',y')$ let $(X',y')$ be a standard Gaussian and $y'$ and independent, uniform $\{\pm 1\}$ random variable, while for generating $(x,X,y)$ $(X,y)$ are as given in Lemma \ref{SQLowerBoundLem}.
\end{itemize}
Note that these differ only in the last case, but Lemma \ref{SQLowerBoundLem} implies that these options are hard to distinguish in SQ.

At this point, we are effectively done as it is not hard to see that a testable learning algorithm for LTFs with Gaussian marginals and Massart noise can be used to distinguish $(x',X',y')$ and $(x,X,y)$ with one extra query of tolerance $\beta$ used to measure the empirical error of a returned hypothesis. This is because the testable learner is only allowed to reject $(x,X,y)$ (thus rejection implies that we are in the other case). However, a learner applied to $(x,X,y)$ must return a hypothesis with error at most $1/2-3\beta$, while a learner applied to $(x',X',y')$ 
cannot return a hypothesis with error less than $1/2-\beta,$ 
and an approximation of the empirical error to tolerance $\beta$ can distinguish between these two possibilities.
\end{proof}

\section{Non-Testable Learner for Massart Halfspaces with Unknown Bias}
\label{app:unknowngamma}
Here we expand on the discussion about learning with unknown bias in \Cref{ssec:results}.
\begin{corollary}[Learning Massart Halfspaces with Unknown Bias]
Let $D$ be a distribution with $D_\x=\cN^d$ that satisfies the $\eta$-Massart noise condition with respect to a  $\gamma^*$-biased halfspace.
There exists an algorithm that given as input $\eps,\eta$ and  $N=d^{\polylog(\max(1/\gamma^*,1/\eps))}\poly(1/\eps)$ i.i.d.\ samples from $D$ with probability $0.99$ returns a classifier $h$ such that  $\pr_{(\x,y)\sim D}[h(x)\neq y]\leq \opt +\eps$.  
\end{corollary}
\begin{proof}
Define a geometric sequence $\gamma_i=1/2^{i}$ for each $i=0,1,2,\ldots$. 
Run the tester-learner with parameters $(\eps,\eta,\gamma_i)$ on fresh samples, and whenever it accepts obtain a hypothesis $h_i$ and also compute an independent validation estimate $\widehat{\err}_i$ of $\Pr[h_i(x)\neq y]$ using an additional fresh labeled sample. Let $i_1<i_2<\cdots$ be the indices at which the tester accepts. We stop at the first $k\ge 2$ such that $\widehat{\err}_{i_{k-1}}-\widehat{\err}_{i_k}\le \eps/2$ (i.e., the error curve ``flattens'' across two consecutive accepted $\gamma$'s), and output $h_{i_k}$.

This stopping rule is justified by the following two facts. First, for any $\gamma\le\gamma^*$ the completeness guarantee implies acceptance (with high probability) and $\Pr[h_\gamma(x)\neq y]\le \opt+\eps$, hence once we enter the regime $\gamma\le\gamma^*$ the errors of all subsequently accepted hypotheses differ by at most $\eps$.
However, whenever the tester accepts for  $\gamma’ $ and $ \gamma'',\gamma'>\gamma''$ with $\gamma’ >  2\gamma^*$
we must have $\Pr[h_{\gamma’}(x)\neq y]-\Pr[h_{\gamma''}(x)\neq y]\ge (1-2\eta)\gamma^*$.
This is because $h_{i_{k-1}}$ must have at least $\gamma’-\gamma^*$ more mass for one of the labels which results to at least  $(1-2\eta)(\gamma’-\gamma^*)$ error. Also $h_{\gamma''}$ should have error at most that of the classifier that matches $h^*$ and differs only in bias. So the flattening condition cannot hold before we reach $\gamma\ge\gamma^*$.    
\end{proof}

\end{document}